\documentclass[AEJ,finalmode]{AEA}

\usepackage{amsmath}
\usepackage{amssymb}
\usepackage{enumerate}
\usepackage{enumitem}
\usepackage{multirow}
\usepackage{graphicx}
\usepackage{subcaption}
\usepackage{ulem}

\newtheorem{theorem}{Theorem}[section]

\newtheorem{lemma}[theorem]{Lemma}
\newtheorem{assumption}{Assumption}

\DeclareMathOperator{\EX}{\mathbb{E}}
\DeclareMathOperator{\argmin}{\text{argmin}}
\newcommand\independent{\protect\mathpalette{\protect\independenT}{\perp}}
\def\independenT#1#2{\mathrel{\rlap{$#1#2$}\mkern2mu{#1#2}}} 

\usepackage{natbib}

\draftSpacing{1.5}

\begin{document}

\title{At What Level Should One Cluster Standard Errors in Paired and Small-Strata Experiments?}
\shortTitle{Clustering Standard Errors in Paired Experiments}
\author{Cl\'{e}ment de Chaisemartin and Jaime Ramirez-Cuellar\thanks{%
de Chaisemartin:  Economics Department, Sciences Po 28 rue des Saint-P\`{e}res 75005 Paris, France, clement.dechaisemartin@sciencespo.fr. Ramirez-Cuellar: Microsoft, Office of the Chief Economist, 99/4623, 14820 NE 36th St, Redmond, WA 98052, USA, jaimeramirez@microsoft.com. We are very grateful to Antoine Deeb, Jake Kohlhepp, David McKenzie, Heather Royer, Dick Startz, Doug Steigerwald, Gonzalo Vasquez-Bare, members of the econometrics and labor groups at UCSB, participants of the Advances in Field Experiments Conference 2019, California Econometrics Conference 2019, LAMES 2019, and LACAE 2019 for their helpful comments.}}
\date{\today}
\pubMonth{Month}
\pubYear{Year}
\pubVolume{Vol}
\pubIssue{Issue}

\begin{abstract}
In matched-pairs experiments in which one cluster per pair of clusters is
assigned to treatment, to estimate treatment effects, researchers often regress their
outcome on a treatment indicator and pair fixed effects, clustering standard errors
at the unit-of-randomization level. We show that even if the treatment has no
effect, a 5\%-level $t$-test based on this regression will wrongly conclude that the
treatment has an effect up to 16.5\% of the time. To fix this problem, researchers
should instead cluster standard errors at the pair level. Using simulations, we
show that similar results apply to clustered experiments with small strata.
\end{abstract}

\maketitle
In this paper, we show that a statistical test commonly used by researchers analyzing a certain type of randomized controlled trials (RCTs) has a much larger error rate than previously thought. We then suggest a simple fix.

The type of RCTs this paper applies to are clustered RCTs with large clusters, meaning that there are more than 10 observations per randomization unit, and where the treatment is assigned within pairs of units, or within small strata of less than 10 units. For instance, our paper would apply to an RCT where the researcher pairs some villages, randomly assigns one village per pair to the treatment, and estimates a regression at the villager rather than at the village level, with more than 10 villagers per village. Our paper would also apply to an RCT where the researcher groups villages into small strata of six villages and randomly assigns two villages per pair to the treatment. While such paired or small-strata RCTs with large clusters do not account for the majority of RCTs conducted in economics, it seems that they are still fairly common. We surveyed the universe of published editions of the \textit{American Economic Journal: Applied Economics} (\textit{AEJ Applied}) from 2014 to 2018, and found that this type of RCT represents 20\% of the RCTs published by the journal during that period.

In paired or small-strata RCTs with large clusters, researchers usually estimate the treatment effect by regressing their outcome on the treatment and pair or stratum fixed effects, ``clustering'' their standard errors at the unit-of-randomization level, namely, at the village level in our example.\footnote{Throughout this paper, clustered standard errors refer to the estimators proposed by \citet{liang1986longitudinal}, which are routinely implemented in standard statistical packages. Hereafter, we write ``clustered RCT'' to refer to the RCT design, and ``clustered standard errors'' to refer to the standard errors researchers use.} Then, they typically use the 5\%-level $t$-test based on this regression to assess if the treatment has an effect. As any statistical test, this $t$-test may lead them to commit a type 1 error: even if the treatment does not have an effect, this $t$-test may lead them to wrongly conclude that the treatment has an effect. When using a 5\%-level test, researchers hope that the probability that this would happen, the so-called error rate of the test, is less than 5\%. We show that the error rate of this $t$-test may in fact be much larger than the researcher's 5\% target.

We start by considering paired RCTs with large clusters. There, we show that even if the treatment has no effect, this 5\%-level $t$-test will wrongly conclude that the treatment has an effect up to 16.5\% of the time, an error rate more than three times larger than the targeted one. We then show that to achieve the desired 5\% error rate, researchers should instead cluster their standard errors at the pair level.  Finally, we revisit 371 regressions from the paired RCTs in our survey, and find that clustering at the pair rather than at the randomization-unit level diminishes the number of effects that are significant at the 5\% level by one third.

The intuition underlying our results is rather simple. In regression analysis, clustered standard errors are reliable when the regression's dependent and independent variables are uncorrelated across clusters \citep{cameron2015practitioner}. Therefore, in a paired RCT, clustering at the unit level relies on the assumption that the regression's independent variable, the treatment, is uncorrelated across all randomization units. However, the treatments of the two units in the same pair are perfectly negatively correlated: if unit A is treated, then unit B must be untreated, and vice versa. This is the reason why unit-clustered standard errors are unreliable, and the error rate of the $t$-test based on unit-clustered standard errors differs from that targeted by the researcher. On the other hand, clustering at the pair level only relies on the assumption that units' treatments are uncorrelated across pairs, which is true by design in paired experiments. This is the reason why pair-clustered standard errors are reliable, and the error rate of the $t$-test based on pair-clustered standard errors is equal to the error rate targeted by the researcher.

Our recommendations apply only to paired and clustered RCTs with large clusters. If the RCT is nonclustered, the 5\%-level $t$-test based on unit-clustered standard errors has a 5\% error rate, as targeted, after the degrees-of-freedom (DOF) adjustment automatically implemented in most statistical software. On the other hand, after this DOF adjustment the error rate of the 5\%-level $t$-test based on pair-clustered standard errors is lower than 5\%. Therefore, to achieve their targeted error rate in nonclustered paired RCTs, researchers should either use DOF-adjusted unit-clustered standard errors, or non-DOF-adjusted pair-clustered standard errors. In clustered RCTs, the same applies to regressions estimated at the unit-of-randomization level rather than at the observation level. Finally, in clustered RCTs with strictly fewer than 10 observations per randomization unit, our recommendation is to use pair-clustered standard errors, without the DOF adjustment. Figure \ref{fig:flowchart} below summarizes our recommendations for applied researchers, depending on whether their RCT is clustered and on the number of observations per randomization unit.

Then, we turn to small-strata RCTs with large clusters. Using simulations, we show that our results for paired designs extend to this case. Intuitively, there as well the treatments of units in the same stratum are negatively correlated, so this correlation should be accounted for. Therefore, our simulations show that in those designs too, the error rate of the 5\%-level $t$-test based on unit-clustered standard errors is larger than 5\%. The difference between the $t$-test’s actual and targeted error rates diminishes when the number of units per strata increases. Again, this makes intuitive sense: the larger the strata, the lower the correlation of the treatments of units belonging to the same stratum. For instance, with five units per strata, the error rate of the 5\% level $t$-test based on unit-clustered standard errors is equal to 7.9\%. With 10 units per strata, this error rate is equal to 6.2\%. With more than 10 units per strata, this error rate is lower than 6\%, and becomes very close to the targeted 5\% error rate. This is why, though we acknowledge it is somewhat arbitrary, we use a threshold of 10 units per strata to define an RCT with ``small’’ strata. Our simulations also show that the error rate of the 5\%-level $t$-test based on strata-clustered standard errors is very close to 5\%. Accordingly, in small-strata RCTs with large clusters, we recommend that researchers cluster their standard errors at the strata level. If the RCT has too few strata to cluster at that level, researchers could
use randomization inference, or the standard-error estimator in Section
9.5.1 of \citet{imbens2015causal}, provided each stratum has at least
two treated and two control units.

\subsection*{Related literature}

Our paper is related to several papers that predate ours.  \citet{imai2009essential} show that, when units all have the same number of observations, pair-clustered standard errors are reliable in clustered and paired RCTs.\footnote{\citet{imai2008variance} show similar results in nonclustered and paired RCTs.} With respect to their paper, we show that this result still holds when units have varying numbers of observations, thus justifying pair-level clustering under more realistic assumptions. Moreover, while they focus on finite-sample results, we present large-sample results for $t$-tests based on pair-clustered standard errors.

\citet{bruhn2009pursuit} use simulations to study, in nonclustered paired RCTs, the error rate of a $t$-test without clustering, which is equivalent to unit-clustering when the RCT is nonclustered. They show that this error rate is equal to the targeted error rate. This result may appear to conflict with ours, but this apparent discrepancy comes from the DOF adjustment embedded in most statistical software. In nonclustered RCTs, the regression with pair fixed effects has one fixed effect for each pair of observations. Accordingly, it has approximately half as many regressors as observations, so the DOF adjustment amounts to multiplying nonclustered standard errors by approximately $\sqrt{2}$, which, we show, makes them almost equivalent to the non-DOF-adjusted pair-clustered standard errors. This is why the error rate of DOF-adjusted nonclustered $t$-tests is equal to the targeted error rate in nonclustered paired RCTs.

\citet{abadie2017should} examine the appropriate level of clustering in regression analysis. They define a cluster as a group of units whose treatments are positively correlated. Their results apply to the case where the assignment is fully clustered (all units in the same cluster have the same treatment), and to the case where the assignment is probabilistically clustered (units in the same cluster have positively correlated treatments). Their Corollary 1 states that standard errors need to account for clustering whenever units' treatments are positively clustered within clusters. Our results are consistent with theirs. The main difference is that we consider a case in which units' treatments are negatively correlated within clusters (the pairs in our paper), which is not something they consider. However, our papers share a common theme: that standard errors need to account for clustering when treatments are correlated within some groups of units.

\citet{ATHEY201773} and \citet{bai2021inference} study the error rate of $t$-tests based on unit-clustered standard errors in paired experiments, when pair fixed effects are not included in the regression. They both show that without pair fixed effects, the test's error rate is lower than the targeted error rate. We instead show that when pair fixed effects are included in the regression, the error rate of this test becomes larger than the targeted error rate. In our survey of paired experiments, we find that including pair fixed effects in the regression is a much more common practice than not including those fixed effects.

The rest of this paper is organized as follows. Section 2 presents our survey of paired and small-strata RCTs in economics. Section 3 introduces our main theoretical results. Section 4 presents our simulation study. Section 5 presents our empirical application. Section 6 briefly discusses various extensions of our baseline results, which are fully developed in our Online Appendix. Section 7 concludes. Throughout the paper, a unit refers to a randomization unit (e.g.,\ a village), while an observation refers to the level at which the regression is estimated (e.g.,\ a villager).

\section{Survey of Paired and Small-Strata Experiments in Economics}

We searched the 2014--2018 issues of the \textit{AEJ Applied} for
clustered and paired RCTs, and for clustered and stratified RCTs with 10 or fewer units per strata.
50 field-RCT papers were published over that period. Three RCTs were clustered and relied
on a paired randomization for all of their analysis. One RCT was clustered and relied
on a paired randomization for part of its analysis. Seven RCTs were clustered and used
a stratified design with, on average, 10 or fewer units per strata. 10 of those 11 RCTs have, on average, more than 10 observations per randomization unit, while one stratified RCT has 5.8 observations per randomization unit.
Overall, 11 (22\%) of the 50 field RCTs published by the \textit{AEJ
Applied} over that period are clustered paired or small-strata RCTs, and 10 (20\%) also have large clusters.

To increase our sample of paired RCTs, we also searched the AEA's registry website
(\verb+https://www.socialscienceregistry.org+). We looked at all completed
projects, whose randomization method includes the word ``pair" and
that either have a working or a published paper. We conducted that
search on January 9th, 2019, and found four more clustered and paired RCTs. All of them have, on average, more than 10 observations per randomization unit. Combining our two searches,
we found 15 clustered paired or small-strata RCTs. The list is in Table \ref{tab:litrev}
in the Online Appendix.

We now give descriptive statistics on our sample of RCTs. Across the eight paired RCTs, the median number of pairs is 27, the
median number of observations per unit is 112, and units have on average more
than 10 observations in all RCTs.
To estimate the treatment effect, six articles include pair fixed
effects in all their regressions, one article includes pair fixed
effects in some but not all their regressions, and one article does not include pair fixed effects in any regression. All articles cluster standard errors at the
unit level.

Across the seven small-strata RCTs, the median number of units per
strata is 7, the median number of strata is 48, the median number
of observations per unit is 26, and units have more
than 10 observations on average in all but one RCT. To estimate the treatment effect,
six articles include stratum fixed effects in all their regressions,
and one article does not include stratum fixed effects in any regression.
All articles cluster standard errors at the
unit level.

In the following sections, we focus on paired RCTs. In Section \ref{se:ext_small_strata}
of the Online Appendix, we use simulations to show that the main results
we derive for paired RCTs extend to small-strata RCTs.


\section{Theoretical Results}\label{sec:mainresults}

\subsection{Setup}\label{subsec:setup}

We consider a population of $2P$ units. Unlike \citet{abadie2008}
and \citet{bai2021inference}, we do not assume that the units are
an independent and identically distributed (i.i.d.) sample drawn from a superpopulation. Instead, that population
is fixed, and its characteristics are not random. Our survey suggests that this modeling framework, similar to that in \citet{neyman1923applications} and \citet{abadie2020sampling}, is applicable to the majority of paired- and
small-strata RCTs in our survey. Units are drawn from a larger
population in only one of those fifteen RCTs.\footnote{This is in line with \citet{muralidharan2017experimentation}, who
show that the units are drawn from a larger population in only 31\% of the RCTs published in top-five journals between
2001 and 2016.} In all the other RCTs, the
sample is a convenience sample, consisting of volunteers to receive
the treatment, or of units located in areas where conducting the research
was easier. When the units are an i.i.d.\ sample drawn from a super-population, our results still hold, conditional on the sample.

The $2P$ units are matched into $P$ pairs. Pairs are created by
grouping together units with the closest value of some baseline variables
predicting the outcome. In our fixed-population framework, pairing
is not random, as it depends on fixed units' characteristics. The
pairs are indexed by $p\in\{1,\dots,P\}$, and the two units in pair
$p$ are indexed by $g\in\{1,2\}$. Unit $g$ in pair $p$ has $n_{gp}$
observations, so that pair $p$ has $n_{p}=n_{1p}+n_{2p}$ observations,
and the population has $n=\sum_{p=1}^{P}n_{p}$ observations. When
$n_{gp}>1$ for at least some units, the RCT is clustered; when $n_{gp}=1$
for all units, the RCT is nonclustered.

Treatment is assigned as follows. For all $p\in\{1,\dots,P\}$ and
$g\in\{1,2\}$, let $W_{gp}$ be an indicator variable equal to 1
if unit $g$ in pair $p$ is treated, and to 0 otherwise. We assume
that the treatments satisfy the following conditions. \begin{assumption}[Paired
assignment] \leavevmode \label{asm:1}
\begin{enumerate}
\item \label{asm:1_p1} For all $p$, $W_{1p}+W_{2p}=1$.
\item \label{asm:1_p2} $\mathbb{P}(W_{gp}=1)=\frac{1}{2}$ for all $g$
and $p$.
\item \label{asm:1_p3} $(W_{1p},W_{2p})_{p=1}^{P}$ is jointly independent
across $p$.
\end{enumerate}
\end{assumption} Point \ref{asm:1_p1} requires that in each pair,
one of the two units is treated. Point \ref{asm:1_p2} requires that
the two units have the same probability of being treated.
Point \ref{asm:1_p3} requires that the treatments be independent
across pairs. Assumption \ref{asm:1} is typically satisfied by design
in paired experiments.

Let $y_{igp}(1)$ and $y_{igp}(0)$ represent the potential outcomes
of observation $i$ in unit $g$ and pair $p$ with and without the
treatment, respectively. We follow the randomization-inference literature
\citep[see][]{abadie2020sampling} and assume that potential outcomes
are fixed.\footnote{In a previous version of this paper, we allowed potential outcomes
to be stochastic. Having stochastic potential outcomes does not change
our main results; see \citet{de2019level}.} The observed outcome is $Y_{igp}=y_{igp}(1)W_{gp}+y_{igp}(0)(1-W_{gp})$.
Our target parameter is the average treatment effect (ATE)
\begin{gather*}
\tau=\frac{1}{n}\sum_{p=1}^{P}\sum_{g=1}^{2}\sum_{i=1}^{n_{gp}}[y_{igp}(1)-y_{igp}(0)].
\end{gather*}

We consider two estimators of $\tau$. The first estimator $\widehat{\tau}$
is the OLS estimator from the regression of the observed outcome $Y_{igp}$
on a constant and $W_{gp}$:
\begin{gather}
Y_{igp}=\widehat{\alpha}+\widehat{\tau}W_{gp}+\epsilon_{igp}\qquad i=1,2,\dots,n_{gp};\ g=1,2;\ p=1,\dots,P.\label{eq:regnfe}
\end{gather}
The second estimator is the pair-fixed-effects estimator, $\widehat{\tau}_{fe}$,
obtained from the regression of the observed outcome $Y_{igp}$ on
$W_{gp}$ and a set of pair fixed effects $(\delta_{ig1},\dots,\delta_{igP})$:
\begin{gather}
Y_{igp}=\widehat{\tau}_{fe}W_{gp}+\sum_{p=1}^{P}\widehat{\gamma}_{p}\delta_{igp}+u_{igp},\qquad i=1,\dots,n_{gp};\ g=1,2;\ p=1,\dots,P.\label{eq:regfe}
\end{gather}

\subsection{Properties of Unit- and Pair-Clustered Variance Estimators}\label{subsec:UCVE-PCVE}

We study the variance estimators of $\widehat{\tau}$ and $\widehat{\tau}_{fe}$,
when the regression is clustered at either the pair level or the unit level.
The clustered-variance estimators we study are those proposed in \citet{liang1986longitudinal}.
Lemma \ref{le:CRVE_nfe} in Online Appendix \ref{sec:clu_var_estimators}
gives simple expressions of $\widehat{\mathbb{V}}_{pair}(\widehat{\tau})$
and $\widehat{\mathbb{V}}_{pair}(\widehat{\tau}_{fe})$, the pair-clustered
variance estimators (PCVEs) of $\widehat{\tau}$ and $\widehat{\tau}_{fe}$,
and of $\widehat{\mathbb{V}}_{unit}(\widehat{\tau})$ and $\widehat{\mathbb{V}}_{unit}(\widehat{\tau}_{fe})$,
the unit-clustered variance estimators (UCVEs) of $\widehat{\tau}$
and $\widehat{\tau}_{fe}$.

We now present our main results, that are derived under the following
assumption. \begin{assumption} \label{asm:bal_exp} There is a strictly
positive integer $N$ such that for all $p$, $n_{1p}=n_{2p}=N$.
\end{assumption} Assumption \ref{asm:bal_exp} requires that all
units have the same number of observations. Let
\[
\widehat{\tau}_{p}=\sum_{g}\left[W_{gp}\frac{1}{n_{gp}}\sum_{i}Y_{igp}-(1-W_{gp})\frac{1}{n_{gp}}\sum_{i}Y_{igp}\right]
\]
denote the difference between the average outcome of treated and untreated
observations in pair $p$. Under Assumption \ref{asm:bal_exp}, one
can show that
\[
\widehat{\tau}=\widehat{\tau}_{fe}=\sum_{p=1}^{P}\frac{\widehat{\tau}_{p}}{P},
\]
that both estimators are unbiased for the ATE, and that
\begin{gather}
\mathbb{V}(\widehat{\tau})=\mathbb{V}(\widehat{\tau}_{fe})=\frac{1}{P^{2}}\sum_{p=1}^{P}\mathbb{V}(\widehat{\tau}_{p}).\label{eq:var_hat_tau}
\end{gather}

Let $\tau_{p}\equiv\frac{1}{n_{p}}\sum_{g=1}^{2}\sum_{i=1}^{n_{gp}}[y_{igp}(1)-y_{igp}(0)]$
be the ATE in pair $p$. For all $d\in\{0,1\}$, let $\overline{y}_{gp}(d)\equiv\frac{1}{n_{gp}}\sum_{i}y_{igp}(d)$,
$\overline{y}_{p}(d)\equiv\frac{1}{2}\sum_{g}\overline{y}_{gp}(d)$,
and $\overline{y}(d)\equiv\sum_{p}\overline{y}_{p}(d)/P$, respectively,
denote the average outcome with treatment $d$ in pair $p$'s unit
$g$, in pair $p$, and in the entire population. \begin{lemma} \label{le:4.1}
\leavevmode
\begin{enumerate}
\item \label{le:4.1.1} If Assumptions \ref{asm:1} and \ref{asm:bal_exp}
hold, then $\widehat{\mathbb{V}}_{pair}(\widehat{\tau})=\widehat{\mathbb{V}}_{pair}(\widehat{\tau}_{fe})$,
and
\[
\EX\left[\frac{P}{P-1}\widehat{\mathbb{V}}_{pair}(\widehat{\tau})\right]=\mathbb{V}(\widehat{\tau})+\frac{1}{P(P-1)}\sum_{p=1}^{P}(\tau_{p}-\tau)^{2}\geq\mathbb{V}(\widehat{\tau}).
\]
\item \label{le:4.1.2} If Assumption \ref{asm:bal_exp} holds, then $\widehat{\mathbb{V}}_{pair}(\widehat{\tau})=2\widehat{\mathbb{V}}_{unit}(\widehat{\tau}_{fe})$.
\item \label{le:4.1.3} If Assumptions \ref{asm:1} and \ref{asm:bal_exp}
hold, then
\begin{align*}
\EX\left[\frac{P}{P-1}\left(\widehat{\mathbb{V}}_{unit}(\widehat{\tau})-\widehat{\mathbb{V}}_{pair}(\widehat{\tau})\right)\right] & =\frac{2}{P}\left(\frac{1}{P-1}\sum_{p}\left(\overline{y}_{p}(0)-\overline{y}(0)\right)\left(\overline{y}_{p}(1)-\overline{y}(1)\right)\right.\\
& \left.-\frac{1}{P}\sum_{p}\sum_{g}\frac{1}{2}\left(\overline{y}_{gp}(0)-\overline{y}_{p}(0)\right)\left(\overline{y}_{gp}(1)-\overline{y}_{p}(1)\right)\right).
\end{align*}
\end{enumerate}
\end{lemma}

\begin{proof} 
See Online Appendix \ref{app:main_proofs}.
\end{proof}

Point \ref{le:4.1.1} of Lemma \ref{le:4.1} shows that the PCVEs
without and with pair fixed effects are equal, and that after a DOF
correction, their expectation is at least as large as the variance
of $\widehat{\tau}$. If the treatment effect is heterogeneous across
pairs, $\frac{1}{P(P-1)}\sum_{p=1}^{P}(\tau_{p}-\tau)^{2}>0$ so the
inequality is strict: the PCVEs are upward-biased estimators for the
variance of $\widehat{\tau}$. If the treatment effect does not vary
across pairs, the inequality becomes an equality: the PCVEs are unbiased
for the variance of $\widehat{\tau}$.\footnote{The displayed equation in Point \ref{le:4.1.1} is almost identical
to Proposition 1 in \citet{imai2009essential}, up to a DOF
adjustment. We restate that result from their paper for completeness.} Building upon Point \ref{le:4.1.1} of Lemma \ref{le:4.1}, in the
Online Appendix we show that when the number of pairs grows, $(\widehat{\tau}-\tau)/\widehat{\mathbb{V}}_{pair}(\widehat{\tau})$
and $(\widehat{\tau}_{fe}-\tau)/\widehat{\mathbb{V}}_{pair}(\widehat{\tau}_{fe})$,
the $t$-statistics of the difference-in-means and fixed-effects estimators
using the PCVEs, both converge to a normal distribution with a mean equal
to 0 and a variance lower than 1 in general, but equal to 1 when the
treatment effect is homogenous across pairs (see Point \ref{th:asym_p2}
of Theorem \ref{th:asym}). Comparing those $t$-statistics to critical
values of a standard normal leads to a test with an error rate at most equal
to the researcher's target. For instance, if the average treatment effect $\tau$ is equal to zero,
by comparing $\left|\widehat{\tau}/\sqrt{\widehat{V}_{pair}(\widehat{\tau})}\right|$
to 1.96, one would wrongly conclude that $\tau \ne 0$ at most 5\% of the time, as desired.

On the other hand, Point \ref{le:4.1.2} of Lemma \ref{le:4.1} shows
that the UCVE with pair fixed effects is equal to a half of the PCVEs.
Combined with Point \ref{le:4.1.1} of Lemma \ref{le:4.1}, this implies
that the UCVE with pair fixed effects may severely underestimate the
variance of $\widehat{\tau}$: if the treatment effect is constant
across pairs, its expectation is equal to half of the variance of
$\widehat{\tau}$. Building upon Point \ref{le:4.1.2} of Lemma \ref{le:4.1},
in the Online Appendix we show that when the number of pairs grows, $(\widehat{\tau}_{fe}-\tau)/\widehat{\mathbb{V}}_{unit}(\widehat{\tau}_{fe})$,
the $t$-statistic of the fixed-effects estimator using the UCVE,
converges to a normal distribution with a mean equal to 0 and a variance
twice as large as that of the $t$-statistic using the PCVE (see Point
\ref{th:asym_p3} of Theorem \ref{th:asym}). Therefore, comparing
that $t$-statistic to critical values of a standard normal may yield
a test with a substantially larger error rate than the researcher's target. For instance, if the average treatment effect $\tau$ is equal to zero
and the treatment effect is homogenous across pairs, by comparing $\left|\widehat{\tau}_{fe}/\sqrt{\widehat{V}_{unit}(\widehat{\tau}_{fe})}\right|$
to 1.96, one would wrongly conclude that $\tau \ne 0$ 16.5\% of
the time, an error rate more than three times larger than the researcher's target.

With heterogeneous treatment effects across pairs, the
error rates of the $t$-tests using the PCVEs may be lower than the researcher's target, while the error rate of the $t$-test using
the UCVE with pair fixed effects may be equal to that target. However, in practice
we do not know if the treatment effect is constant or heterogeneous,
and it is common to require that a test have an error rate no larger than some target uniformly across
all possible data-generating processes. The $t$-tests using the PCVEs
satisfy that property, unlike the $t$-test using the UCVE with pair
fixed effects.

Finally, Point \ref{le:4.1.3} of Lemma \ref{le:4.1} shows that without
pair fixed effects, the expectation of the difference between the
UCVE and PCVE is proportional to the difference between the between-pair
and within-pair covariance of the two potential outcomes. In most
applications, both terms should be positive, as the two potential
outcomes should be positively correlated. One may also expect the
difference between those two terms to be positive, as units in the
same pair should have more similar potential outcomes than units in
different pairs. For instance, in the extreme case where units in
the same pair have equal potential outcomes, the second term is equal
to 0. Consequently, the expectation of the difference between the
UCVE and the PCVE should often be positive. Then, it follows from Point
\ref{le:4.1.1} of Lemma \ref{le:4.1} that the UCVE without pair
fixed effects is a more upward-biased estimator of the variance of
$\widehat{\tau}$ than the PCVEs, and that it remains upward-biased
even if the treatment effect is constant across pairs. Finally, building upon
Point \ref{le:4.1.3} of Lemma \ref{le:4.1}, in the Online Appendix
we show that the error rate of the $t$-statistic of the difference-in-means estimator
using the UCVE is lower and further away from the researcher's target than the error rate of the $t$-test making use
of the PCVEs (see Point \ref{th:asym_p4}
of Theorem \ref{th:asym}).

Intuitively, the UCVEs are biased because clustering at the unit level
does not account for the perfect negative correlation of the treatments
of the two units in the same pair. Cluster-robust standard errors
rely on the assumption that observations' outcomes and treatments
are uncorrelated across clusters \citep[see][]{cameron2015practitioner}.
This assumption is violated when one clusters at the unit level, but it
holds when one clusters at the pair level.

The direction of the bias of the UCVE depends on whether pair fixed
effects are included in the regression. When pair fixed effects are
not included in the regression, the UCVE will in general overestimate
the variance of $\widehat{\tau}$. This result may be relatively intuitive.
With positive correlations between observations, as is
often the case with time-series data, the variance of an estimator
is usually larger than what it would be without those correlations.
Then, one would expect that negative correlations would reduce an
estimator's variance. This is indeed what we find in Point \ref{le:4.1.3}
of Lemma \ref{le:4.1}: the UCVE, which estimates $\widehat{\tau}$'s
variance as if the treatments of two units in the same pair were not
negatively correlated, is larger than needed.

On the other hand, when pair fixed effects are included in the regression,
the UCVE may underestimate the variance of $\widehat{\tau}$. This
result is less intuitive. It comes from the fact that with pair fixed
effects in the regression, the sample residuals $u_{igp}$
are by construction uncorrelated with the pair fixed effects, which
implies that for every $p$, the sum of the residuals in pair $p$
is zero:
\[
\sum_{i,g}u_{igp}=0.
\]
Splitting the summation between $g=1$ and $g=2$, using the fact
that under Assumption \ref{asm:bal_exp} units 1 and 2 have the same
number of observations, and letting $\overline{u}_{g,p}$ denote the
average residuals of observations in unit $g$ of pair $p$, the previous
display implies that $\overline{u}_{1,p}=-\overline{u}_{2,p}$, which
in turn implies that $\left(\overline{u}_{1,p}\right)^{2}=\left(\overline{u}_{2,p}\right)^{2}$:
by construction, the squares of the average residuals are equal in
the treated and control units of each pair. Now, one can show that
with pair fixed effects, the UCVE is proportional to
\[
\frac{1}{(2P)^{2}}\sum_{p=1}^{P}\sum_{g=1}^{2}\left(\overline{u}_{g,p}\right)^{2},
\]
the sum, across all units, of their average squared residuals, divided
by the number of units squared. Accordingly, $\widehat{\mathbb{V}}_{unit}(\widehat{\tau}_{fe})$
treats $\left(\overline{u}_{1,p}\right)^{2}$ and $\left(\overline{u}_{2,p}\right)^{2}$
as if they were independent to estimate the variance of $\widehat{\tau}_{fe}$,
while they are equal to each other. Instead, the PCVE is proportional
to
\[
\frac{1}{P^{2}}\sum_{p=1}^{P}\left(\overline{u}_{1,p}\right)^{2}.
\]
$\widehat{\mathbb{V}}_{pair}(\widehat{\tau}_{fe})$ uses only one
squared-residual per pair to estimate the variance of $\widehat{\tau}_{fe}$.

As Section \ref{sec:6} below shows, our recommendation of using the
PCVE rather than the UCVE in clustered-paired RCTs and regressions with pair
fixed effects leads to a significant reduction in the number of effects that are
significant at the 5\% level in the published papers we revisit. One may then
wonder whether our results contradict those in \citet{bai2019optimality},
who shows that pairing is the optimal RCT design to maximize statistical
precision.\footnote{\citet{bai2019optimality} studies this question in nonclustered
RCTs. The optimal design in clustered RCTs has not been derived yet,
though we conjecture that the result in \citet{bai2019optimality}
carries through to clustered RCTs where units all have the same number
of observations.} The short answer is that our findings do not contradict his important
result. \citet{bai2019optimality} shows that the RCT design that
minimizes $\mathbb{V}(\widehat{\tau})$, and therefore the mean-squared
error of $\widehat{\tau}$, is a specific paired design. We do not
derive any new result on $\mathbb{V}(\widehat{\tau})$, so our results
have no bearing on his. Instead, our main result is to show that $\widehat{\mathbb{V}}_{unit}(\widehat{\tau}_{fe})$,
a commonly used variance estimator in paired experiments, can be severely
downward-biased. Instead, we recommend using another estimator,
$\widehat{\mathbb{V}}_{pair}(\widehat{\tau}_{fe})$, that is not downward
biased, and leads to a $t$-test with an error rate no larger than the researcher's target when the treatment does not have an effect. Using $\widehat{\mathbb{V}}_{pair}(\widehat{\tau}_{fe})$
instead of $\widehat{\mathbb{V}}_{unit}(\widehat{\tau}_{fe})$, researchers will conclude less often that the treatments they consider have an effect,
but comparing the power of those two
$t$-tests is not a fair comparison: the former test has an error rate no larger than the researcher's target when the treatment does not have an effect, unlike the latter one. Overall, while paired RCTs may not be as powerful
as the use of a spuriously-low variance estimator had led researchers
to believe, they remain a very powerful RCT design, the one that
leads to the lowest mean-squared error of $\widehat{\tau}$.

There is only one case where our results could imply that other designs
might be preferable to paired RCTs, though further research is
needed to validate or invalidate this conjecture. In our simulations,
we find that with fewer than 20 pairs, $t$-tests based on the PCVE
become less reliable: with fewer than 40 units, using the PCVE in paired
RCTs may lead to invalid inference. Thus, it may be preferable to
run a more coarsely stratified RCT with at least four units per strata,
and use, for example, the variance estimator proposed in Section 6.1
of \citet{ATHEY201773} for stratified RCTs. However, to our knowledge,
the validity of this alternative inference procedure has not been
assessed yet with a small number of units in the RCT. Note that in the (admittedly
small) sample of eight paired RCTs in our survey, one has five pairs
and another one has 14 pairs. All the other RCTs are close
to the 20-pair ``threshold'' (one has 19 pairs), or above it. Accordingly,
while paired RCTs with far fewer than 20 pairs are not a rarity, they
do not seem to be common either.

\subsection{Accounting for Degrees-of-Freedom Adjustments}\label{subsec:DOF}

The clustered-variance estimators we study are those proposed in \citet{liang1986longitudinal}.
Typically, statistical software report DOF-adjusted versions of those
estimators. For instance, in Stata the default adjustment is to multiply
the Liang and Zeger estimator by $[(n-1)/(n-k)]\times[G/(G-1)]$,
where $n$ is the sample size, $k$ the number of regressors, and
$G$ the number of clusters \citep[see][]{statacorp2017stata}. This
DOF adjustment is implemented when one uses the regress or areg command,
not when one uses the xtregress command \citep[see][]{cameron2015practitioner}.\footnote{Three of the four papers we revisit in Section \ref{sec:6}
use the regress or areg command; one uses the xtreg command.} In R, if the researcher uses the sandwich package, the default DOF
adjustment when declaring a cluster variable is the same as in Stata,
namely, $[(n-1)/(n-k)]\times[G/(G-1)]$. $G/(G-1)$ is close to $1$,
so the important term in the DOF adjustment is $(n-1)/(n-k)$.

In regressions without pair fixed effects, there are only two regressors
(the constant and the treatment), so $(n-1)/(n-k)=(n-1)/(n-2)$. This
quantity is close to 1, so the DOF adjustment leaves the UCVE and the PCVE
almost unchanged. Accordingly, in regressions without pair fixed effects,
the guidance we derived in the previous section also applies to the
DOF-adjusted UCVE and PCVE: the former estimator should not be used,
while the latter estimator can be used.

On the other hand, in regressions with pair fixed effects, the DOF
adjustment may affect the UCVE and the PCVE more substantially. When the
paired RCT is not clustered, the regression has $2P$ observations
and $P+1$ regressors, so $(n-1)/(n-K)=(2P-1)/(P-1)\approx2:$ the
DOF-adjusted UCVE is twice as large as the non-DOF-adjusted UCVE.
This fact and Point \ref{le:4.1.2} of Lemma \ref{le:4.1.3} imply
that in nonclustered RCTs, the DOF-adjusted UCVE with pair fixed
effects is almost equal to the non-DOF-adjusted PCVE with pair fixed
effects and has the same desirable properties. On the other hand,
the DOF-adjusted PCVE with pair fixed effects is now about twice as
large as the non-DOF-adjusted PCVE with pair fixed effects, so this
estimator is upward-biased even under constant treatment effect. Overall,
in nonclustered paired RCTs and regressions with pair fixed effects,
the guidance we derived in the previous section no longer applies
to the DOF-adjusted UCVE and PCVE: the former estimator can be used,
while the latter estimator should not be used.

When the paired RCT is clustered and the regression has pair fixed
effects, the regression has $2P\overline{n}_{u}$ observations and
$P+1$ regressors, where $\overline{n}_{u}$ denotes the average number
of observations across all units. Accordingly, $(n-1)/(n-K)=(2P\overline{n}_{u}-1)/(2P\overline{n}_{u}-(P+1))\approx2\overline{n}_{u}/(2\overline{n}_{u}-1).$
This quantity is decreasing in $\overline{n}_{u}$: the larger the average
number of observations across units, the smaller the DOF adjustment.
Simulations shown in Panel D of Table \ref{tab:size} show that with
$\overline{n}_{u}=5$, the error rate of a $t$-test based on the DOF-adjusted UCVE is
still considerably larger than the researcher's target: unlike what happens in nonclustered
experiments, the DOF-adjustment is not sufficient to ensure the error rate of this $t$-test is equal to the researcher's target. The same panel also shows that with $\overline{n}_{u}=5$,
the error rate of a $t$-test based on the DOF-adjusted PCVE is slightly below the researcher's target,
even under constant treatment effects. When $\overline{n}_{u}=10$,
simulations shown in Panel C of Table \ref{tab:size} show that the error rate of a
$t$-test based on the DOF-adjusted PCVE is now very close to the researcher's target.
Overall, in clustered paired RCTs with more than 10 observations per
unit, the guidance we derived in the previous section also applies
to the DOF-adjusted UCVE and PCVE: the former estimator should not
be used, while the latter estimator can be used. In clustered paired
RCTs with strictly fewer than 10 observations per unit, we recommend
using the PCVE without the DOF adjustment, which is
in line with a recommendation in \citet{cameron2015practitioner}
in a different context. In Stata, the xtregress command computes this
estimator.

\subsection{Should Pair Fixed Effects Be Included in the Regression?}\label{subsec:FE}

Though this paper is primarily concerned with the estimation of the
variance of treatment effect estimators, our recommendations crucially
depend on whether pair fixed effects are included in the regression. In this section, we discuss the pros and cons of including
such pair fixed effects. (Our paper does not bring any new
result to this longstanding discussion; we rely on earlier results,
sometimes specializing them to the case of paired RCTs.)

In nonclustered experiments, or in clustered experiments where in
each pair the two randomization units have the same number of observations
($n_{1p}=n_{2p}$ for all $p$), if no randomization unit attrits
from the sample, adding pair fixed effects to the regression leaves
the treatment coefficient unchanged: $\widehat{\tau}=\widehat{\tau}_{fe}$.
Because $\widehat{\tau}$ and $\widehat{\tau}_{fe}$ are equal, their
variances are also equal: adding pair fixed effects to the regression
does not lead to any precision gain in nonclustered paired
experiments or in clustered experiments where in each pair the two
randomization units have the same number of observations. When there
is attrition, $\widehat{\tau}$ and $\widehat{\tau}_{fe}$ will differ:
$\widehat{\tau}_{fe}$ will only leverage observations from pairs
where both randomization units are observed, while $\widehat{\tau}$
will also leverage observations from pairs where only one of the randomization
units is observed. \citet{king2007politically} argue in favor of
dropping pairs with one attriting unit, while \citet{bai2019optimality}
shows they can be kept. At any rate, even if one would prefer to drop
those pairs, one can simply do so before running the regression, rather
than running the regression with pair fixed effects in the full sample. Overall,
in nonclustered experiments, or in clustered experiments
where in each pair the two randomization units have the same number
of observations, there is no strong argument for or against adding
pair fixed effects to the regression.

In clustered experiments where there are pairs where the two randomization
units have different numbers of observations ($n_{1p}\ne n_{2p}$
for some $p$), adding pair fixed effects to the regression may change
the treatment coefficient: $\widehat{\tau}\ne\widehat{\tau}_{fe}$.
In such cases, one can show that $\widehat{\tau}$, the standard difference
in means estimator, converges toward our target parameter $\tau$,
the average treatment effect, when the number of pairs goes to infinity.
$\widehat{\tau}_{fe}$ on the other hand does not converge toward
$\tau$: one can show that it converges toward a parameter that has
sometimes been called a variance weighted average \citep[see][]{angrist2008mostly}
of the average treatment effect in each pair.\footnote{Specifically, let $\tau_{gp}$ denote the average treatment effect
in unit $g$ of pair $p$. One can show that $\widehat{\tau}_{fe}$
is consistent for a weighted average, across pairs, of $1/2(\tau_{1p}+\tau_{2p})$,
where pairs in which the numbers of observations of the two units are
close receive more weight than pairs where the numbers of observations
are different.} This parameter may differ from $\tau$ if the treatment effect varies
across pairs. Accordingly, unlike $\widehat{\tau}$, $\widehat{\tau}_{fe}$
may be biased for $\tau$, even asymptotically. On the other hand,
the variance of $\widehat{\tau}_{fe}$ is often lower than that of
$\widehat{\tau}$ \citep[see][]{imai2009essential}. Overall, if one
is primarily interested in consistently estimating $\tau$, pair fixed
effects should not be included in the regression. If one wants to
use the most precise estimator, pair fixed effects should be included
in the regression.

Figure \ref{fig:flowchart} summarizes our recommendations for practitioners,
regarding whether pair fixed effects should be included in the
regression, and regarding which variance estimator one should use.

\begin{figure*}[p]
\centering
\begin{subfigure}[b]{\textwidth}
\resizebox{\textwidth}{!}{\includegraphics[width=1\textwidth]{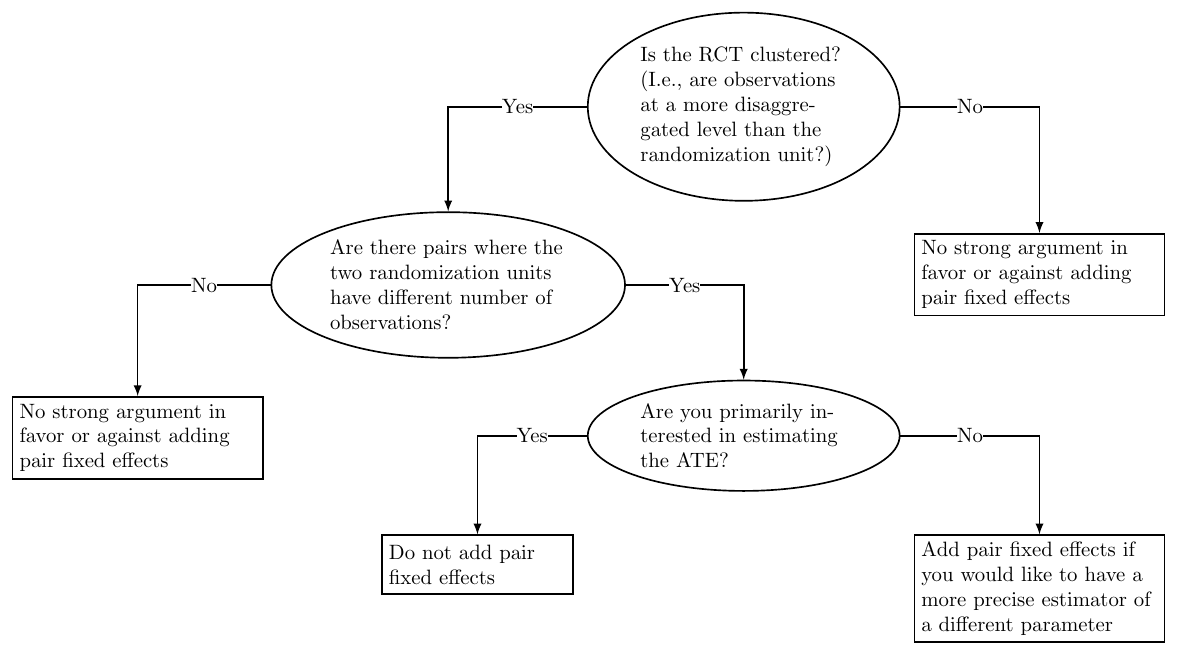}}
\centering \caption[Network2]{{Decision 1: Should the regression include pair fixed effects?
(See Section 3.4 for more details.)}}
\label{fig:flowchart_fes}
\end{subfigure}
\vskip\baselineskip
\begin{subfigure}[b]{\textwidth}
\centering
\includegraphics[width=1\textwidth]{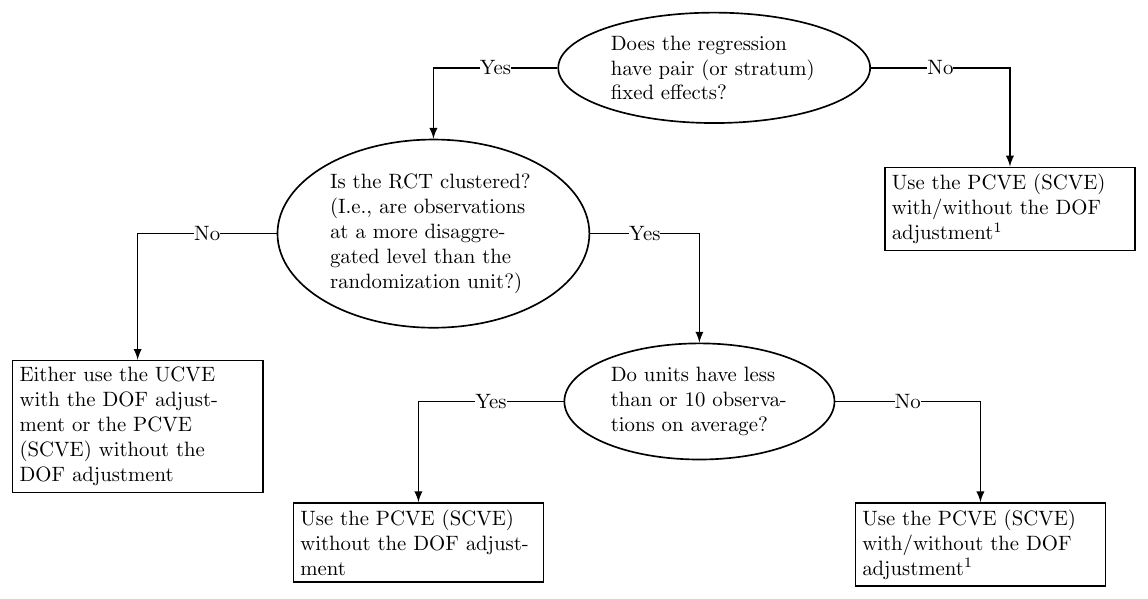}
\caption[]{{Decision 2: Which variance estimator should one use? (See
Section 3.3 for more details.)}}
\label{fig:flowchart_variance_estimators}
\end{subfigure}
\caption{Recommendations for Practitioners}
\label{fig:flowchart}
\begin{figurenotes}
UCVE = unit-clustered
variance estimators. PCVE = pair-clustered
variance estimators. DOF = degrees of freedom. \footnotemark[1]{In
these cases, the PCVE with and without the DOF adjustment are very similar.}
\end{figurenotes}
\end{figure*}

\section{Simulations Using Real Data}\label{se:simulations}

We perform Monte-Carlo simulations using a real data set. We use the
data from the microfinance RCT in \citet{crepon2015replication}.
The authors matched 162 Moroccan villages into 81 pairs, and in each
pair, they randomly assigned one village to a microfinance treatment.
They sampled households from each village and measured their outcomes
such as their credit access and income. The number of observations
varies substantially across units: the average number of villagers
per village is 34.1, with a standard deviation of 9.2, a minimum of
13 and a maximum of 58.

In the paper, the authors report the effect of the microfinance intervention
on 82 outcome variables.\footnote{Across the 82 outcomes, the median intracluster correlation coefficient
is $0.063$ at the village level, and $0.054$ at the pair level.} For each outcome, we construct potential outcomes assuming no treatment effect,
i.e., $y_{igpk}(0)=y_{igpk}(1)=Y_{igpk}$, where $Y_{igpk}$ is the
value of outcome $k$ for villager $i$ in village $g$ and pair
$p$. We then simulate 1000 treatment assignments $W_{k}^{j}=((W_{11,k}^{j},W_{21,k}^{j}),\dots,(W_{1P,k}^{j},W_{2P,k}^{j}))$,
assigning one of the two villages to treatment in each pair. Then,
we regress $Y_{igpk}$ on the simulated treatment. We estimate regressions
with and without pair fixed effects, clustering at the pair level
and at the village level. Thus, we obtain four $t$-statistics, and
four 5\% level $t$-tests. Importantly, those $t$-tests are based
on Stata's regress command, so they make use of DOF-adjusted variance
estimators. The estimated error rate of each $t$-test is the percentage
of times, across the 82,000 regressions (82 outcomes
$\times$ 1000 simulations), that the $t$-statistic is greater in absolute value than 1.96, meaning that the test leads the researcher to wrongly conclude that the treatment has an effect. Because the data is generated with a constant treatment effect of zero, these error rates should be equal to
5\% if the tests are valid.

Column (1) of Panel A of Table \ref{tab:size} shows the results using
the authors' actual data set, with 81 pairs and villages' actual number
of villagers. The error rates of the $t$-tests using pair-clustered variance
estimators (PCVEs) are close to 5\%, irrespective of whether pair fixed
effects are included in the regression. On the other hand, when the
unit-clustered variance estimator (UCVE) is used with pair fixed effects,
the error rate of the $t$-test is equal to 17.4\%, very close to the 16.5\%
error rate predicted by Point \ref{th:asym_p3} of Theorem \ref{th:asym}.
Finally, the error rate of the $t$-test with the UCVE and no pair fixed
effects is equal to 1.4\%, well below 5\%. Columns (2), (3), and (4) show that we obtain similar results
if we use a random sample of 40, 30, and 20 pairs. With fewer than
20 pairs, the PCVE becomes downward-biased. One may then have to use
randomization inference tests.

Panel B (resp. C) of Table \ref{tab:size} shows the error rates
of the four $t$-tests, in a data set where villages all have 20 (resp.
10) villagers. In each village, the villagers are a random sample
from the village's population, that does not vary across simulations.\footnote{Some villages have fewer than 20 villagers. For a village with, for
example, 13 villagers, we draw 7 villagers from the village's population
and add them to the original villagers.} Results are similar to Panel A.

Panel D shows the error rates of the four $t$-tests, in a data
set where villages all have 5 villagers. Again, the error rate of the $t$-test
with the PCVE and no pair fixed effects is close to 5\%. On the other
hand, the error rate of the $t$-test with the PCVE and pair fixed effects is now below 5\%.
As discussed in Section \ref{subsec:DOF}, this is due to the fact
that the DOF-adjustment is not negligible anymore with 5 villagers
per village. The error rate of the $t$-test with the UCVE and pair fixed effects is still much higher than 5\%, but less so than in Panel A.
Finally, the error rate of the $t$-test
with the UCVE and no pair fixed effects is still well below 5\%, though less so than in Panel A. These simulations justify the guidance
above: in clustered RCTs with strictly fewer than 10 observations
per unit, one should either use the PCVE without pair fixed effects,
with or without the DOF adjustment, or the PCVE with pair fixed effects
without the DOF adjustment.

Finally, Panel E of Table \ref{tab:size} shows the error rates
of the four $t$-tests, in a data set where  a quarter of the villages have five
villagers, a quarter have 10 villagers, a quarter have 20 villagers, and a quarter have
their actual number of villagers. In Columns (1) and (2), results
are fairly similar to those in Panel A. In Columns (3) and (4), the
error rates of the $t$-tests using the PCVEs are larger than 5\% (though much
less so than the $t$-test using the UCVE with pair fixed effects). This
is related to the results in \citet{carter2017asymptotic}, who find
that when clusters have very heterogeneous sizes, one needs a larger
number of clusters to ensure that asymptotic distributions yield accurate
approximations of the finite-sample distribution of cluster-robust
$t$-statistics. Note that this phenomenon is absent in Panel A, while
village sizes are already fairly heterogeneous in those simulations.
In applications where units have very heterogeneous numbers of observations,
researchers may need to perform their own simulations to assess whether
$t$-tests using the PCVEs can be used.

\begin{table}[htbp]
\caption{Error Rates of $T$-tests, in Simulations Based on \cite{crepon2015estimating}}
\label{tab:size}
    \begin{tabular}{lccccc}
        \hline \hline
         \multirow{3}{2cm}{Clustering level} & \multirow{3}{2cm}{Pair Fixed Effects} & \multicolumn{4}{c}{5\% level $t$-test error rate}  \\
         & & \multirow{2}{1.75cm}{With 81 pairs (1)} & \multirow{2}{1.75cm}{With 40 pairs (2)} & \multirow{2}{1.75cm}{With 30 pairs (3)} & \multirow{2}{1.75cm}{With 20 pairs (4)} \\
         \\
        \hline
        \multicolumn{6}{l}{\textit{Panel A: Actual village sizes}} \\
        \hspace{0.2cm} pair&Yes&0.0515&0.0519&0.0527&0.0559\\
		\hspace{0.2cm} pair&No&0.0527&0.0540&0.0541&0.0572\\
		\hspace{0.2cm} unit&Yes&0.1704&0.1742&0.1802&0.1840\\
		\hspace{0.2cm} unit&No&0.0137&0.0171&0.0178&0.0192\\
        \multicolumn{6}{l}{\textit{Panel B: All villages have 20 villagers}} \\
        \hspace{0.2cm} pair&Yes&0.0464&0.0508&0.0510&0.0537\\
		\hspace{0.2cm} pair&No&0.0491&0.0540&0.0541&0.0562\\
		\hspace{0.2cm} unit&Yes&0.1663&0.1692&0.1741&0.1737\\
		\hspace{0.2cm} unit&No&0.0161&0.0210&0.0190&0.0259\\
        \multicolumn{6}{l}{\textit{Panel C: All villages have 10 villagers}} \\
        \hspace{0.2cm} pair&Yes&0.0432&0.0439&0.0464&0.0454\\
		\hspace{0.2cm} pair&No&0.0489&0.0503&0.0538&0.0498\\
		\hspace{0.2cm} unit&Yes&0.1575&0.1581&0.1622&0.1686\\
		\hspace{0.2cm} unit&No&0.0214&0.0238&0.0249&0.0265\\
        \multicolumn{6}{l}{\textit{Panel D: All villages have 5 villagers}} \\
        \hspace{0.2cm} pair&Yes&0.0360&0.0384&0.0371&0.0363\\
		\hspace{0.2cm} pair&No&0.0480&0.0514&0.0502&0.0477\\
		\hspace{0.2cm} unit&Yes&0.1393&0.1372&0.1367&0.1421\\
		\hspace{0.2cm} unit&No&0.0287&0.0281&0.0313&0.0284\\
    \multicolumn{6}{l}{\textit{Panel E: Heterogeneous village sizes}} \\
    \hspace{0.2cm} pair&Yes&0.0517&0.0543&0.0597&0.0673\\
	\hspace{0.2cm} pair&No&0.0567&0.0571&0.0623&0.0692\\
	\hspace{0.2cm} unit&Yes&0.1701&0.1716&0.1775&0.1797\\
	\hspace{0.2cm} unit&No&0.0180&0.0231&0.0228&0.0271\\
        \hline
        \end{tabular}

\begin{tablenotes}
The table reports the error rates of four 5\% level $t$-tests in \cite{crepon2015estimating}. For each of the 82 outcomes in the paper,  we randomly drew 1000 simulated treatment assignments, following the paired assignment used by the authors, and regressed the outcome on the simulated treatment. The four $t$-tests are computed, respectively, without and with pair fixed effects in the regression, and clustering standard errors at the village or at the pair level. All $t$-tests are based on Stata's regress command, so they make use of DOF-adjusted variance estimators. The error rate of each test is the percent of times, across the 82,000 regressions (82 outcomes $\times$ 1000 replications), that the test leads the researcher to wrongly conclude that the treatment has an effect. Column (1) (resp.\ (2), (3), (4)) shows the results using the original sample of 81 pairs (resp.\ a fixed sample of 40, 30, 20 randomly selected pairs). In Panel A, villages all have their actual number of villagers. In Panel B (resp. C, D), each village has 20 (resp. 10, 5) villagers, that are a fixed random sample from the village's population. In Panel E, 1/4 of villages have 5 villagers, 1/4 have 10 villagers, 1/4 have 20 villagers, and 1/4 have their actual number of villagers.
\end{tablenotes}
\end{table}


\section{Application}\label{sec:6}

In this section, we revisit the paired RCTs in our survey. The data
used in four of those papers is publicly available \citep{beuermann2015replication,bruhn2016replication,crepon2015replication,glewwe2016replication}.
Those four papers used a clustered RCT, and all have more than 10
observations per randomization unit (across the four papers, the lowest
average number of observations per randomization unit is 21.5). The
authors estimated the effect of the treatment in 294 regressions,
clustering at the unit level. In Panel A of Table \ref{tab:replic},
we re-estimate those regressions, clustering at the pair level, and
including the same controls as the authors. In the 240 regressions
with pair fixed effects, the average of the unit-clustered variance
estimator (UCVE) divided by the pair-clustered variance estimator (PCVE) is equal
to 0.548. The UCVE divided by the PCVE is not always exactly equal to 1/2, because Assumption
\ref{asm:bal_exp} is not always satisfied, but these fractions all are quite
close to 1/2, as predicted by Lemma \ref{le:ext_fe} in the Online Appendix. The authors
originally found that the treatment has a 5\%-level significant effect
in 110 regressions. Using the PCVE, we find significant effects in
just 74 regressions. In the 54 regressions without pair fixed effects,
the UCVE is on average 1.18 times larger than the PCVE. The authors
originally found 31 significant effects, we find 36 significant effects
using the PCVE.

The data
used in the remaining four papers is not publicly available. Three of those papers estimated 131 regressions with
pair fixed effects, clustering standard errors at the unit level.\footnote{Across these papers, the lowest number of observations
per randomization unit is 99.0.} For those regressions, we multiply the UCVE by the average value of
of the PCVE divided by the UCVE found in Panel A of Table \ref{tab:replic} to
predict the value of the PCVE. Panel B of Table \ref{tab:replic}
shows that while the authors originally found a 5\%-level significant
effect in 51 regressions, we find significant effects in just 34 regressions.
The fourth paper estimated regressions only without pair fixed effects.
Because without fixed effects, the value of the PCVE divided by the UCVE can vary significantly across regressions,
we do not try to predict the PCVEs of that paper.

\begin{table}[htbp]
\caption{Using Unit- or Pair-Level Clustered Variance Estimators in Paired RCTs}
    \label{tab:replic}

    \begin{tabular}{lcccc}
        \hline
        & \multirow{5}{8em}{Unit-level divided by pair-level clustered variance estimators} &  \multirow{5}{6em}{Number of 5\%-level significant effects with UCVE} & \multirow{5}{6em}{Number of 5\%-level significant effects with PCVE} & \multirow{5}{5em}{Number of Regressions} \\
        \\
        \\
        \\
        \\ \hline
        \multicolumn{5}{l}{\textit{Panel A: Articles with publicly available data}} \\
        \hspace{0.2cm} \multirow{2}{2.cm}{with pair fixed effects} & \multirow{2}{*}{0.548} & \multirow{2}{*}{110} & \multirow{2}{*}{74} & \multirow{2}{*}{240} \\
        \\ 
        \\
        \hspace{0.2cm} \multirow{3}{2cm}{without pair fixed effects} & \multirow{3}{*}{1.184} & \multirow{3}{*}{31} & \multirow{3}{*}{36} & \multirow{3}{*}{54} \\
        \\
        \\ 
        \\
        \multicolumn{5}{l}{\textit{Panel B: Articles without publicly available data}} \\
        \hspace{0.2cm} \multirow{2}{2cm}{with pair fixed effects} &  & \multirow{2}{*}{51} & \multirow{2}{*}{34} & \multirow{2}{*}{131} \\
        \\
        \hline
    \end{tabular}

\begin{tablenotes}
The table shows the effect of using pair-clustered variance estimators (PCVE) rather than unit-level clustered variance estimators (UCVE) in seven of the paired RCTs we found in our survey. In Panel A, we consider four papers whose data is available online, and re-estimate their regressions clustering standard errors at the pair level. Column 1 shows the ratio of the unit- and pair-level clustered variance estimators, separately for regressions without and with pair fixed effects. Column 2 (resp.\ 3) shows the number of 5\%-level significant effects using unit- (resp.\ pair-) clustered standard errors. In Panel B, we consider three other papers whose data is not available online, and use the average ratio of the unit- and pair- clustered variance estimators found in Panel A to predict the value of the pair-clustered estimator in the regressions with pair fixed effects estimated by those papers. Column 2 (resp.\ 3) shows the number of 5\%-level significant effects using unit- (resp.\ predicted pair-) clustered standard errors.
\end{tablenotes}
\end{table}

\section{Extensions}

In our Online Appendix, we consider various extensions. In Appendix \ref{se:ext_small_strata}, we present
simulations showing that our results for paired RCTs extend to stratified
RCTs with few units per strata.

Assumption \ref{asm:bal_exp}, which requires that all units have
the same number of observations, allows us to derive the stark results
in Lemma \ref{le:4.1}. Under Assumption \ref{asm:bal_exp}, $\widehat{\tau}$
and $\widehat{\tau}_{fe}$ are equal, but the UCVE drastically changes
when one adds pair fixed effects to the regression. This is obviously
undesirable: the two estimators are equal, their variances are equal,
so their variance estimators should not be drastically different.
In practice, however, Assumption \ref{asm:bal_exp} often fails. In
that case, we show in Section \ref{sec:ext_assm2} of the Online Appendix
that our main conclusions still hold. Without that assumption, the
PCVEs remain upward-biased in general and unbiased if the treatment
effect is homogeneous across pairs. On the other hand, the UCVE with
pair fixed effects may still be downward-biased. Specifically, Point
\ref{le:4.1.2} of Lemma \ref{le:4.1} still holds if the number of
observations per unit varies across pairs, as long as the two units
in a pair have the same number of observations. If the number of observations
per unit varies within pairs, Point \ref{le:4.1.2} of Lemma \ref{le:4.1}
still approximately holds, unless units in the same pair have very
heterogeneous numbers of observations. Indeed, Lemma \ref{le:ext_fe}
shows that $\widehat{\mathbb{V}}_{unit}(\widehat{\tau}_{fe})/\widehat{\mathbb{V}}_{pair}(\widehat{\tau}_{fe})$
is included between $1/2$ and $5/9$ as long as $n_{1p}/n_{2p}$
is included between 0.5 and 2 for all $p$, meaning that in each pair
the first unit has between half and twice as many observations as
the second one.

In Appendix \ref{sec:pop}, we study two alternatives to the PCVE.
With heterogeneous treatment effects across pairs, the PCVE overestimates
the variance of the treatment effect estimator. To increase power,
one may want to use an unbiased estimator of that variance. We study
two alternatives, the pair-of-pairs estimator proposed by \citet{abadie2008},
and a variance estimator proposed by \citet{bai2021inference}.\footnote{Other alternatives have been proposed. For instance, \citet{fogarty2018mitigating}
proposes to use covariates that predict the treatment effect heterogeneity
across pairs to form a less-upward-biased estimator than the pair-clustered
one. We do not consider this estimator, merely because it lends itself
less easily to the automatic replication exercise we conduct: in each
application, one has to determine the relevant covariates to include,
based on context-specific knowledge.} Both are unbiased, or at least consistent, when units are an i.i.d.\ sample
drawn from a superpopulation. In the set-up we consider, where units
are a convenience sample, we show that those two estimators
are upward-biased, like the PCVE. They are less upward-biased than
the PCVE when the treatment effect is less heterogeneous within than
between pairs of pairs, and more upward-biased otherwise. We compute
the three estimators in the regressions in our survey, and find that
they are on average equivalent, so it does not seem one can expect
large power gains from using those two alternative estimators. Moreover,
simulations based on the data from \citet{crepon2015estimating} show
that $t$-tests using those two estimators have a drawback relative
to the $t$-test using the PCVE. The corresponding $t$-statistics
are approximately normally distributed only if the sample has more
than a couple hundred pairs. On the other hand, the $t$-test based
on the PCVE is approximately normally distributed with as few as 20
pairs.


\section{Conclusion}

In paired or small-strata RCTs with large clusters, researchers usually estimate the treatment effect by regressing their outcome on the treatment and pair or stratum fixed effects, ``clustering'' their standard errors at the unit-of-randomization level. Then, they typically use the 5\%-level $t$-test based on this regression to determine if the treatment has an effect or not. As any statistical test, this $t$-test may lead them to commit a type 1 error. Specifically, it may lead them to wrongly conclude that the treatment has an effect, while the truth is that the treatment does not have an effect. But when using a 5\%-level test, their hope is that the probability that this would happen, the so-called error rate of the test, is no larger than 5\%. We show that unfortunately, the error rate of this $t$-test may be much larger than the researcher's 5\% target. We then show that to achieve their desired error rate, researchers should cluster their standard errors at the pair or at the strata level, rather than at the unit-of-randomization level. Clustering at the pair rather than at the unit level in a sample of 371 regressions from published paired RCTs reduces the number of significant effects by 1/3.


\bibliographystyle{aea}
\bibliography{biblio}

\appendix
\newpage
\setcounter{page}{1}
\pagestyle{plain}

\huge
\begin{center}

      {\titlesize\sffamily\textbf{At What Level Should One Cluster Standard Errors in Paired and Small-Strata Experiments?} \par}%
      \vskip 16pt%
      {\authorsize\textit{By \ }\textsc{Cl\'{e}ment de Chaisemartin and Jaime Ramirez-Cuellar}\par}
      \vskip 16pt%
\textbf{Online Appendix}
\end{center}

\normalsize


\section{A. Proof of Lemma \ref{le:4.1}}
\label{app:main_proofs}

We first introduce some notation.
Let $T_p = n_{1p} W_{1p}+n_{2p} W_{2p}$ and $C_p = n_{1p} (1-W_{1p})+n_{2p} (1-W_{2p})$ be the number of treated and untreated observations in pair $p$. Let $T = \sum_{p=1}^P T_p$ and $C = \sum_{p=1}^P C_p$ be the total number of treated and untreated observations. Let $SET_p = \sum_{g = 1}^2\sum_{i=1}^{n_{gp}} W_{gp} {\epsilon}_{igp}$ and $SEU_p = \sum_{g = 1}^2\sum_{i=1}^{n_{gp}} (1-W_{gp}) {\epsilon}_{igp}$ respectively be the sum of the residuals $\epsilon_{igp}$ for the treated and untreated observations in pair $p$.

$\widehat{\tau}$ is the well-known difference-in-means estimator:
    \begin{gather*}
        \widehat{\tau} = \sum_{p=1}^P \sum_{g=1}^2 \sum_{i=1}^{n_{gp}} \frac{Y_{igp}W_{gp}}{T}
        - \sum_{p=1}^P \sum_{g=1}^2 \sum_{i=1}^{n_{gp}} \frac{Y_{igp}(1-W_{gp})}{C}.
    \end{gather*}
Remember that $\widehat{\tau}_p=\sum_{g=1}^2 \left[W_{gp}\sum_{i=1}^{n_{gp}} \frac{Y_{igp}}{n_{gp}} - (1-W_{gp})\sum_{i=1}^{n_{gp}} \frac{Y_{igp}}{n_{gp}}\right]$ is the difference between the average outcome of treated and untreated observations in pair $p$.
It follows from, e.g., Equation (3.3.7) in \cite{angrist2008mostly} and a few lines of algebra that
        \begin{gather*}
            \widehat{\tau}_{fe} = \sum_{p=1}^{P} \omega_p \widehat{\tau}_p, \qquad \text{where} \quad
            \omega_p= \frac{\left(n_{1p}^{-1}+n_{2p}^{-1}\right)^{-1}}{\sum_{p'=1}^P \left(n_{1p'}^{-1}+n_{2p'}^{-1}\right)^{-1}}.
        \end{gather*}


\subsubsection*{Point \ref{le:4.1.1}}~\\

\textit{Proof of $\widehat{\mathbb{V}}_{pair}(\widehat{\tau}) = \widehat{\mathbb{V}}_{pair}(\widehat{\tau}_{fe})$}\\

It follows from Equations \eqref{eq:regnfe} and \eqref{eq:regfe} that
    \begin{gather*}
        \widehat{\alpha} + \widehat{\tau}W_{gp} + \epsilon_{igp} = \widehat{\tau}_{fe}W_{gp} + \sum_{p=1}^P \widehat{\gamma}_p\delta_{igp} + u_{igp}.
    \end{gather*}
Rearranging and using the fact that under Assumption \ref{asm:bal_exp} $\widehat{\tau} = \widehat{\tau}_{fe}$, one obtains that for every $p$:
\begin{gather}
            \epsilon_{igp} = \widehat{\gamma}_p - \widehat{\alpha} + u_{igp}.
            \label{res_nfe}
\end{gather}

Then,
    \begin{align}
            \widehat{\mathbb{V}}_{pair}(\widehat{\tau})
            & = \frac{1}{T^2}\sum_{p=1}^P \left(SET_p - SEU_p \right)^2
            \notag
            \\
            & = \frac{1}{T^2} \sum_p \left[\sum_g \sum_i (2W_{gp}-1)\epsilon_{igp}\right]^2
            \notag
            \\
            & = \frac{1}{T^2} \sum_p \left[\sum_g \sum_i (2W_{gp}-1)(\widehat{\gamma}_p - \widehat{\alpha} + u_{igp})\right]^2
            \notag
            \\
            & = \frac{1}{T^2} \sum_p \left[\sum_g \sum_i (2W_{gp}-1)u_{igp}
                    + (\widehat{\gamma}_p - \widehat{\alpha}) \sum_g \sum_i (2W_{gp}-1) \right]^2
            \notag
            \\
            & = \frac{4}{T^2} \sum_p \left(\sum_g \sum_i W_{gp}u_{igp} \right)^2.
            \label{eq:vhat_pair1}
    \end{align}
    The first equality follows from Point \ref{le:CRVE_nfe_p1} of Lemma \ref{le:CRVE_nfe} and Assumption \ref{asm:bal_exp}. The third equality follows from Equation \eqref{res_nfe}. The fifth follows from the following two facts. First, $\sum_g \sum_i (2W_{gp}-1)u_{igp} = 2\sum_g \sum_i W_{gp}u_{igp} - \sum_g \sum_i u_{igp}=2\sum_g \sum_i W_{gp}u_{igp}$, since $\sum_g \sum_i u_{igp} = 0$ by definition of $u_{igp}$. Second, $ (\widehat{\gamma}_p - \widehat{\alpha}) \sum_g \sum_i (2W_{gp}-1) = (\widehat{\gamma}_p - \widehat{\alpha}) \left[\sum_g \sum_i W_{gp} - \sum_g \sum_i (1-W_{gp})\right] = (\widehat{\gamma}_p - \widehat{\alpha})[T_p - C_p] = 0$, where the last equality comes from the fact that $n_{1p} = n_{2p} $ by Assumption \ref{asm:bal_exp}.

Similarly,
    \begin{align}
            \widehat{\mathbb{V}}_{pair}(\widehat{\tau}_{fe})
            & = \frac{4}{T^2}\sum_{p=1}^P SET_{p,fe}^2 = \frac{4}{T^2}\sum_{p=1}^P \left(\sum_g \sum_i W_{gp}u_{igp}\right)^2,
            \label{eq:PCRVE_balgroup}
    \end{align}
    where the first equality follows from Equation \eqref{eq:fe_num} in the proof of Lemma \ref{le:CRVE_nfe} and Assumption \ref{asm:bal_exp}.
    Combining Equations \eqref{eq:vhat_pair1} and \eqref{eq:PCRVE_balgroup} yields $\widehat{\mathbb{V}}_{pair}(\widehat{\tau}) = \widehat{\mathbb{V}}_{pair}(\widehat{\tau}_{fe})$.

\medskip
\textit{Proof of $\EX\left[\frac{P}{P-1}\widehat{\mathbb{V}}_{pair}(\widehat{\tau})\right] =\mathbb{V}(\widehat{\tau})+\frac{1}{P(P-1)}\sum_{p=1}^P(\tau_p-\tau)^2$}\\

    Under Assumption \ref{asm:bal_exp}, $T=C=n/2$, so
    \begin{align}
        \widehat{\mathbb{V}}_{pair}(\widehat{\tau}) & = \sum_{p=1}^P \left(  \frac{SET_p}{T} - \frac{SEU_p}{C} \right)^2
        \notag
        \\
        & = \frac{4}{n^2}\sum_{p=1}^P \left(SET_p - SEU_p \right)^2
        \notag
        \\
        & = \frac{4}{n^2}\sum_{p=1}^P \left(\sum_g \sum_i (W_{gp}\epsilon_{igp} - (1-W_{gp})\epsilon_{igp}) \right)^2
        \notag
        \\
        & = \frac{4}{n^2}\sum_{p=1}^P \left(\sum_g \sum_i (2W_{gp}-1)\epsilon_{igp} \right)^2
        \notag
        \\
        & = \frac{4}{n^2}\sum_{p=1}^P \left(\sum_g (2W_{gp}-1) \sum_i (Y_{igp}-\widehat{\tau}W_{gp}-\widehat{\alpha}) \right)^2
        \notag
        \\
        & = \frac{4}{n^2}\sum_{p=1}^P \left(\sum_g (2W_{gp}-1) \left(\sum_i Y_{igp}-\widehat{\tau}W_{gp}\frac{n_{p}}{2}-\widehat{\alpha}\frac{n_{p}}{2}\right) \right)^2
        \notag
        \\
        & = \frac{4}{n^2}\sum_{p=1}^P \left(\sum_g (2W_{gp}-1)\sum_i Y_{igp}-\widehat{\tau}\frac{n_{p}}{2}\sum_g (2W_{gp}-W_{gp})-\widehat{\alpha}\frac{n_{p}}{2}\sum_g (2W_{gp}-1) \right)^2
        \notag
        \\
        & = \frac{4}{n^2}\sum_{p=1}^P \left(\sum_g (2W_{gp}-1)\sum_i Y_{igp}-\widehat{\tau}\frac{n_{p}}{2}\sum_g W_{gp}\right) ^2
        \notag
        \\
        & = \frac{4}{n^2}\sum_{p=1}^P \left(\sum_g (2W_{gp}-1)\sum_i Y_{igp}-\widehat{\tau}\frac{n_{p}}{2}\right) ^2
        \notag \\
        & = \frac{4}{n^2}\sum_{p=1}^P \left(\widehat{\tau}_p\frac{n_p}{2}-\widehat{\tau}\frac{n_{p}}{2}\right)^2
        \notag
        \\
        & = \frac{1}{P^2} \sum_{p=1}^P \left(\widehat{\tau}_p-\widehat{\tau}\right)^2. \label{eq:PCVE_simple}
    \end{align}
    The third equality comes from the definition of $SET_p$ and $SEU_p$. The fifth equality follows from the Equation \eqref{eq:regnfe}. The sixth equality follows from $n_{1p}=n_{2p}=n_p/2$, which is a consequence of Assumption \ref{asm:bal_exp}. The eighth equality comes from the fact that $\sum_g (2W_{gp}-1) = 0$, which follows from Point \ref{asm:1_p1} of Assumption \ref{asm:1}. The ninth equality follows from Point \ref{asm:1_p1} of Assumption \ref{asm:1}. The tenth equality follows from  $\sum_g (2W_{gp}-1)\sum_i Y_{igp} = \sum_g W_{gp}\sum_i Y_{igp} - \sum_g (1-W_{gp})\sum_i Y_{igp} = n_p \widehat{\tau}_p/2$. The eleventh equality follows from Assumption \ref{asm:bal_exp}.

    Now, consider Equation \eqref{eq:PCVE_simple}. Adding and subtracting $\tau$ and $\tau_p=\EX[\widehat{\tau}_p]$,
        \begin{align*}
        \widehat{\mathbb{V}}_{pair}(\widehat{\tau})    & =  \frac{1}{P^2} \sum_{p=1}^P\left((\widehat{\tau}_p-\tau_p)-(\widehat{\tau}-\tau) +(\tau_p-\tau)\right)^2
        \\
        & =  \frac{1}{P^2} \sum_{p=1}^P\left[(\widehat{\tau}_p-\tau_p)^2+(\widehat{\tau}-\tau)^2 +(\tau_p-\tau)^2 -2(\widehat{\tau}_p-\tau_p)(\widehat{\tau}-\tau)\right.
        \\
        & \hphantom{space needed to align the equation}\left.
        +2(\widehat{\tau}_p-\tau_p)(\tau_p-\tau)- 2(\widehat{\tau}-\tau)(\tau_p-\tau)\right].
        \intertext{Taking the expected value, and given that $\EX[\widehat{\tau}]=\tau$ and $\EX[\widehat{\tau}_p]=\tau_p$,}
        \EX[\widehat{\mathbb{V}}_{pair}(\widehat{\tau})] & =  \frac{1}{P^2} \sum_{p=1}^P\left[\mathbb{V}(\widehat{\tau}_p)+\mathbb{V}(\widehat{\tau}) +(\tau_p-\tau)^2  -2\text{Cov}(\widehat{\tau},\widehat{\tau}_p)\right]
        \\
            & =  \frac{1}{P^2} \sum_{p=1}^P\left[\left(1-\frac{2}{P}\right)\mathbb{V}(\widehat{\tau}_p)+\mathbb{V}(\widehat{\tau}) +(\tau_p-\tau)^2\right]
        \\
            & = \left(1-\frac{2}{P}\right)\mathbb{V}(\widehat{\tau})+\frac{1}{P^2}\sum_{p=1}^P\mathbb{V}(\widehat{\tau}) +\frac{1}{P^2}\sum_{p=1}^P (\tau_p-\tau)^2
        \\
            & = \left(1-\frac{1}{P}\right)\mathbb{V}(\widehat{\tau}) + \frac{1}{P^2}\sum_{p=1}^P (\tau_p-\tau)^2.
        \end{align*}
    The second equality follows from the fact that by Point \ref{asm:1_p3} of Assumption \ref{asm:1} and Assumption \ref{asm:bal_exp}, $\text{Cov}(\widehat{\tau}_p,\widehat{\tau})=\text{Cov}\left(\widehat{\tau}_p,\sum_{p'} \frac{1}{P}\widehat{\tau}_{p'}\right)=\frac{1}{P}\mathbb{V}(\widehat{\tau}_p)$. The third equality comes from Equation \eqref{eq:var_hat_tau}. This proves the result.

\noindent \textbf{QED}.

\subsubsection*{Point \ref{le:4.1.2}}~\\

The result directly follows from Points \ref{le:CRVE_fe_p1} and \ref{le:CRVE_fe_p2} of Lemma \ref{le:CRVE_nfe} and the fact that $n_{1p}=n_{2p}=n_p/2$ under Assumption \ref{asm:bal_exp}.
\\
\noindent \textbf{QED}.


\subsubsection*{Point \ref{le:4.1.3}}~\\

Let $\overline{Y}_{gp} \equiv \sum_i Y_{igp}/n_{gp}$, $\widehat{Y}_p(1) \equiv \sum_g W_{gp}\overline{Y}_{gp}$, $\widehat{Y}_p(0) \equiv \sum_g (1-W_{gp})\overline{Y}_{gp}$, and $\widehat{Y}(d) \equiv \sum_p \widehat{Y}_{p}(d)/P$, for $d\in \{0,1\}$.
    \begin{align}
        \EX[\widehat{Y}_p(1)] & = \EX\left[\sum_g W_{gp}\overline{y}_{gp}(1)\right]=  \frac{1}{2}\sum_g \overline{y}_{gp}(1) = \overline{y}_p(1).
        \label{eq:exp_yhatp1}
    \end{align}
    The second equality follows from Point \ref{asm:1_p2} of Assumption \ref{asm:1}.
    Similarly,
    \begin{gather}
        \EX[\widehat{Y}_p(0)] = \EX[\overline{y}_p(0) ]
        \label{eq:exp_yhatp0}
        \\
        \EX[\widehat{Y}(d)] = \overline{y}(d), \quad \text{ for } d\in\{0,1\}.
        \label{eq:exp_yhatd}
    \end{gather}
    Then, one has
    \begin{align}
        \widehat{\mathbb{V}}_{unit}(\widehat{\tau})
        - \widehat{\mathbb{V}}_{pair}(\widehat{\tau}) & = \frac{8}{n^2} \sum_p SET_pSEU_p
        \notag
        \\
        & = \frac{8}{n^2} \sum_p \left(\sum_g  W_{gp} \sum_i (y_{igp}(1) - \widehat{Y}(1))\right)\left(\sum_g (1-W_{gp}) \sum_i (y_{igp}(0) - \widehat{Y}(0))\right)
        \notag 
        \\
        & = \frac{8}{n^2} \sum_p \frac{n_{p}^2}{4}\left(\sum_g  W_{gp} \sum_i \frac{y_{igp}(1)}{n_{gp}} - \widehat{Y}(1)\right)\left(\sum_g (1-W_{gp}) \sum_i \frac{y_{igp}(0)}{n_{gp}} - \widehat{Y}(0)\right)
        \notag
        \\
        & =  \frac{2}{P^2} \sum_p \widehat{Y}_p(1)\widehat{Y}_p(0) - \frac{2}{P}\widehat{Y}(1)\widehat{Y}(0)
        \label{eq:diff_ucve_pcve}
    \end{align}
    The first equality follows from Points \ref{le:CRVE_nfe_p1} and \ref{le:CRVE_nfe_p2} of Lemma \ref{le:CRVE_nfe} and Assumption \ref{asm:bal_exp}. The second equality follows from the definitions of $SET_p$, $SEU_p$, and $\epsilon_{igp}$. The third equality follows from Point \ref{asm:1_p1} of Assumption \ref{asm:1}, and Assumption \ref{asm:bal_exp}. The fourth equality follows from Assumption \ref{asm:bal_exp} and some algebra.
    Taking the expectation of \eqref{eq:diff_ucve_pcve},
    \begin{align}
    & \EX\left[\widehat{\mathbb{V}}_{unit}(\widehat{\tau})
        - \widehat{\mathbb{V}}_{pair}(\widehat{\tau})\right] \notag \\
        & = \frac{2}{P^2} \sum_p \left(\text{Cov}(\widehat{Y}_p(1),\widehat{Y}_p(0))\right)+\frac{2}{P^2} \sum_p (\overline{y}_p(1)-\overline{y}(1))(\overline{y}_p(0)-\overline{y}(0)) - \frac{2}{P}\text{Cov}(\widehat{Y}(1),\widehat{Y}(0))
        \notag
        \\
        & = \frac{2}{P^2} \sum_p \left(\text{Cov}(\widehat{Y}_p(1),\widehat{Y}_p(0))\right)+\frac{2}{P^2} \sum_p (\overline{y}_p(1)-\overline{y}(1))(\overline{y}_p(0)-\overline{y}(0)) - \frac{2}{P}\text{Cov}\left(\frac{1}{P}\sum_p\widehat{Y}_p(1),\frac{1}{P}\sum_p\widehat{Y}_p(0)\right)
        \notag \\
        & = \frac{2(P-1)}{P^3} \sum_p \left(\text{Cov}(\widehat{Y}_p(1),\widehat{Y}_p(0))\right)+\frac{2}{P^2} \sum_p (\overline{y}_p(1)-\overline{y}(1))(\overline{y}_p(0)-\overline{y}(0)).        \notag
    \end{align}
    The first equality follows from adding and subtracting
    $\frac{2}{P}\EX[\widehat{Y}(1)]\EX[\widehat{Y}(0)]$ and $\frac{2}{P^2}\sum_p\EX[\widehat{Y}_p(1)]\EX[\widehat{Y}_p(0)]$, and from
    Equations \eqref{eq:exp_yhatp1}, \eqref{eq:exp_yhatp0} and \eqref{eq:exp_yhatd}.
    The third equality follows from Point \ref{asm:1_p3} of Assumption \ref{asm:1}. Therefore,
    \small
    \begin{align}
        \frac{P}{P-1}\EX\left[\widehat{\mathbb{V}}_{unit}(\widehat{\tau})
        - \widehat{\mathbb{V}}_{pair}(\widehat{\tau})\right] & = \frac{2}{P^2} \sum_p \left(\text{Cov}(\widehat{Y}_p(1),\widehat{Y}_p(0))\right) +\frac{2}{P(P-1)} \sum_p (\overline{y}_p(0)-\overline{y}(0))(\overline{y}_p(1)-\overline{y}(1)).
        \label{eq:lemma4.3}
    \end{align}
    \normalsize
    Finally,
    \begin{align}
    \text{Cov}\left(\widehat{Y}_p(1),\widehat{Y}_p(0)\right)
     & =  \EX[\widehat{Y}_p(1)\widehat{Y}_p(0)]-\EX[\widehat{Y}_p(1)]\EX[\widehat{Y}_p(0)]
    \notag \\
    & =  \left(\frac{1}{2}\overline{y}_{1p}(1) \overline{y}_{2p}(0)+\frac{1}{2}\overline{y}_{2p}(1) \overline{y}_{1p}(0)\right) - \left(\frac{1}{2}\sum_g \overline{y}_{gp}(1)\right)\left(\frac{1}{2}\sum_g \overline{y}_{gp}(0)\right) \notag \\
    & = \frac{1}{4}\overline{y}_{1p}(1) \overline{y}_{2p}(0)+\frac{1}{4}\overline{y}_{2p}(1) \overline{y}_{1p}(0) -\frac{1}{4}\overline{y}_{1p}(1) \overline{y}_{1p}(0) -\frac{1}{4}\overline{y}_{2p}(1) \overline{y}_{2p}(0) \notag \\
    & = \frac{1}{4}\left(\overline{y}_{1p}(1) - \overline{y}_{2p}(1) \right) \left(\overline{y}_{2p}(0) - \overline{y}_{1p}(0) \right) \notag \\
    & = - \frac{1}{2} \sum_g \left(\overline{y}_{gp}(0) - \overline{y}_{p}(0) \right)\left(\overline{y}_{gp}(1) - \overline{y}_{p}(1) \right)
    \label{eq:cov_haty1_haty0_exp}
\end{align}
The second equality follows from Points \ref{asm:1_p1} and \ref{asm:1_p2} of Assumption \ref{asm:1}, and Equations \eqref{eq:exp_yhatp1} and \eqref{eq:exp_yhatp0}. The third, fourth, and fifth equalities follow after some algebra.
The result follows plugging Equation \eqref{eq:cov_haty1_haty0_exp} into \eqref{eq:lemma4.3}.
\\
\noindent \textbf{QED}.

\section{B. Large sample results for the pair- and unit-clustered variance estimators}

In this section, we present the large sample distributions of the $t$-tests attached to the four variance estimators we considered in Section \ref{sec:mainresults}. Let
\begin{align*}
\sigma_{pair}^2 =& \lim_{P\rightarrow +\infty} \frac{P\mathbb{V}(\widehat{\tau})}{P\mathbb{V}(\widehat{\tau}) + \frac{1}{P} \sum_{p} (\tau_p-\tau)^2},\\
\Delta_{cov,P}=&\frac{1}{P}\sum_p(\overline{y}_p(0)-\overline{y}(0))(\overline{y}_p(1)-\overline{y}(1))- \frac{1}{P}\sum_{p}\frac{1}{2}\sum_{g}\left(\overline{y}_{gp}(0)-\overline{y}_{p}(0)\right)\left(\overline{y}_{gp}(1)-\overline{y}_{p}(1)\right),\\
\text{ and }\sigma_{unit}^2 =& \underset{P\rightarrow+\infty}{\lim} \frac{P\mathbb{V}(\widehat{\tau})}{P\mathbb{V}(\widehat{\tau}) +\frac{1}{P} \sum_p (\tau_p-\tau)^2 +2\Delta_{cov,P}},
\end{align*}
where Assumption \ref{asm:1b} below ensures the limits in the previous display exist.
\begin{assumption}
    \leavevmode
    \label{asm:1b}
    \begin{enumerate}
        \item \label{asm:1b_p1} For every $d$, $g$ and $p$, there is a constant $M$ such that $\left|\overline{y}_{gp}(d)\right|<M<+\infty$.
        \item \label{asm:1b_p2} When $P\rightarrow +\infty$, $\frac{1}{P} \sum_{p}\tau_p$, $\frac{1}{P} \sum_{p} (\tau_p-\tau)^2$, and $\Delta_{cov,P}$ converge towards finite limits, and  $P\mathbb{V}(\widehat{\tau})$ and $P\mathbb{V}(\widehat{\tau}) +\frac{1}{P} \sum_p (\tau_p-\tau)^2 +2\Delta_{cov,P}$ converge towards strictly positive finite limits.
        \item \label{asm:1b_p3} As $P\rightarrow +\infty$, $\sum_{p=1}^P \EX[|\widehat{\tau}_p-\tau_p|^{2+\epsilon}]/S_P^{2+\epsilon} \rightarrow 0 $ for some $\epsilon>0$, where  $S^2_P \equiv P^2\mathbb{V}(\widehat{\tau})$.
    \end{enumerate}
\end{assumption}

Point \ref{asm:1b_p1} of Assumption \ref{asm:1b} guarantees that we can apply the strong law of large numbers (SLLN) in Lemma 1 in \cite{liu1988bootstrap} to the sequence $(\widehat{\tau}^2_p)_{p=1}^{+\infty}$. Point \ref{asm:1b_p2} ensures that $P\mathbb{V}\left(\widehat{\tau}\right)$ and $P\widehat{\mathbb{V}}_{unit}(\widehat{\tau})$ do not converge towards 0. Point \ref{asm:1b_p3} guarantees that we can apply the Lyapunov central limit theorem to $(\widehat{\tau}_p)_{p=1}^{+\infty}$.
\begin{theorem}($t$-stats' asymptotic behavior)
    \leavevmode
    \label{th:asym}
    Under Assumptions \ref{asm:1}, \ref{asm:bal_exp} and \ref{asm:1b},
    \begin{enumerate}
        \item \label{th:asym_p2} $(\widehat{\tau} - \tau)/\sqrt{\widehat{\mathbb{V}}_{pair}(\widehat{\tau})}=(\widehat{\tau}_{fe} - \tau)/\sqrt{\widehat{\mathbb{V}}_{pair}(\widehat{\tau}_{fe})} ~{\overset{d}{\longrightarrow}}~ \mathcal{N}(0,\sigma_{pair}^2)$. $\sigma_{pair}^2\leq 1$, and if $\tau_p=\tau$ for every $p$, $\sigma_{pair}^2=1$.
        \item \label{th:asym_p3} $(\widehat{\tau}_{fe} - \tau)/\sqrt{\widehat{\mathbb{V}}_{unit}(\widehat{\tau}_{fe})} ~{\overset{d}{\longrightarrow}}~ \mathcal{N}(0,2\sigma_{pair}^2)$.
        \item \label{th:asym_p4} $(\widehat{\tau} - \tau)/\sqrt{\widehat{\mathbb{V}}_{unit}(\widehat{\tau})} ~{\overset{d}{\longrightarrow}}~ \mathcal{N}(0,\sigma_{unit}^2)$.
	    \item \label{th:asym_p6.b} $ \sigma_{unit}^2 \leq \sigma_{pair}^2 $  if and only if $\Delta_{cov,P}$ converges towards a positive limit.
	\end{enumerate}
\end{theorem}
\proof{See Online Appendix \ref{sec:Online Appendix}.}

Point \ref{th:asym_p4} is related to Theorem 3.1 in
\cite{bai2021inference}, who show that when $n_{gp}=1$, the $t$-test in Point \ref{th:asym_p4} under-rejects. The asymptotic variance we obtain is different from theirs, because our results are derived under different assumptions. For instance, we assume a fixed population, while \cite{bai2021inference} assume that the experimental units are an i.i.d.\ sample drawn from an infinite superpopulation, and that asymptotically the expectation of the potential outcomes of two units in the same pair become equal.

\section{C. Clustered variance estimators}\label{sec:clu_var_estimators}

\begin{lemma}[Clustered variance estimators for $\widehat{\tau}$ and $\widehat{\tau}_{fe}$]
    \leavevmode
    \label{le:CRVE_nfe}
    \begin{enumerate}
        \item \label{le:CRVE_nfe_p1} The pair-clustered variance estimator (PCVE) of $\widehat{\tau}$ is
        $
        \widehat{\mathbb{V}}_{pair}(\widehat{\tau}) = \sum_{p=1}^P \left(  \frac{SET_p}{T} - \frac{SEU_p}{C} \right)^2.
        $
        \item \label{le:CRVE_nfe_p2} The unit-clustered variance estimator (UCVE) of $\widehat{\tau}$ is
        $
        \widehat{\mathbb{V}}_{unit}(\widehat{\tau}) = \sum_{p=1}^P \left(  \frac{SET_p^2}{T^2} +  \frac{SEU_p^2}{C^2} \right).
        $
        \item \label{le:CRVE_fe_p1} The PCVE of $\widehat{\tau}_{fe}$ is
        $
        \widehat{\mathbb{V}}_{pair}(\widehat{\tau}_{fe}) =\sum_{p=1}^P \omega_p^2 \left(\widehat{\tau}_p-\widehat{\tau}_{fe}\right)^2.
        $
        \item \label{le:CRVE_fe_p2} The UCVE of $\widehat{\tau}_{fe}$ is
        $
        \widehat{\mathbb{V}}_{unit}(\widehat{\tau}_{fe}) =
                    \sum_{p=1}^P \omega_p^2 \left(\widehat{\tau}_p-\widehat{\tau}_{fe}\right)^2\left(\left(\frac{n_{1p}}{n_p}\right)^2 +\left(\frac{n_{2p}}{n_p}\right)^2\right).
        $
    \end{enumerate}
\end{lemma}

\begin{proof}
See Online Appendix \ref{sec:Online Appendix}.
\end{proof}


\section{D. Variance estimators that rely on pairs of pairs}\label{sec:pop}

We also study two other estimators of $\mathbb{V}\left(\widehat{\tau}\right)$.
Those estimators have been proposed in the one-observation-per-unit special case, but it is straightforward to extend them to the case where all units have the same number of observations, as stated in Assumption \ref{asm:bal_exp}.\footnote{Extending those variance estimators when Assumption \ref{asm:bal_exp} fails is left for future work.}

The first alternative estimator we consider is a slightly modified version of the pairs-of-pairs (POP) variance estimator (POPVE) proposed by \cite{abadie2008}. We only define it when the number of pairs $P$ is even, but in our application in Subsection \ref{sec:pop_app} below we propose a simple method to extend it to cases where the number of pairs is odd.
Let $x_{g,p}$ denote the value of a predictor of the outcome in pair $p$'s unit $g$. Pairs are ordered according to their value of $\frac{x_{1,p}+x_{2,p}}{2}$, the two pairs with the lowest value are matched together, the next two pairs are matched together, and so on and so forth. Let $R=\frac{P}{2}$. For any $r\in \{1,...,R\}$ and for any $p\in \{1,2\}$, let $\widehat{\tau}_{pr}$ denote the treatment effect estimator in pair $p$ of POP $r$. Then, the POPVE is defined as $$\widehat{\mathbb{V}}_{pop}(\widehat{\tau})=\frac{1}{P^2}\sum_{r=1}^R(\widehat{\tau}_{1r}-\widehat{\tau}_{2r})^2.$$
$x_{g,p}$, the variable used to match pairs into POPs, could be the average value of the outcome at baseline in pair $p$'s unit $g$. Or it could be the covariate used to form the pairs, when only one covariate is used.
In our application in subsection \ref{sec:pop_app}, we use the baseline outcome to match pairs into POPs, because the covariates used to match units into pairs are unavailable in most of the data sets of the papers we revisit. Based on Lemma \ref{le:4.1_pop}, we will argue below that the baseline outcome should often be a good choice to match pairs into POPs. The variable one uses to form POPs should be pre-specified and not a function of the treatment assignment. Otherwise, researchers could try to find the variable minimizing the POPVE, which would lead to incorrect inference.

There are two differences between the POPVE and the variance estimator proposed in Equation (3) in \cite{abadie2008}. First, we match pairs with respect to a single covariate, while \cite{abadie2008} consider matching with respect to a potentially multidimensional vector of covariates. This difference is not of essence: we could easily allow pairs to be matched on several covariates. We focus on the unidimensional case as that is the one we use in our application, where the matching is done based on the baseline outcome. Second, the estimator in \cite{abadie2008} matches pairs with replacement, while $\widehat{\mathbb{V}}_{pop}(\widehat{\tau})$ matches pairs without replacement. If after ordering pairs according to their value of $\frac{x_{1,p}+x_{2,p}}{2}$, pair $2$ is closer to pair $3$ than pair $4$, pair $2$ is matched to pairs 1 and 3 in \cite{abadie2008}, while $\widehat{\mathbb{V}}_{pop}(\widehat{\tau})$ matches pair $1$ to pair $2$ and pair $3$ to pair $4$. Matching without replacement makes the properties of $\widehat{\mathbb{V}}_{pop}(\widehat{\tau})$ easier to analyze.

The second alternative variance estimator we consider is that proposed by \cite{bai2021inference} in their Equation (20) (BRSVE). Again, we define this estimator when the number of pairs $P$ is even. With our notation, their estimator is $$\widehat{V}_{brs}(\widehat{\tau})=\frac{1}{P^2}\sum_{p=1}^P\widehat{\tau}_p^2-\frac{1}{2}\left(\frac{2}{P^2}\sum_{r=1}^R\widehat{\tau}_{1r}\widehat{\tau}_{2r}+\frac{\widehat{\tau}^2}{P}\right).$$
\cite{bai2021inference} propose another variance estimator in their Equation (28). That estimator is less amenable to simple comparisons with the UCVE, PCVE, and POPVE, so we do not analyze its properties. However, we compute it in our applications, and find that it is typically similar to the POPVE and BRSVE.


\subsection{Finite-sample results}

Let $\tau_{\cdot r} = \frac{1}{2}(\tau_{1r}+\tau_{2r})$ denote the average treatment effect in POP $r$.

\begin{lemma}
    \label{le:4.1_pop}
    \leavevmode
    If Assumptions \ref{asm:1} and \ref{asm:bal_exp} hold and $P$ is even,
    \begin{enumerate}
        \item \label{le:pop_p1}
        $\EX\left[\widehat{\mathbb{V}}_{pop}(\widehat{\tau})\right] =\mathbb{V}(\widehat{\tau})+\frac{1}{P^2}\sum_{r=1}^R(\tau_{1r}-\tau_{2r})^2$.
        \item \label{le:brs_p1}
        $\widehat{\mathbb{V}}_{brs}(\widehat{\tau})=\frac{1}{2}\widehat{\mathbb{V}}_{pair}(\widehat{\tau})+\frac{1}{2}\widehat{\mathbb{V}}_{pop}(\widehat{\tau})$.
        \item \label{le:pcve_versus_brs}
        If $\frac{1}{R}\sum_{r=1}^R \sum_{p=1,2}\frac{1}{2} (\tau_{pr}-\tau_{\cdot r})^2 \leq \frac{1}{R-1}\sum_{r=1}^R (\tau_{\cdot r}-\tau)^2$,
        \begin{enumerate}
    \item $\EX\left[\widehat{\mathbb{V}}_{pop}(\widehat{\tau})\right]\leq \EX\left[\frac{P}{P-1}\widehat{\mathbb{V}}_{pair}(\widehat{\tau})\right]$,
\item $\EX\left[\widehat{\mathbb{V}}_{pop}(\widehat{\tau})\right]\leq \EX\left[\frac{P}{P-1}\widehat{\mathbb{V}}_{brs}(\widehat{\tau})\right]$,
\item $\EX\left[\widehat{\mathbb{V}}_{brs}(\widehat{\tau})\right]\leq \EX\left[\frac{P}{P-1}\widehat{\mathbb{V}}_{pair}(\widehat{\tau})\right]$.
        \end{enumerate}
    \end{enumerate}
\end{lemma}
\begin{proof}
    See Online Appendix \ref{sec:Online Appendix}.
\end{proof}

Point \ref{le:pop_p1} of Lemma \ref{le:4.1_pop} shows that the POPVE is upward biased in general, and unbiased if the treatment effect is constant within POP. The less treatment effect heterogeneity within POP, the less upward biased the POPVE. An important practical consequence of Point \ref{le:pop_p1} is that the variable used to form POPs should be a good predictor of pairs' treatment effect. The baseline value of the outcome may often be a good predictor of pairs' treatment effect. For instance, treatments sometimes produce a stronger effect on units with the lowest baseline outcome, thus leading to a catch-up mechanism \citep[see for instance][]{glewwe2016better}.

Point \ref{le:pop_p1} of Lemma \ref{le:4.1_pop} is related to Theorem 1 in \cite{abadie2008}, though there are a few differences. \cite{abadie2008} assume that the experimental units are drawn from a super population, and show that once properly normalized, their estimator is consistent for the normalized conditional variance of $\widehat{\tau}$.\footnote{In our setting, the covariates are assumed to be fixed, so the fact that we consider the unconditional variance of $\widehat{\tau}$ while they consider its conditional variance does not explain the difference between our results.} The fact that the POPVE is upward biased in Lemma \ref{le:4.1_pop} and consistent in their Theorem 1
is because we do not assume that the experimental units are an i.i.d.\ sample from a super population. The intuition is the following. In \cite{abadie2008}, when the number of units grows, the covariates $X_i$ on which pairing is based become equal to the same value $x$ for units in the same POP: with an infinity of units, each unit can be matched to another unit with the same $X_i$, and each pair can be matched to another pair with the same $X_i$. Then, asymptotically those units are an i.i.d. sample drawn from the super-population conditional on $X_i=x$, and they all have the same expectation of their treatment effect. Treatment effect heterogeneity within POPs, the source of the POPVE's upward bias in Lemma \ref{le:4.1_pop}, vanishes asymptotically. On the other hand, with a convenience sample, units in the same POP may have asymptotically the same covariates, but they could still have different treatment effects, because they are not i.i.d. draws from a superpopulation.

Point \ref{le:brs_p1} shows that the BRSVE is equal to the average of the PCVE and POPVE. Then, it follows from Point \ref{le:4.1.1} of Lemma \ref{le:4.1} and Point \ref{le:pop_p1} of Lemma \ref{le:4.1_pop} that $\frac{P}{P-1}\widehat{\mathbb{V}}_{brs}(\widehat{\tau})$ is upward biased.
Point \ref{le:brs_p1} is related to Lemma 6.4 and Theorem 3.3 in \cite{bai2021inference}, where the authors show that $P\widehat{\mathbb{V}}_{brs}(\widehat{\tau})$ is consistent for the normalized variance of $\widehat{\tau}$. Here as well, the fact that $P\widehat{\mathbb{V}}_{brs}(\widehat{\tau})$ is upward biased in Lemma \ref{le:4.1_pop} and consistent in \cite{bai2021inference}
comes from the fact we do not assume that the experimental units are an i.i.d.\ sample drawn from a super population.

Finally, Point \ref{le:pcve_versus_brs} shows that if the treatment effect varies less within than across POPs, the POPVE is less upward biased than the degrees-of-freedom-adjusted PCVE and BRSVE, and the BRSVE is less upward biased than the degrees-of-freedom-adjusted PCVE. A sufficient condition to have that the treatment effect varies less within than across POPs is $\frac{1}{R}\sum_{r=1}^R(\tau_{1r}-\tau)(\tau_{2r}-\tau)\geq 0$, meaning that the treatment effects of the two pairs in the same POP are positively correlated.


\subsection{Large-sample results}

\begin{assumption}
    \leavevmode
    \label{asm:1pop}
    When $P\rightarrow +\infty$, $\frac{1}{P} \sum_{r} (\tau_{1r}-\tau_{2r})^2$ converges towards a finite limit.
\end{assumption}

Let
\begin{align*}
\sigma^2_{pop} =& \lim_{P\rightarrow +\infty} \frac{P\mathbb{V}(\widehat{\tau})}{P\mathbb{V}(\widehat{\tau}) + \frac{1}{P} \sum_{r} (\tau_{1r}-\tau_{2r})^2},\\
\sigma_{brs}^2 =& \lim_{P\rightarrow +\infty} \frac{P\mathbb{V}(\widehat{\tau})}{P\mathbb{V}(\widehat{\tau}) + \frac{1}{2P} \sum_{r} (\tau_{1r}-\tau_{2r})^2+ \frac{1}{2P} \sum_{p} (\tau_p-\tau)^2},
\end{align*}
where Assumptions \ref{asm:1b} and \ref{asm:1pop} ensure the limits in the previous display exist.

\begin{theorem}($t$-stats' asymptotic behavior)
    \leavevmode
    \label{th:asym_app}
    Under Assumptions \ref{asm:1}, \ref{asm:bal_exp}, \ref{asm:1b}, and \ref{asm:1pop}, 
    \begin{enumerate}
        \item \label{th:asym_p2prime} $(\widehat{\tau} - \tau)/\sqrt{\widehat{\mathbb{V}}_{pop}(\widehat{\tau})} ~{\overset{d}{\longrightarrow}}~ \mathcal{N}(0,\sigma_{pop}^2)$. $\sigma_{pop}^2\leq 1$, and if $\tau_{1r}=\tau_{2r}$ for every $r$, $\sigma_{pop}^2=1$.
        \item \label{th:asym_p2primeprime} $(\widehat{\tau} - \tau)/\sqrt{\widehat{\mathbb{V}}_{brs}(\widehat{\tau})} ~{\overset{d}{\longrightarrow}}~ \mathcal{N}(0,\sigma_{brs}^2)$. $\sigma_{brs}^2\leq 1$, and if $\tau_p=\tau$ for every $p$, $\sigma_{brs}^2=1$.
	\item
	    \label{th:asym_p6.a} $ \sigma_{pair}^2 \leq \sigma_{brs}^2 \leq \sigma_{pop}^2$  if and only if $0\leq \lim_{P\rightarrow +\infty} \frac{1}{R} \sum_{r=1}^R (\tau_{1r}-\tau)(\tau_{2r}-\tau)$.
	\end{enumerate}
\end{theorem}
\proof{See Online Appendix \ref{sec:Online Appendix}.}

Points \ref{th:asym_p2prime} and \ref{th:asym_p2primeprime} of Theorem \ref{th:asym_app} show that when the number of pairs grows, the $t$-statistic using the POPVE and BRSVE, respectively, converges to a normal distribution with a mean equal to 0 and a variance lower than 1 in general, but equal to 1 when the treatment effect is homogenous across pairs. Therefore, those $t$-tests under-reject. Point \ref{th:asym_p6.a} shows that whenever there is a positive correlation between the treatment effects of the two pairs in the same POP, the $t$-test using the POPVE under-rejects less than that using the BRSVE, which itself under-rejects less than that using the PCVE.

\subsection{Simulations}
\label{sec:pop_simu}

For 26 of the 82 regressions in \cite{crepon2015estimating}, the baseline outcome is available in the authors' data set, so for those outcomes we can simulate the POPVE and BRSVE as well. Those estimators are defined under Assumption \ref{asm:bal_exp}, which does not hold. Therefore in those simulations, we aggregate the data at the village level. We use two samples of 80 and 20 randomly selected pairs out of the original 81 pairs, so as to have an even number of pairs. For each outcome, we simulate 3,000 vectors of treatment assignments, assigning one of the two villages to treatment in each pair. Then, we compute $\widehat{\tau}$, $\widehat{\mathbb{V}}_{pair}(\tau)$, $\widehat{\mathbb{V}}_{pop}(\tau)$, and $\widehat{\mathbb{V}}_{brs}(\tau)$, and the three corresponding 5\% level $t$-tests.

The estimated error rate of each $t$-test is shown in Table \ref{tab:size_popsimus} below. The error rate of the $t$-test using the PCVE is close to 5\% with as few as 20 pairs. On the other hand, the error rates of the $t$-tests using the POPVE and BRSVE are larger than 5\%, even with 80 pairs. Accordingly, we run simulations again, duplicating the random sample of 80 pairs twice to have 160 pairs. The error rate of the $t$-test using the BRSVE is now close to 5\%, but the error rate of the $t$-test using the POPVE is still larger than 5\%. With a sample of 320 pairs obtained by duplicating the random sample of 80 pairs four times, all tests have error rates close to 5\%. With 20 and 80 pairs, we find in our simulations that the correlation between $\widehat{\mathbb{V}}_{pop}(\tau)$ and $|\widehat{\tau}|$ is much weaker than that between $\widehat{\mathbb{V}}_{pair}(\tau)$ and $|\widehat{\tau}|$. This explains why the $t$-test using $\widehat{\mathbb{V}}_{pop}(\tau)$ over-rejects, despite the fact $\widehat{\mathbb{V}}_{pop}(\tau)$ is unbiased: when $|\widehat{\tau}|$ is large, $\widehat{\mathbb{V}}_{pop}(\tau)$ is less likely to be large than $\widehat{\mathbb{V}}_{pair}(\tau)$, so the POPVE $t$-test rejects more often. With 160 and 320 pairs, this phenomenon becomes less pronounced. Overall, the asymptotic approximations in Points \ref{th:asym_p2prime} and \ref{th:asym_p2primeprime} of Theorem \ref{th:asym_app} seem to hold only with a large number of pairs, contrary to that in Point \ref{th:asym_p2} of Theorem \ref{th:asym}.

\begin{table}
    \centering
    \caption{Simulations with data aggregated at village-level to compute $\widehat{\mathbb{V}}_{pop}$ and $\widehat{\mathbb{V}}_{brs}$}
    \label{tab:size_popsimus}
    \begin{tabular}{lcccc}
        \hline \hline
         \multirow{3}{2cm}{Variance estimator} & \multicolumn{4}{c}{5\% level $t$-test error rate}  \\
         & \multirow{2}{1.7cm}{With 20 pairs} & \multirow{2}{1.7cm}{With 80 pairs} & \multirow{2}{1.7cm}{With 160 pairs} & \multirow{2}{1.7cm}{With 320 pairs} \\
         \\
        \hline
        PCVE		   & 0.0504	& 0.0505 & 0.0506 & 0.0503 \\
        POPVE      & 0.1301	& 0.0818 & 0.0656 & 0.0565 \\
        BRSVE	   & 0.0808	& 0.0619 &	0.0571 & 0.0530 \\
         \hline
        \end{tabular}

    \begin{tablenotes}
    The table reports the error rates of three 5\% level $t$-tests in \cite{crepon2015estimating}, aggregating data at the village level. For each of the 26 outcomes in the paper for which the baseline outcome is available,  we randomly drew 3,000 simulated treatment assignments, following the paired assignment used by the authors, and computed the treatment effect estimator $\widehat{\tau}$, the pair-clustered variance estimator (PCVE), the pairs-of-pairs variance estimator (POPVE) in \cite{abadie2008}, the variance estimator in \citet{bai2021inference} (BRSVE), and the three corresponding $t$-tests. The error rate of each test is the percent of times, across the 78,000 regressions (26 outcomes $\times$ 3,000 replications), that the test leads the researcher to wrongly conclude that the treatment has an effect. Column 2 (resp.\ 3, 4, 5) shows the results using a random sample of 20 pairs (resp.\ a random sample of 80 pairs, the same random sample of 80 pairs duplicated twice, the same random sample of 80 pairs duplicated four times).
    \end{tablenotes}
\end{table}

\subsection{Application}
\label{sec:pop_app}

For 152 of the 294 regressions in Panel A of Table \ref{tab:replic}, the baseline outcome is available in the data set, so we can compute the POPVE and BRSVE. Those estimators are defined under Assumption \ref{asm:bal_exp}, which does not hold in all those regressions. Therefore, we compute the POPVE and BRSVE after aggregating the data at the unit level. When the number of pairs is odd, we compute the POPVE twice, first excluding the pair with the lowest value of the baseline outcome, then excluding the pair with the highest value of the baseline outcome, and we finally take the average of the two estimators. We do the same for the BRSVE when the number of pairs is odd. We also recompute the PCVE without pair fixed effects with the aggregated data, using the exact same sample as that used to compute the POPVE and BRSVE. Across those 152 regressions, the POPVE divided by the PCVE is on average equal to 1.026. The BRSVE divided by the PCVE is on average equal to 1.014.\footnote{The variance estimator in Equation (28) of \cite{bai2021inference} is also on average higher than the PCVE.}
In those regressions, the POPVE and BRSVE do not lead to power gains.


\section{E. Extension: stratified experiments with few units per strata}
\label{se:ext_small_strata}

In this section, we perform Monte-Carlo simulations to assess how our results in Section \ref{sec:mainresults} extend to stratified RCTs where the number of units per strata is larger than two, but still fairly small. Three main findings emerge. First, the error rates of $t$-tests using stratum-clustered standard errors are equal to 5\%. Second, the error rates of $t$-tests using standard errors clustered at the unit level are larger than 5\% in regressions with stratum fixed effects, but decrease as the number of units per strata increases. With 5 units per strata, and averaging across Panels A to D of Table \ref{tab:ext_size} below, the error rate of a 5\% level test with UCVE and stratum fixed effects is around 7.9\%, while with 10 units per strata this error rate is around 6.2\%. Finally, the error rates of $t$-tests using standard errors clustered at the unit level tend to be lower than 5\% in regressions without stratum fixed effects.

We draw the potential and observed outcomes from the following data generating process (DGP),
\begin{gather}
    Y_{igp} = W_{gp}y_{igp}(1) + (1-W_{gp})y_{igp}(0) + \gamma_p, \qquad i=1,\dots,n_{gp};\ g=1,\dots,G;\ p=1,\dots,P,
    \label{eq:strat}
\end{gather}
where $y_{igp}(1) $ and $ y_{igp}(0)$ are independent and both follow a $\mathcal{N}(0,1)$ distribution, $\{\gamma_{p}\}_p \sim \text{iid } \mathcal{N}(0,\sigma^2_\gamma)$, and $(y_{igp}(1),y_{igp}(0))\independent \gamma_{p}$. We either let $\sigma_{\eta} = 0$ or $\sigma_{\eta} = \sqrt{0.1}$.  $\sigma_{\eta} = 0$ corresponds to a model with no stratum common shock, while $\sigma_{\eta} = \sqrt{0.1}$ corresponds to a model with a shock. We draw potential outcomes once and keep them fixed, so $y_{igp}(1) $, $ y_{igp}(0)$ and $\gamma_{p}$ do not vary across simulations.

Each stratum has $G$ units. We vary $G$ from two to ten. If $G$ is even, then half of the units are randomly assigned to the control and the remaining to the treatment. If $G$ is odd, then $(G+1)/2
$ units are randomly assigned to the control.
We also set $n_{gp} = 5 $ or $n_{gp} = 100$, and we let the number of strata $P$ be equal to $100$.

We compute $t$-tests based on unit- and stratum-clustered standard errors in regressions of the outcome on the treatment with and without stratum fixed effects. We perform 10,000 simulations for each DGP. Table \ref{tab:ext_size} presents the error rates of the $t$-tests in each DGP.  

$t$-tests using stratum-clustered standard errors achieve error rates close to 5\% for all data configurations (as in Table \ref{tab:size}, with $n_{gp} = 5 $, the $t-$test using the PCVE with stratum fixed effects under-rejects slightly, due to the DOF-adjustment). In contrast, $t$-tests based on unit-clustered standard errors in regressions with stratum fixed effects overreject the true null of no treatment effect. These results are in line with Points \ref{th:asym_p2} and \ref{th:asym_p3} of Theorem \ref{th:asym}, which covered the special case where $G=2$. $t$-tests based on unit-clustered standard errors in regressions with stratum fixed effects over-reject less as the number of units per strata increases from two (column 2) to ten (column 10). Interestingly, it seems that unit-clustered standard errors are approximately equal to $\sqrt{\frac{G-1}{G}}$ times the stratum-clustered standard errors. If $G=2$, the ratio of those two standard errors is exactly equal to $\sqrt{(2-1)/2}=\sqrt{1/2}$ as shown in Lemma \ref{le:4.1}, but this relationship seems to still hold in expectation for larger values of $G$.

In Panel A, $t$-tests based on unit-clustered standard errors in regressions without stratum fixed effects have error rates close to 5\%. When $\sigma_{\eta}=0$, there is no between and within strata heterogeneity in $\overline{y}_{gp}(0)$, so it follows from Point \ref{th:asym_p4} of Theorem \ref{th:asym} that in the special case where $G=2$, $t$-tests based on unit-clustered standard errors in regressions without stratum fixed effects have error rates close to 5\%. Our simulations suggest that this result still holds when $G>2$. However, in Panel B, $t$-tests using unit-clustered standard errors in regressions without stratum fixed effects have error rates lower than 5\%, because there is now between-strata heterogeneity in $\overline{y}_{gp}(0)$. We obtain similar results with five observations per unit (Panels C and D).
\begin{table}[]
\centering
\caption{Error rates of $t$-test in simulated stratified RCTs with small strata}

\resizebox{\textwidth}{!}{
    \begin{tabular}{lccccccccc}
    \hline
    &  \multicolumn{9}{c}{Number of units per strata} \\
    &  2 & 3 & 4 & 5 & 6 & 7 & 8 & 9 & 10
    \\ \hline
    \\
    \multicolumn{10}{l}{\textit{Panel A. iid standard normal potential outcomes and $n_{gp} = 100$}}\\
    \\
UCVE without FE&0.0415&0.0734&0.0501&0.0480&0.0501&0.0478&0.0468&0.0538&0.0479\\
UCVE with FE&0.1685&0.1160&0.0928&0.0806&0.0705&0.0696&0.0631&0.0662&0.0614\\
SCVE without FE&0.0557&0.0587&0.0529&0.0545&0.0514&0.0516&0.0495&0.0524&0.0518\\
SCVE with FE&0.0554&0.0582&0.0527&0.0544&0.0511&0.0515&0.0493&0.0524&0.0516\\
$\widehat{{s.e.}}_{unit}(\widehat{\tau}_{fe})/\widehat{{s.e.}}_{strat}(\widehat{\tau}_{fe})$&0.7053&0.8169&0.8720&0.8986&0.9188&0.9308&0.9408&0.9492&0.9551\\
    \\
    \multicolumn{10}{l}{\textit{Panel B. Stratum-level shock affecting potential outcomes and $n_{gp} = 100$}}
    \\
UCVE without FE&0.0000&0.0000&0.0000&0.0000&0.0000&0.0000&0.0000&0.0000&0.0000\\
UCVE with FE&0.1682&0.1105&0.0909&0.0812&0.0775&0.0671&0.0641&0.0650&0.0637\\
SCVE without FE&0.0518&0.0507&0.0532&0.0524&0.0566&0.0504&0.0511&0.0546&0.0528\\
SCVE with FE&0.0510&0.0506&0.0531&0.0520&0.0565&0.0503&0.0508&0.0545&0.0526\\
$\widehat{{s.e.}}_{unit}(\widehat{\tau}_{fe})/\widehat{{s.e.}}_{strat}(\widehat{\tau}_{fe})$&0.7053&0.8163&0.8700&0.8995&0.9187&0.9313&0.9420&0.9494&0.9548\\
    \\
    \multicolumn{10}{l}{\textit{Panel C. iid standard normal potential outcomes and $n_{gp} = 5$}}
    \\
UCVE without FE&0.0547&0.0456&0.0474&0.0575&0.0492&0.0527&0.0512&0.0527&0.0544\\
UCVE with FE&0.1478&0.0981&0.0872&0.0745&0.0719&0.0687&0.0629&0.0646&0.0627\\
SCVE without FE&0.0515&0.0542&0.0542&0.0552&0.0557&0.0548&0.0514&0.0543&0.0552\\
SCVE with FE&0.0397&0.0471&0.0481&0.0502&0.0522&0.0519&0.0486&0.0515&0.0529\\
$\widehat{{s.e.}}_{unit}(\widehat{\tau}_{fe})/\widehat{{s.e.}}_{strat}(\widehat{\tau}_{fe})$&0.7053&0.8163&0.8695&0.8979&0.9175&0.9327&0.9404&0.9485&0.9551\\
    \\
    \multicolumn{10}{l}{\textit{Panel D. Stratum-level shock affecting potential outcomes and $n_{gp}=5$}}
    \\
UCVE without FE&0.0128&0.0152&0.0207&0.0122&0.0146&0.0173&0.0130&0.0160&0.0158\\
UCVE with FE&0.1539&0.1033&0.0793&0.0730&0.0687&0.0682&0.0625&0.0640&0.0671\\
SCVE without FE&0.0533&0.0546&0.0529&0.0520&0.0513&0.0564&0.0507&0.0536&0.0574\\
SCVE with FE&0.0430&0.0469&0.0476&0.0470&0.0475&0.0526&0.0482&0.0510&0.0545\\
$\widehat{{s.e.}}_{unit}(\widehat{\tau}_{fe})/\widehat{{s.e.}}_{strat}(\widehat{\tau}_{fe})$&0.7053&0.8171&0.8703&0.8991&0.9179&0.9315&0.9419&0.9484&0.9552\\
    \\
    \hline
    \end{tabular}
    }
    \begin{tablenotes}
         The table shows the error rates of $t$-tests based on unit- and stratum-clustered standard errors in regressions with and without stratum fixed effects. Across simulations, we vary the number of units per strata from two to ten ($G= 2,\dots,10$); we vary the number of observations per unit to either $n_{gp} = 5$ or $n_{gp} = 100$; and we set the number of strata to $P = 100$. For each value of $G$, we simulated 10,000 samples from the following data generating processes: independent and identically distributed (iid) standard normal potential outcomes in Panels A and C, and a model with an additive stratum-level shock affecting both potential outcomes in Panel B and D. UCVE and SCVE stand for unit- and stratum-clustered variance estimators, respectively. FE stands for stratum fixed effects. $\frac{\widehat{{s.e.}}_{unit}(\widehat{\tau}_{fe})}{\widehat{{s.e.}}_{strat}(\widehat{\tau}_{fe})}$ is the average across simulations of the ratio of standard errors clustering at the unit and stratum levels in regressions with stratum fixed effects.
    \end{tablenotes}

\label{tab:ext_size}
\end{table}

\clearpage

\newpage

\section{F. Articles in our survey of paired or small strata experiments}

\begin{table}
\caption{Paired RCTs and stratified RCTs with small strata}
\label{tab:litrev}

    \begin{tabular}{lc}
    Reference & Search source
    \\ \hline
    Paired RCTs\\
    \hspace{0.5cm} \cite{ashraf2006deposit} & AEA registry \\
    \hspace{0.5cm} \cite{banerjee2015miracle} & \textit{AEJ: Applied} \\
   \hspace{0.5cm} \cite{crepon2015estimating} & \textit{AEJ: Applied} \\
    \hspace{0.5cm} \cite{beuermann2015one} & \textit{AEJ: Applied} \\
    \hspace{0.5cm} \cite{fryer2016vertical} & AEA registry \\
    \hspace{0.5cm} \cite{glewwe2016better} & AEA registry \\
    \hspace{0.5cm} \cite{bruhn2016impact}   & \textit{AEJ: Applied} \\
    \hspace{0.5cm} \cite{fryer2017management}  & AEA registry \\
    Small-strata RCTs \\
    \hspace{0.5cm} \cite{attanasio2015impacts} & \textit{AEJ: Applied} \\
    \hspace{0.5cm} \cite{angelucci2015microcredit} & \textit{AEJ: Applied} \\
    \hspace{0.5cm} \cite{ambler2015channeling} & \textit{AEJ: Applied} \\
    \hspace{0.5cm} \cite{bjorkman2017experimental} & \textit{AEJ: Applied} \\
    \hspace{0.5cm} \cite{banerji2017impact} & \textit{AEJ: Applied} \\
    \hspace{0.5cm} \cite{lafortune2018role} & \textit{AEJ: Applied} \\
    \hspace{0.5cm} \cite{somville2018saving} & \textit{AEJ: Applied}    \\ \hline
	\end{tabular}

  \begin{tablenotes}
        The table presents economics papers that have conducted clustered and paired RCTs, or clustered and stratified RCTs with ten or less units per strata. We searched the \textit{AEJ: Applied Economics} for papers published in 2014-2018 and using the words ``random" and ``experiment" in the abstract, title, keywords, or main text. Four of those papers had conducted a clustered and paired RCT and seven had conducted a clustered and stratified RCT with ten units or less per strata. We also searched the AEA's registry website for RCTs (\verb+https://www.socialscienceregistry.org+). We looked at all completed projects, whose randomization method included the word ``pair" and that had either a working or a published paper. Thus, we found four more papers that had conducted a clustered and paired RCT. \cite{beuermann2015one} use a paired design to estimate the spillover effects of the intervention they consider. Their estimation of the direct effects of that intervention relies on another type of randomization. We only include their spillover analysis in our survey and in our replication.
  \end{tablenotes}

\end{table}

\clearpage

\newpage



\section{G. Results when the number of observations varies across units}
\label{sec:ext_assm2}

In this section, we extend some of the results in Section \ref{sec:mainresults} to instances where units may have different numbers of observations, as is often the case in practice.

\subsection{Upward bias of the pair-clustered variance estimator (PCVE)}

In this subsection, we show that when units have different numbers of observations, our recommendation of using the PCVE still applies.

When units have different numbers of observations, there are several estimators of the treatment effect one may consider.
$\widehat{\tau}$, the standard difference in means estimator, is such that
 \begin{gather*}
        \widehat{\tau} = \frac{1}{T}\sum_{p=1}^P \sum_{g=1}^2 \sum_{i=1}^{n_{gp}} Y_{igp}W_{gp}
        - \frac{1}{C}\sum_{p=1}^P \sum_{g=1}^2 \sum_{i=1}^{n_{gp}} Y_{igp}(1-W_{gp}),
    \end{gather*}
where $T$ and $C$ respectively denote the total number of treated and control observations. When the number of observations varies across units, $T$ and $C$ are stochastic. For instance, assume one has two pairs. In pair 1, units 1 and 2 both have 1 observation, but in pair 2 unit 1 has 1 observations while unit 2 has 2 observations. Then, $T$ is equal to 2 with probability 1/2, and to 3 with probability 1/2. These stochastic denominators in $\widehat{\tau}$ make it impossible to derive a closed-form expression of its expectation and variance. One can still show that when the number of pairs goes to infinity, $\widehat{\tau}$ converges toward $\tau$, the average treatment effect, and one could also use the delta method to show that $\widehat{\tau}$ is asymptotically normal and derive its asymptotic variance. However, throughout the paper we have focused on estimators' finite sample variances.

Therefore, instead of $\widehat{\tau}$ we consider another, closely related estimator, whose expectation and variance are straightforward to derive even when the number of observations varies across units, and which is unbiased for a causal effect that differs from $\tau$ but that is still relatively natural (see \cite{imai2009essential} for closely related discussions).
Let $\widetilde{\tau}$ denote the coefficient of $W_{gp}$ in the weighted OLS regression of $Y_{igp}$ on a constant and $W_{gp}$, with weights $V_{gp} =  n_{p}/n_{gp}$.\footnote{Specifically, the intercept $\widetilde{\alpha}$ and $\widetilde{\tau}$ are such that   $(\widetilde{\alpha},\widetilde{\tau})=  \argmin_{\alpha, \tau} \sum_{p} \sum_g \sum_i V_{gp}(Y_{igp}-\alpha - \tau W_{gp})^2$.} Let $\widetilde{\alpha}$ be the intercept in that regression. One can show that 
\begin{gather}
    \widetilde{\tau} = \frac{1}{P} \sum_p \frac{n_{p}}{\overline{n}} \sum_g  (W_{gp}\overline{Y}_{gp}-(1-W_{gp})\overline{Y}_{gp}) = \frac{1}{P} \sum_p \frac{n_{p}}{\overline{n}}\widehat{\tau}_p,
    \label{eq:tilde_tau}
\end{gather}
where
$\bar{n} = n/P$.
Under Assumption \ref{asm:bal_exp}, $\widetilde{\tau}=\widehat{\tau}$. Hence, $\widetilde{\tau}$ generalizes 
$\widehat{\tau}$ to the case where the number of observations varies across units. $\widetilde{\tau}$ is also one of the estimators considered by \cite{imai2009essential}, though the fact $\widetilde{\tau}$ can be obtained by weighted least squares is not noted therein.

$\widetilde{\tau}$ is generally not unbiased for $\tau$, unless in every pair, the two units have the same number of observations, i.e., $n_{1p} = n_{2p}$ for all $p$ \citep{imai2009essential}. On the other hand, $\widetilde{\tau}$ is unbiased for $$\tau^* = \frac{1}{P}\sum_p \frac{n_{p}}{\overline{n}} \left(\frac{\tau_{1p}}{2}+\frac{\tau_{2p}}{2}\right)$$
where $\tau_{gp} = \frac{1}{n_{gp}}\sum_{i=1}^{n_{gp}} [y_{igp}(1)-y_{igp}(0)]$
denotes the average treatment effect in unit $g$ of pair $p$.\footnote{With a slight abuse of notation, $\tau_{1r}$ and $\tau_{2r}$ refer to the ATE in pairs 1 and 2 of POP $r$, while $\tau_{1p}$ and $\tau_{2p}$ refer to the ATE in units 1 and 2 of pair $p$.} $\tau^*$ is
a weighted average of the pair-specific average treatment effects $(\tau_{1p}+\tau_{2p})/2$. Those pair-specific average treatment effects give equal weight to the average treatment effect in each unit, rather than weighting them according to their number of observations like $\tau_p$ does.
\cite{imai2009essential} show that
\begin{gather*}
    \mathbb{V}(\widetilde{\tau}) =\frac{1}{4P^2}\sum_p \frac{n_p^2}{\overline{n}^2} \left(\Delta_{p}(1)+\Delta_{p}(0) \right)^2,
\end{gather*}
where $\Delta_{p}(1) \equiv \overline{y}_{1p}(1)-\overline{y}_{2p}(1)$ and $\Delta_{p}(0)\equiv \overline{y}_{1p}(0)-\overline{y}_{2p}(0)$.
They propose various estimators of that variance, and show that they are upward biased. Instead, we rely on the fact $\widetilde{\tau}$ can be obtained by weighted least squares to propose an estimator whose properties have not been studied in the randomization-inference framework we consider: the PCVE attached to $\widetilde{\tau}$.

First, the following lemma extends Lemma \ref{le:CRVE_nfe} to the PCVE in a weighted OLS regression.\footnote{We follow the definition of clustered variance estimators for weighted least squares in Equation (15) of \cite{cameron2015practitioner}.}
\begin{lemma}[Pair-clustered variance estimator for $\widetilde{\tau}$]
    \leavevmode
    \label{le:CRVE_nfe_ext}
        $    \widehat{\mathbb{V}}_{pair}(\widetilde{\tau})  = \frac{1}{P^2} \sum_p  \frac{n_p^2}{\overline{n}^2} \left[\widehat{\tau}_p - \widetilde{\tau}\right]^2.$
\end{lemma}
\proof{See Online Appendix \ref{sec:Online Appendix}.}

Then, we study the asymptotic distribution of the $t$-statistic attached to $\widetilde{\tau}$ and $\widehat{\mathbb{V}}_{pair}(\widetilde{\tau})$. To do so, we make
the following assumption.
\begin{assumption}
    \label{asm:3_ext}
    \leavevmode
    \begin{enumerate}
        \item \label{asm:3_ext_p1} For all $g$ and $p$, $1\leq n_{gp}\leq N$ for some fixed $N<+\infty$.
        \item \label{asm:3_ext_p2}    As $P\rightarrow +\infty$, $\frac{1}{P} \sum_p \left(\frac{n_p}{\overline{n}}\right)^2$, $\frac{1}{P} \sum_p \left( \frac{n_p}{\overline{n}}\right)^2 \EX[\widehat{\tau}_p]$, and $\frac{1}{P} \sum_p \left( \frac{n_p}{\overline{n}}\right)^2 (\EX[\widehat{\tau}_p])^2$
 converge to strictly positive constants, and $\tau^*=\frac{1}{P}\sum_p \frac{n_{p}}{\overline{n}} \left(\frac{\tau_{1p}}{2}+\frac{\tau_{2p}}{2}\right)$
 converges to a constant $\tau^{\infty}$.
       \item \label{asm:3_ext_p3} As $P\rightarrow +\infty$, $\sum_{p=1}^P \EX\left[|\frac{n_p}{\overline{n}}|^{2+\epsilon}|\widehat{\tau}_p-\EX[\widehat{\tau}_p]|^{2+\epsilon}\right]/{\widetilde{S}}_P^{2+\epsilon} \rightarrow 0 $ for some $\epsilon>0$, where  ${\widetilde{S}}^2_P \equiv P^2\mathbb{V}(\widetilde{\tau})$.
    \end{enumerate}
\end{assumption}
Point \ref{asm:3_ext_p1} of Assumption \ref{asm:3_ext} requires that the number of observations in every unit is greater than 1 and lower than some fixed $N$. Combined with Point \ref{asm:1b_p2} of Assumption \ref{asm:1b}, Point \ref{asm:3_ext_p2} of Assumption \ref{asm:3_ext} ensures that $P\widehat{\mathbb{V}}_{pair}(\widetilde{\tau})$ converges towards a strictly positive limit. Point \ref{asm:3_ext_p3} guarantees that we can apply the Lyapunov central limit theorem to $\left(\frac{n_p}{\overline{n}}\widehat{\tau}_p\right)_{p=1}^{+\infty}$.
Let $\sigma_{wls}^2 = \lim_{P\rightarrow +\infty} \frac{P\mathbb{V}(\widetilde{\tau})}{P\mathbb{V}(\widetilde{\tau}) + \frac{1}{P} \sum_{p} \left(\frac{n_p}{\overline{n}}\right)^2 \left(\EX\left(\widehat{\tau}_p\right)-\tau^{\infty}\right)^2}$.
\begin{theorem}
    \leavevmode
    \label{th:asym_ext}
 If Assumptions \ref{asm:1} and \ref{asm:3_ext}, and Points \ref{asm:1b_p1} and \ref{asm:1b_p2} of Assumption \ref{asm:1b} hold, \linebreak $(\widetilde{\tau} - \tau^*)/\sqrt{\widehat{\mathbb{V}}_{pair}(\widetilde{\tau})} ~{\overset{d}{\longrightarrow}}~ \mathcal{N}(0,\sigma_{wls}^2)$. $\sigma_{wls}^2\leq 1$, and if $\tau_{gp}=\tau$ for every $g$ and $p$, or if $n_{1p}=n_{2p}$ and $\tau_p=\tau$ for every $p$, then  $\sigma_{wls}^2=1$.
\end{theorem}
\proof{See Online Appendix \ref{sec:Online Appendix}.}

This theorem shows that when the number of pairs grows, the $t$-statistic of the weighted least squares estimator using the PCVE converges to a normal distribution with a mean equal to 0 and a variance lower than 1 in general, but equal to 1 when the treatment effect is homogenous across units, or when the treatment effect is homogenous across pairs and in every pair the two units have the same number of observations.

Theorem \ref{th:asym_ext} shows that when units have different numbers of observations, the PCVE attached to $\widetilde{\tau}$ is upward biased asymptotically. We now show that the same holds for $\widehat{\tau}_{fe}$, the pair fixed effects estimator, provided one applies some kind of degrees-of-freedom correction to its PCVE. As shown in Point \ref{le:CRVE_fe_p1} of Lemma \ref{le:CRVE_nfe}, the PCVE of $\widehat{\tau}_{fe}$ is $
        \widehat{\mathbb{V}}_{pair}(\widehat{\tau}_{fe}) =\sum_{p=1}^P \omega_p^2 \left(\widehat{\tau}_p-\widehat{\tau}_{fe}\right)^2.$ Let $\widetilde{\omega}_p = \omega_p \left(1-2\omega_p\right)^{-1/2} $.
\begin{lemma}[The adjusted PCVE for $\widehat{\tau}_{fe}$ is upward biased]
\label{le:exp_PCVE_fe}
Under Assumption \ref{asm:1}, and if $\omega_p\leq 1/2$ for all $p$, $\EX\left[\sum_p \widetilde{\omega}_p^2 (\widehat{\tau}_p-\widehat{\tau}_{fe})^2\right] =\mathbb{V}(\widehat{\tau}_{fe})\left(1+\textstyle\sum_p \widetilde{\omega}_p^2\right)+ \sum_p \widetilde{\omega}_p^2[\EX(\widehat{\tau}_p-\widehat{\tau}_{fe})]^2$.
\begin{proof}
See Online Appendix \ref{sec:Online Appendix}.
\end{proof}
\end{lemma}
Lemma \ref{le:exp_PCVE_fe} shows that the adjusted PCVE, where the $\omega_p$ are replaced by $\widetilde{\omega}_p$, is upward biased for the variance of $\widehat{\tau}_{fe}$. The adjustment in $\widetilde{\omega}_p$ is similar to a degrees-of-freedom adjustment. In fact, under Assumption \ref{asm:bal_exp}, the adjusted PCVE is equal to $\frac{P}{P-2}\widehat{\mathbb{V}}_{pair}(\widehat{\tau}_{fe})$. The requirement that $\omega_p\leq 1/2$ for all $p$ is mild. For instance, if $n_{1p}=n_{2p}$ for all $p$, this only requires that every pair has fewer observations than all other pairs combined. If there is an integer $L$ such that $n_p\leq L$ for every $p$, one can show that $\underset{P\rightarrow +\infty}{\liminf} \EX\left[P\left(\widehat{\mathbb{V}}_{pair}(\widehat{\tau}_{fe})-\mathbb{V}(\widehat{\tau}_{fe})\right)\right]\geq 0$: the unadjusted PCVE is also upward biased asymptotically. When the number of observations varies across units, $\widehat{\mathbb{V}}_{pair}(\widehat{\tau}_{fe})$ does not coincide with the estimator of the variance of $\widehat{\tau}_{fe}$ considered in \cite{imai2009essential}. It seems that Lemma \ref{le:exp_PCVE_fe} above is the first result to justify the use of the PCVE attached to $\widehat{\tau}_{fe}$, in paired RCTs where the number of observations varies across units.

\subsection{Ratio of the UCVE and PCVE with pair fixed effects}

In this subsection, we derive the ratio of the UCVE and PCVE with pair fixed effects when units have different numbers of observations.

\begin{lemma}[Ratio of the UCVE and PCVE with pair fixed effects when units have different numbers of observations]
\label{le:ext_fe}
\leavevmode
$\widehat{\mathbb{V}}_{unit}(\widehat{\tau}_{fe})/\widehat{\mathbb{V}}_{pair}(\widehat{\tau}_{fe}) = \sum_p \zeta_p \left(\left(\frac{n_{1p}}{n_{p}}\right)^2+ \left(\frac{n_{2p}}{n_{p}}\right)^2\right) $, where, for all $p$ $\zeta_p\geq 0$ and $\sum_p \zeta_p = 1$. Therefore, $\widehat{\mathbb{V}}_{unit}(\widehat{\tau}_{fe})/\widehat{\mathbb{V}}_{pair}(\widehat{\tau}_{fe}) \in \left[\frac{1}{2},1\right]$.
\end{lemma}
\begin{proof}
The formula for $\widehat{\mathbb{V}}_{unit}(\widehat{\tau}_{fe})/\widehat{\mathbb{V}}_{pair}(\widehat{\tau}_{fe})$ follows from Points \ref{le:CRVE_fe_p1} and \ref{le:CRVE_fe_p2} of Lemma \ref{le:CRVE_nfe}, with $$\zeta_p=\frac{\omega_p^2 \left(\widehat{\tau}_p-\widehat{\tau}_{fe}\right)^2}{\sum_{p=1}^P \omega_p^2 \left(\widehat{\tau}_p-\widehat{\tau}_{fe}\right)^2}.$$
$n_{1p}^2+n_{2p}^2 \leq (n_{1p}+n_{2p})^2 $, so $\left(\frac{n_{1p}}{n_{p}}\right)^2+ \left(\frac{n_{2p}}{n_{p}}\right)^2\leq 1$. $(n_{1p}-n_{2p})^2 = n_{1p}^2-2n_{1p}n_{2p}+n_{2p}^2\geq 0$, so $2n_{1p}^2+2n_{2p}^2\geq (n_{1p}+n_{2p})^2$, and $\left(\frac{n_{1p}}{n_{p}}\right)^2+ \left(\frac{n_{2p}}{n_{p}}\right)^2\geq \frac{1}{2}$. Therefore, $\widehat{\mathbb{V}}_{unit}(\widehat{\tau}_{fe})/\widehat{\mathbb{V}}_{pair}(\widehat{\tau}_{fe}) \in \left[\frac{1}{2},1\right]$.
\end{proof}

Lemma \ref{le:ext_fe} shows that $\widehat{\mathbb{V}}_{unit}(\widehat{\tau}_{fe})/\widehat{\mathbb{V}}_{pair}(\widehat{\tau}_{fe})$ is a weighted average across pairs of the sum of the squared shares that each unit accounts for in the pair. The sum of these squared shares is included between a half and one, so this ratio is included between a half and one. Figure \ref{fig:unbal} plots this ratio when $n_{1p}/n_{2p}$ is constant across pairs. $\widehat{\mathbb{V}}_{unit}(\widehat{\tau}_{fe})/\widehat{\mathbb{V}}_{pair}(\widehat{\tau}_{fe})$ is close to $1/2$ when $n_{1p}/n_{2p}$ is included between 0.5 and 2, meaning that the first unit has between half and twice as many observations as the second one. For instance, if in every pair, one unit has twice as many observations as the other, then the ratio of the two variances is equal to $5/9$. Based on Figure \ref{fig:unbal}, one can also derive an upper bound for $\widehat{\mathbb{V}}_{unit}(\widehat{\tau}_{fe})/\widehat{\mathbb{V}}_{pair}(\widehat{\tau}_{fe})$, when $n_{1p}/n_{2p}$ varies across pairs. For instance, if in every pair, one unit has at most twice as many observations as the other, as should often be the case in practice, then the ratio of the two variances is at most equal to $5/9$.
Overall, Lemma \ref{le:ext_fe} shows that Point \ref{le:4.1.2} of Lemma \ref{le:4.1} still approximately holds when units in each pair have different numbers of observations, unless they have an extremely unbalanced number of observations.


\begin{figure}[ht]
    \centering
    \includegraphics{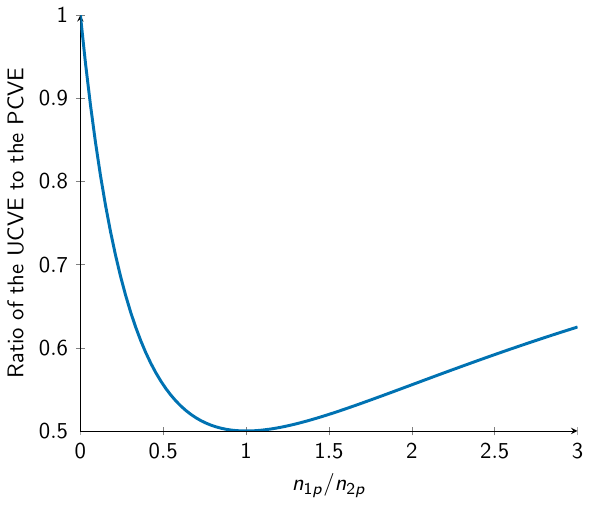}
    \caption{Ratio of Unit-Clustered and Pair-Clustered Variance Estimators with Pair Fixed Effects}
    \label{fig:unbal}
    \begin{tablenotes}
    UCVE and PCVE stand for unit- and pair clustered variance estimators, respectively. $n_{1p}$ and $n_{2p}$ are the number of observations in units 1 and 2 of pair $p$, respectively.
    \end{tablenotes}
\end{figure}



\newpage

\section{H. Proofs of the results in the Online Appendix}
\label{sec:Online Appendix}


\subsection{Proof of Theorem \ref{th:asym}}

    {
    The proof relies on Lemma \ref{le:plim_tilde_tau_q} and on the two equations below.

    Using a similar reasoning as that used to show Equation \eqref{eq:cond_tau_p} in the proof of Lemma \ref{le:plim_tilde_tau_q}, one can show that
    \begin{align}
        \EX\left[\left|\widehat{Y}_p(d)\right|^{2+\epsilon}\right] \leq M_1 < +\infty.
        \label{eq:cond_haty1_p}
    \end{align}
    for all $d$ and $p$ and for some $M_1>0$.

    By Lemma \ref{le:plim_tilde_tau_q}, Assumption \ref{asm:bal_exp}, and Point \ref{asm:1b_p2} of Assumption \ref{asm:1b},

    \begin{gather}
        \widehat{\tau} = \frac{1}{P} \sum_{p} \widehat{\tau}_p \xrightarrow{\mathbb{P}} \underset{P\rightarrow+\infty}{\lim}\frac{1}{P} \sum_{p} \EX[\widehat{\tau}_p] =\underset{P\rightarrow+\infty}{\lim}\frac{1}{P} \sum_{p} {\tau}_p = \underset{P\rightarrow+\infty}{\lim}\tau.
        \label{eq:consistency_hattau}
    \end{gather}

    }

{
    \subsubsection*{Point \ref{th:asym_p2}}~\\

    Note that by Point \ref{asm:1_p3} of Assumption \ref{asm:1}, $\widehat{\tau}-\tau= \widehat{\tau}-\EX[\widehat{\tau}]= \sum_p (\widehat{\tau}_p-\EX[\widehat{\tau}_p])/P$ is a sum of independent random variables $(\widehat{\tau}_p-\EX[\widehat{\tau}_p])_{p=1}^P$ with mean zero and with a finite variance by Equation \eqref{eq:cond_tau_p}}. As $\sum_{p=1}^P \EX[|\widehat{\tau}_p-\tau_p|^{2+\epsilon}/S_P^{2+\epsilon}] \rightarrow 0 $ for some $\epsilon>0$ (by Point \ref{asm:1b_p3} of Assumption \ref{asm:1b}), then, by the Lyapunov central limit theorem, $(\widehat{\tau}-\tau)/(S_P/P) = \sum_p (\widehat{\tau}_p-\tau_p)/S_P  ~{\overset{d}{\longrightarrow}}~  \mathcal{N}(0,1)$ as $P\rightarrow +\infty$, where $S_P^2 = \sum_{p=1}^P \mathbb{V}(\widehat{\tau}_p) = P^2 \mathbb{V}(\widehat{\tau})$.
    Therefore,
\begin{equation}\label{eq:as_normality}
(\widehat{\tau}-\tau)/\sqrt{\mathbb{V}(\widehat{\tau})}   ~{\overset{d}{\longrightarrow}}~  \mathcal{N}(0,1).
\end{equation}
Then,
   \begin{align}
        P\widehat{\mathbb{V}}_{pair}(\widehat{\tau}) - P\mathbb{V}(\widehat{\tau})
        & = \sum_{p=1}^P \frac{\widehat{\tau}_p^2}{P} - \widehat{\tau}^2 - \sum_{p=1}^P \frac{\mathbb{V}(\widehat{\tau}_p)}{P}
        \notag
        \\
        & = \sum_{p=1}^P \frac{\widehat{\tau}_p^2}{P} - \widehat{\tau}^2 - \sum_{p=1}^P \frac{\EX[\widehat{\tau}_p^2]-\EX[\widehat{\tau}_p]^2}{P}
        \notag
        \\
        & = \sum_{p=1}^P \frac{\widehat{\tau}_p^2-\EX[\widehat{\tau}_p^2]}{P} - \widehat{\tau}^2 + \sum_{p=1}^P \frac{\tau_p^2}{P}
        \label{eq:plim_vpair1} \\
        & ~{\overset{\mathbb{P}}{\longrightarrow}}~ \underset{P\rightarrow +\infty}{\lim}\frac{1}{P}\sum_{p=1}^P (\tau_p-\tau)^2.
        \label{eq:plim_vpair2}
    \end{align}

    {
    The first equality follows from Equations \eqref{eq:var_hat_tau} and \eqref{eq:PCVE_simple}. The third equality follows from $\EX[\widehat{\tau}_{p}] = \tau_p$.
    Let's consider each of the terms in Equation \eqref{eq:plim_vpair1}.
    $\sum_{p=1}^P \frac{\widehat{\tau}_p^2-\EX[\widehat{\tau}_p^2]}{P} ~{\overset{\mathbb{P}}{\longrightarrow}}~ 0$ by Lemma \ref{le:plim_tilde_tau_q}. Then, $\widehat{\tau}^2~{\overset{\mathbb{P}}{\longrightarrow}}~ \underset{P\rightarrow +\infty}{\lim}\tau^2$ by Equation \eqref{eq:consistency_hattau} and the continuous mapping theorem (CMT).
    Equation \eqref{eq:plim_vpair2} follows from these facts, and from Point \ref{asm:1b_p2} of Assumption \ref{asm:1b}.

    Given Equation \eqref{eq:plim_vpair2}, Point \ref{asm:1b_p2} of Assumption \ref{asm:1b}, the Slutsky Lemma and the CMT,  as $P \rightarrow +\infty$,}
    \begin{equation}\label{eq:as_norm_pair}
    \frac{\widehat{\tau}-\tau}{\sqrt{\widehat{\mathbb{V}}_{pair}(\widehat{\tau})}} = \frac{\widehat{\tau}-\tau}{\sqrt{\mathbb{V}(\widehat{\tau})}} \sqrt{\frac{P\mathbb{V}(\widehat{\tau})}{P\widehat{\mathbb{V}}_{pair}(\widehat{\tau})}}   ~{\overset{d}{\longrightarrow}}~  \mathcal{N}(0,\sigma_{pair}^2).
    \end{equation}
    Finally, by Lemma \ref{le:4.1}, $\widehat{\mathbb{V}}_{pair}(\widehat{\tau}) = \widehat{\mathbb{V}}_{pair}(\widehat{\tau}_{fe})$, and by Assumption \ref{asm:bal_exp}, $\widehat{\tau} = \widehat{\tau}_{fe}$.

\noindent \textbf{QED}.


    \subsubsection*{Point \ref{th:asym_p3}}~\\
    
    By Lemma \ref{le:4.1.3}, $\widehat{\mathbb{V}}_{pair}(\widehat{\tau}) = 2\widehat{\mathbb{V}}_{unit}(\widehat{\tau}_{fe})$, so given Point \ref{th:asym_p2} of this theorem, the result follows.

\noindent \textbf{QED}.


    {
    \subsubsection*{Point \ref{th:asym_p4}}~\\
    \begin{align}
        &P\widehat{\mathbb{V}}_{unit}(\widehat{\tau}) - P\widehat{\mathbb{V}}_{pair}(\widehat{\tau})\notag\\
        & =  \frac{2}{P} \sum_p \widehat{Y}_p(1)\widehat{Y}_p(0) - 2\frac{1}{P}\sum_p\widehat{Y}_p(1)\frac{1}{P}\sum_p\widehat{Y}_p(0)
        \notag
        \\
        & ~{\overset{\mathbb{P}}{\longrightarrow}}~ 2\underset{P\rightarrow+\infty}{\lim} \left\{\frac{1}{P} \sum_p \EX[\widehat{Y}_p(1)\widehat{Y}_p(0)] - \EX[\widehat{Y}(1)]\EX[\widehat{Y}(0)]\right\}
        \notag
        \\
        & = 2\underset{P\rightarrow+\infty}{\lim} \frac{1}{P}\sum_p\bigg\{\left(\overline{y}_{p}(0)- \overline{y}(0)\right)\left(\overline{y}_{p}(1)- \overline{y}(1)\right) -\frac{1}{2}\sum_g(\overline{y}_{gp}(0)-\overline{y}_{p}(0))(\overline{y}_{gp}(1)-\overline{y}_{p}(1))\bigg\}.
        \label{eq:lim_ucve}
    \end{align}
    The first equality follows from Equation \eqref{eq:diff_ucve_pcve}.
The convergence arrow follows from the fact $\EX\left[\left|\widehat{Y}_p(1)\widehat{Y}_p(0)\right|^{1+\epsilon/2}\right]$ is bounded uniformly in $p$ by Equation \eqref{eq:cond_haty1_p} and the Cauchy-Schwarz inequality, from the fact that $\EX\left[\left|\widehat{Y}_p(d)\right|^{1+\epsilon/2}\right]$ is also bounded uniformly in $p$, from Point \ref{asm:1_p3} of Assumption \ref{asm:1}, from the SLLN in Lemma 1 in \cite{liu1988bootstrap}, from the CMT, and from Point \ref{asm:1b_p2} of Assumption \ref{asm:1b}. The last equality follows from the same steps as those used to prove Lemma \ref{le:4.1.3}. The result follows from Equations \eqref{eq:lim_ucve}, \eqref{eq:plim_vpair2}, and \eqref{eq:as_normality}, and a reasoning similar to that used to prove Equation \eqref{eq:as_norm_pair}.

\noindent \textbf{QED}.

}


\subsection{Proof of Lemma \ref{le:CRVE_nfe}}

\subsubsection*{Point \ref{le:CRVE_nfe_p1}}~\\
    First, we introduce the formulas for the PCVE and UCVE in a general linear regression. Let $\epsilon_{igp}$ be the residual from the regression of $Y_{igp}$ on a $K$-vector of covariates $\boldsymbol{X}_{igp}$, and $\boldsymbol{X}$ the $(n\times K)$ matrix whose rows are $\boldsymbol{X}_{igp}'$.
    The PCVE of the OLS estimator, $\widehat{\boldsymbol{\beta}}$, is defined as follows (\cite{liang1986longitudinal}, \cite{abadie2017should})
        \begin{gather}
            \widehat{\mathbb{V}}_{pair}(\widehat{\boldsymbol{\beta}}) = (\boldsymbol{X}'\boldsymbol{X})^{-1}
                \left(
                \sum_{p=1}^P
                \left(\sum_{g=1}^2 \sum_{i=1}^{n_{gp}} \epsilon_{igp}\boldsymbol{X}_{igp}\right)
                \left(\sum_{g=1}^2\sum_{i=1}^{n_{gp}} \epsilon_{igp}\boldsymbol{X}_{igp}\right)'
                \right)
                (\boldsymbol{X}'\boldsymbol{X})^{-1}.
        \label{eq:VEs_p1}
        \end{gather}
    The UCVE of the OLS estimator, $\widehat{\boldsymbol{\beta}}$, is defined as follows
        \begin{gather}
        \widehat{\mathbb{V}}_{unit}(\widehat{\boldsymbol{\beta}}) = (\boldsymbol{X}'\boldsymbol{X})^{-1}
        \left(
        \sum_{p=1}^P \sum_{g=1}^2
         \left( \sum_{i=1}^{n_{gp}} \epsilon_{igp}\boldsymbol{X}_{igp}\right)
        \left(\sum_{i=1}^{n_{gp}} \epsilon_{igp}\boldsymbol{X}_{igp}\right)'
        \right)
        (\boldsymbol{X}'\boldsymbol{X})^{-1}.
        \label{eq:VEs_p2}
        \end{gather}

Subtract from Equation \eqref{eq:regnfe} the average outcome in the population $\overline{Y} \equiv \frac{1}{n}\sum_p \sum_g \sum_i Y_{igp} = \widehat{\alpha} + \widehat{\tau}\overline{W} + \overline{\epsilon}$, where $\overline{W} \equiv \frac{1}{n}\sum_p \sum_g \sum_i W_{gp}$, and  $\overline{\epsilon} \equiv \frac{1}{n}\sum_p \sum_g \sum_i \epsilon_{igp} = 0$ by construction. Then,
\begin{gather}
    Y_{igp} - \overline{Y} = \widehat{\tau}(W_{gp} - \overline{W}) + \epsilon_{igp}.
    \label{eq:regnfe_demean}
\end{gather}
Apply Equation \eqref{eq:VEs_p1} to the residuals and covariates of the regression defined by Equation \eqref{eq:regnfe_demean}.\footnote{The clustered variance estimators of $\widehat{\tau}$ in the demeaned regression in Equation \eqref{eq:regnfe_demean} and in the regression with an intercept in Equation \eqref{eq:regnfe} are equal \citep{cameron2015practitioner}.} Then,
\begin{align}
    \widehat{\mathbb{V}}_{pair}(\widehat{\tau}) & = \frac{\sum_p \left[\sum_g (W_{gp}-\overline{W})\sum_i \epsilon_{igp}\right]^2}{\left[\sum_p \sum_g \sum_i (W_{gp}-\overline{W})^2\right]^2}.
    \label{eq:pcve_vhat_htau}
\end{align}
The numerator of $\widehat{\mathbb{V}}_{pair}(\widehat{\tau})$ equals
\begin{align}
    \sum_p \left[\sum_g (W_{gp}-\overline{W})\sum_i \epsilon_{igp}\right]^2
    & = \sum_p \left[(1-\overline{W})SET_p -\overline{W} SEU_p\right]^2
    \notag
    \\
    & = \sum_p \left[\frac{C}{n}SET_p -\frac{T}{n} SEU_p\right]^2.
    \label{eq:num_pcve_htau}
\end{align}
The first equality follows from the definition of $SET_p$ and $SEU_p$. The second equality follows from the definition of $T$ and $C$.
\\
The denominator of $\widehat{\mathbb{V}}_{pair}(\widehat{\tau})$ equals
\begin{align}
    \left[\sum_p \sum_g \sum_i (W_{gp}-\overline{W})^2\right]^2 & = \left[\sum_p \sum_g (W_{gp}-\overline{W})^2 n_{gp}\right]^2
    \notag
    \\
    & = \left[(1-\overline{W})^2 \sum_p T_p +\overline{W}^2 \sum_p C_p  \right]^2
    \notag
    \\
    & = \left[\frac{C^2}{n^2} T +\frac{T^2}{n^2} C  \right]^2
    \notag
    \\
    & = \left[\frac{CT}{n}\right]^2.
    \label{eq:den_pcve_htau}
\end{align}
The first equality follows from $(W_{gp} - \overline{W})$ being constant across units. The second equality follows from the definition of $T_p$ and $C_p$. The third equality follows from the definition of $T$ and $C$.
\\
Then, combining Equations \eqref{eq:pcve_vhat_htau}, \eqref{eq:num_pcve_htau} and \eqref{eq:den_pcve_htau},
\begin{align*}
    \widehat{\mathbb{V}}_{pair}(\widehat{\tau}) & = \frac{\sum_p \left[\frac{C}{n}SET_p -\frac{T}{n} SEU_p\right]^2}{\left[\frac{CT}{n}\right]^2}
    \\
    & = \sum_p \left[\frac{SET_p}{T} -\frac{SEU_p}{C}\right]^2.
\end{align*}
\noindent \textbf{QED}.


\subsubsection*{Point \ref{le:CRVE_nfe_p2}}~\\

Apply Equation \eqref{eq:VEs_p2} to the residuals and covariates of the regression defined by Equation \eqref{eq:regnfe_demean}. Then,
\begin{align}
    \widehat{\mathbb{V}}_{unit}(\widehat{\tau}) & = \frac{\sum_p \sum_g \left[(W_{gp}-\overline{W})\sum_i \epsilon_{igp}\right]^2}{\left[\sum_p \sum_g \sum_i (W_{gp}-\overline{W})^2\right]^2}.
    \label{eq:UCVE_vhat_htau}
\end{align}
The numerator of $\widehat{\mathbb{V}}_{unit}(\widehat{\tau})$ equals
\begin{align}
    \sum_p \sum_g \left[(W_{gp}-\overline{W})\sum_i \epsilon_{igp}\right]^2
    & = \sum_p \sum_g (W_{gp}-\overline{W})^2\left(\sum_i \epsilon_{igp}\right)^2
    \notag
    \\
    & = \sum_p \left[(1-\overline{W})^2 SET_p^2 + \overline{W}^2 SEU_p^2 \right]
    \notag
    \\
    & = \sum_p \left[\frac{C^2}{n^2} SET_p^2 + \frac{T^2}{n^2} SEU_p^2 \right].
    \label{eq:num_UCVE_htau}
\end{align}
The second equality follows from the definition of $SET_p$ and $SEU_p$. The third equality follows from the definition of $T$ and $C$.
Then, combining Equations \eqref{eq:den_pcve_htau}, \eqref{eq:UCVE_vhat_htau} and \eqref{eq:num_UCVE_htau},
\begin{align*}
    \widehat{\mathbb{V}}_{unit}(\widehat{\tau}) & = \frac{ \sum_p \left[\frac{C^2}{n^2} SET_p^2 + \frac{T^2}{n^2} SEU_p^2 \right]}{\left[\frac{CT}{n}\right]^2}
    \\
    & = \sum_p \left[\frac{SET_p^2}{T^2}  + \frac{SEU_p^2}{C^2}  \right].
\end{align*}
\noindent \textbf{QED}.


\subsubsection*{Point \ref{le:CRVE_fe_p1}}~\\

Let $SET_{p,fe}= \sum_{g = 1}^2\sum_{i=1}^{n_{gp}} W_{gp} {u}_{igp}$ and $SEU_{p,fe} = \sum_{g = 1}^2\sum_{i=1}^{n_{gp}} (1-W_{gp}){u}_{igp}$ respectively be the sum of the residuals $u_{igp}$ for the treated and untreated observations in pair $p$. Averaging Equation \eqref{eq:regfe} across units in pair $p$,
    \begin{gather}
        \overline{Y}_{p}= \widehat{\tau}_{fe}\overline{W}_{p} + \widehat{\gamma}_p + \overline{u}_{p} \label{eq:fe_devmean1},
        \intertext{where
        $\overline{Y}_{p} = \frac{1}{n_{p}}\sum_{g=1}^2 \sum_{i=1}^{n_{gp}} Y_{igp}$, $\overline{W}_{p}= \frac{1}{n_{p}}\sum_{g=1}^2 \sum_{i=1}^{n_{gp}}  W_{gp} = \frac{1}{n_{p}}\sum_{g=1}^2 W_{gp}n_{gp} = \frac{T_p}{n_p}$, and $\overline{u}_{p} = \frac{1}{n_{p}}\sum_{g=1}^2 \sum_{i=1}^{n_{gp}} u_{igp}$.
        Substracting Equation \eqref{eq:fe_devmean1} from Equation \eqref{eq:regfe},}
        Y_{igp} - \overline{Y}_{p}= \widehat{\tau}_{fe} (W_{gp}-\overline{W}_{p}) + u_{igp} - \overline{u}_{p}.
        \label{eq:fe_devmean2}
    \end{gather}
    $\{u_{ijp'}\}$ is orthogonal to the pair-$p$ fixed effect indicator $\{\delta_{igp}\}$, so
    \begin{align}
        &\sum_{p'=1}^P \sum_{g=1}^2 \sum_{i=1}^{n_{jp'}} u_{ijp'}\delta_{igp} = 0 \nonumber \\
        \Leftrightarrow &\sum_{g=1}^2 \sum_{i=1}^{n_{gp}} u_{igp} = 0,
        \label{eq:sum_res_fe}
    \end{align}
    where the equivalence holds because $\delta_{igp} = 1$ if and only if observation $i$ belongs to pair $p$. This implies that for all $p$ $\overline{u}_{p}=0$. Equation \eqref{eq:fe_devmean2} then becomes a regression with one covariate and the same residuals as in Equation \eqref{eq:regfe}:
\begin{align}
        &Y_{igp} - \overline{Y}_{p}= \widehat{\tau}_{fe} (W_{gp}-\overline{W}_{p}) + u_{igp}. \label{eq:fe_devmean4}
\end{align}
Now, it follows from Equations \eqref{eq:VEs_p1} and \eqref{eq:fe_devmean4} that\footnote{The clustered variance estimators of $\widehat{\tau}_{fe}$ in the regression residualized from the pair fixed effects in Equation \eqref{eq:fe_devmean4} and in the regression with pair fixed effects in Equation \eqref{eq:regfe} are equal \citep{cameron2015practitioner}.}
    \begin{gather}
        \widehat{\mathbb{V}}_{pair}(\widehat{\tau}_{fe}) =
        \frac{\left[\sum_{p=1}^P \left(\sum_{g=1}^2 \sum_{i=1}^{n_{gp}} u_{igp}(W_{gp}-\overline{W}_{p}) \right)^2\right]}{\left(\sum_{p=1}^P \sum_{g=1}^2 \sum_{i=1}^{n_{gp}} (W_{gp}-\overline{W}_{p})^2  \right)^2}.
            \label{eq:vhat_pair_fe}
    \end{gather}
    The denominator of $\widehat{\mathbb{V}}_{pair}(\widehat{\tau}_{fe})$ equals
    \begin{align}
            \left[\sum_p\sum_g \sum_i (W_{gp}-\overline{W}_{p})^2\right]^2 & = \left[\sum_p\sum_g (W_{gp}-\overline{W}_{p})^2 n_{gp}\right]^2
            \notag
            \\
            & = \left[\sum_p[T_p(1-\overline{W}_p)^2+C_p\overline{W}_p^2]\right]^2
            \notag
            \\
            & = \left[\sum_p\left(T_p\frac{C_p^2}{n_p^2}+C_p\frac{T_p^2}{n_p^2}\right)\right]^2
            \notag
            \\
            & = \left[\sum_p \frac{T_pC_p}{n_p}\right]^2
            \notag
            \\
            & = \left[\sum_p (n_{1p}^{-1}+n_{2p}^{-1})^{-1}\right]^2.
            \label{eq:fe_den}
    \end{align}
    The numerator of $\widehat{\mathbb{V}}_{pair}(\widehat{\tau}_{fe})$ is equal to
    \begin{align}
        \sum_{p=1}^P \left(\sum_{g=1}^2 \sum_{i=1}^{n_{gp}} u_{igp}(W_{gp}-\overline{W}_{p}) \right)^2
            & = \sum_{p=1}^P \left(\sum_{g=1}^2 (W_{gp}-\overline{W}_{p})\sum_{i=1}^{n_{gp}} u_{igp} \right)^2
        \notag
        \\
        & = \sum_{p=1}^P \left(-\overline{W}_{p}(SET_{p,fe}+SEU_{p,fe})+SET_{p,fe}\right)^2
        \notag
        \\
        & = \sum_{p=1}^P \left(SET_{p,fe}\right)^2,
        \label{eq:fe_num}
    \end{align}
    where $SET_{p,fe}+SEU_{p,fe} =  \sum_{g=1}^2 \sum_{i=1}^{n_{gp}} u_{igp} = 0$ from Equation \eqref{eq:sum_res_fe}. Finally,
\begin{align}
	SET_{p,fe}
	& = \sum_{g,i} W_{gp}[Y_{igp}-\widehat{\gamma}_p-\widehat{\tau}_{fe}W_{gp}]
	\notag
	\\
	& = \sum_{g,i} W_{gp}Y_{igp}-(\widehat{\gamma}_p+\widehat{\tau}_{fe}) \sum_{g,i} W_{gp}
	\notag
	\\
	& = \sum_{g,i} W_{gp}Y_{igp}-(\overline{Y}_p-\widehat{\tau}_{fe}\overline{W}_p+\widehat{\tau}_{fe}) \overline{W}_pn_p
	\notag
	\\
	& = \sum_{g,i} W_{gp}Y_{igp}-\overline{W}_p\sum_{g,i}Y_{igp}-(1-\overline{W}_p)\widehat{\tau}_{fe}\overline{W}_pn_p
	\notag
	\\	
	& = \sum_{g,i} W_{gp}Y_{igp}-\overline{W}_p\left[\sum_{g,i}W_{gp}Y_{igp}+\sum_{g,i}(1-W_{gp})Y_{igp}\right]-(1-\overline{W}_p)\overline{W}_pn_p\widehat{\tau}_{fe}
	\notag
	\\	
	& = (1-\overline{W}_p)\sum_{g,i} W_{gp}Y_{igp}-\overline{W}_p\sum_{g,i}(1-W_{gp})Y_{igp}-(1-\overline{W}_p)\overline{W}_pn_p\widehat{\tau}_{fe}
	\notag	
	\\
	& = (1-\overline{W}_p)\overline{W}_pn_p\left(\frac{\sum_{g,i} W_{gp}Y_{igp}}{\overline{W}_pn_p}-\frac{\sum_{g,i}(1-W_{gp})Y_{igp}}{(1-\overline{W}_p)n_p}-\widehat{\tau}_{fe}\right)
	\notag	
	\\
	& = \frac{n_{1p}n_{2p}}{n_p^2}n_p\left(\frac{\sum_{g,i} W_{gp}Y_{igp}}{\sum_{g,i} W_{gp}}-\frac{\sum_{g,i}(1-W_{gp})Y_{igp}}{\sum_{g,i} (1-W_{gp})}-\widehat{\tau}_{fe}\right)
	\notag
	\\
	& = \frac{n_{1p}n_{2p}}{n_{1p}+n_{2p}}(\widehat{\tau}_p-\widehat{\tau}_{fe}).
	\label{eq:set_pfe_simple}
\end{align}
The first equality follows from the definition of $SET_{p,fe}$. The second equality follows from the definition of $u_{igp}$ in Equation \eqref{eq:regfe}. The third equality follows from the definition of $\overline{W}_p$ and Equations \eqref{eq:fe_devmean1} and \eqref{eq:sum_res_fe}. The ninth equality follows from the definition of $\widehat{\tau}_p$.

     Therefore, combining Equations \eqref{eq:vhat_pair_fe}, \eqref{eq:fe_den}, \eqref{eq:fe_num} and \eqref{eq:set_pfe_simple},
    \begin{align*}
        \widehat{\mathbb{V}}_{pair}(\widehat{\tau}_{fe})  &
         = \sum_{p=1}^P \omega_p^2 (\widehat{\tau}_p-\widehat{\tau})^2
    \end{align*}
\noindent \textbf{QED}.


\subsubsection*{Point \ref{le:CRVE_fe_p2}}~\\

Applying the definition of the UCVE from Equation \eqref{eq:VEs_p2} to the regression in Equation \eqref{eq:fe_devmean4},
    \begin{gather}
        \widehat{\mathbb{V}}_{unit}(\widehat{\tau}_{fe}) =
                \frac{\left[\sum_{p=1}^P \sum_{g=1}^2 \left(\sum_{i=1}^{n_{gp}} u_{igp}(W_{gp}-\overline{W}_{p}) \right)^2\right]}{\left(\sum_{p=1}^P \sum_{g=1}^2 \sum_{i=1}^{n_{gp}} (W_{gp}-\overline{W}_{p})^2  \right)^2}.
                \label{eq:GCRVE_fe1}
    \end{gather}
    The numerator of $\widehat{\mathbb{V}}_{unit}(\widehat{\tau}_{fe})$  equals
    \begin{align}
        \sum_{p=1}^P \sum_{g=1}^2 \left(\sum_{i=1}^{n_{gp}} u_{igp}(W_{gp}-\overline{W}_{p}) \right)^2
        & = \sum_{p=1}^P \sum_{g=1}^2 (W_{gp}-\overline{W}_{p})^2 \left(\sum_{i=1}^{n_{gp}} u_{igp} \right)^2
        \notag
        \\
        & = \sum_{p=1}^P \left( (1-\overline{W}_{p})^2 SET_{p,fe}^2 + \overline{W}_{p}^2 SEU_{p,fe}^2 \right)
        \notag
        \\
        & = \sum_{p=1}^P SET_{p,fe}^2\left( \frac{C_p^2}{n_p^2}  + \frac{T_p^2}{n_p^2} \right)
        \notag
        \\
        & = \sum_{p=1}^P \frac{C_p^2T_p^2}{n_p^2} SET_{p,fe}^2 \left(  \frac{1}{T_p^2} + \frac{1}{C_p^2} \right)
        \notag
        \\
        & = \sum_{p=1}^P (n_{1p}^{-1}+n_{2p}^{-1})^{-2} SET_{p,fe}^2 \left(\frac{1}{n_{1p}^2} + \frac{1}{n_{2p}^2} \right).
        \label{eq:GCRVE_fe2}
    \end{align}
    The second equality follows from the definitions of $SET_{p,fe}$ and $SEU_{p,fe}$. The third equality follows from Equation \eqref{eq:sum_res_fe}, i.e., $SET_{p,fe} + SEU_{p,fe} = \sum_g \sum_i u_{igp} = 0$, for all $p$, so $SET_{p,fe}^2 = SEU_{p,fe}^2$, and the definitions of $T_p$ and $C_p$. Finally, combining Equations \eqref{eq:fe_den}, \eqref{eq:set_pfe_simple}, \eqref{eq:GCRVE_fe1} and \eqref{eq:GCRVE_fe2},
    \begin{align*}
        \widehat{\mathbb{V}}_{unit}(\widehat{\tau}_{fe})
            &= \sum_{p=1}^P \omega_p^2 (\widehat{\tau}_p-\widehat{\tau})^2\left( \left(\frac{n_{1p}}{n_p}\right)^2 + \left(\frac{n_{2p}}{n_p}\right)^2 \right).
    \end{align*}

\noindent \textbf{QED}.



\subsection{Proof of Lemma \ref{le:4.1_pop}}

    \subsubsection*{Point \ref{le:pop_p1}}~\\

    \begin{align}
    	\widehat{\mathbb{V}}_{pop}(\widehat{\tau}) & = \frac{1}{P^2} \sum_{r=1}^R (\widehat{\tau}_{1r}-\widehat{\tau}_{2r})^2,
    	\notag
    	\\
    	& = \frac{1}{P^2} \sum_{r=1}^R (\widehat{\tau}_{1r}^2+\widehat{\tau}_{2r}^2-2\widehat{\tau}_{1r}\widehat{\tau}_{2r}).
    	\notag
    	\intertext{Taking expected value,}
    	\EX[\widehat{\mathbb{V}}_{pop}(\widehat{\tau})] & =  \frac{1}{P^2} \sum_{r=1}^R \EX(\widehat{\tau}_{1r}^2+\widehat{\tau}_{2r}^2-2\widehat{\tau}_{1r}\widehat{\tau}_{2r}),
    	\notag
    	\\
    	& = \frac{1}{P^2} \sum_{r=1}^R (\mathbb{V}(\widehat{\tau}_{1r})+\mathbb{V}(\widehat{\tau}_{2r})+\tau_{1r}^2+\tau_{2r}^2-2\tau_{1r}\tau_{2r}),
    	\notag
    	\\
    	& = \frac{1}{P^2} \sum_{p=1}^P \mathbb{V}(\widehat{\tau}_{p}) + \frac{1}{P^2} \sum_{r=1}^R (\tau_{1r}-\tau_{2r})^2,
    	\notag
    	\\
    	& = \mathbb{V}(\widehat{\tau}) +  \frac{1}{P^2} \sum_{r=1}^R (\tau_{1r}-\tau_{2r})^2.
    	\label{eq:exp_varpop}
    \end{align}
The second equality follows from properties of the variance and that $\EX[\widehat{\tau}_{1r}]=\tau_{1r}$ and $\EX[\widehat{\tau}_{2r}]=\tau_{2r}$. The third equality follows from $P=2R$. The fourth equality follows from Equation \eqref{eq:var_hat_tau}.
\noindent \textbf{QED}.

    \subsubsection*{Point \ref{le:brs_p1}}~\\

\begin{align*}
    \widehat{\mathbb{V}}_{brs}(\widehat{\tau}) & = \frac{1}{P^2} \sum_p \widehat{\tau}^2_p - \frac{1}{2} \left(\frac{2}{P^2}\sum_r\widehat{\tau}_{1r}\widehat{\tau}_{2r}+\frac{\widehat{\tau}^2}{P}\right).
    \\
    & = \frac{1}{2P^2} \sum_p (\widehat{\tau}_p-\widehat{\tau})^2 +  \frac{1}{2P^2}\sum_r(\widehat{\tau}_{1r}^2+\widehat{\tau}_{2r}^2 -2\widehat{\tau}_{1r}\widehat{\tau}_{2r}).
    \\
    & = \frac{1}{2}\widehat{\mathbb{V}}_{pair}(\widehat{\tau})+\frac{1}{2}\widehat{\mathbb{V}}_{pop}(\widehat{\tau}).
\end{align*}
\noindent \textbf{QED}.

    \subsubsection*{Point \ref{le:pcve_versus_brs}}~\\

	\begin{gather*}
		\EX[\widehat{\mathbb{V}}_{pop}(\widehat{\tau})]  \leq \EX\left[\frac{P}{P-1}\widehat{\mathbb{V}}_{pair}(\widehat{\tau})\right],
		\\ \Leftrightarrow
		 (2R-1)\sum_{r=1}^R (\tau_{1r}-\tau_{2r})^2 \leq 2R\sum_{p=1}^P (\tau_p - \tau)^2,
		\\ \Leftrightarrow		
		(2R-1)\sum_{r=1}^R (\tau_{1r}^2+\tau_{2r}^2-2\tau_{1r}\tau_{2r}) \leq 2R\sum_{r=1}^R [\tau_{1r}^2 - 2\tau_{1r}\tau+\tau^2+\tau_{2r}^2- 2\tau_{2r}\tau+\tau^2],
		\\ \Leftrightarrow
		0\leq \sum_{r=1}^R (\tau_{1r}-\tau_{2r})^2 + 2R \sum_{r=1}^R [2\tau_{1r}\tau_{2r}- 2(\tau_{1r}+\tau_{2r})\tau+2\tau^2],	
		\\ \Leftrightarrow
		0\leq \sum_{r=1}^R (\tau_{1r}-\tau_{2r})^2 + 4R\sum_{r=1}^R (\tau_{1r}-\tau)(\tau_{2r}-\tau).
	\intertext{The second inequality follows from Points \ref{le:4.1.1} and \ref{le:pop_p1} of this lemma. Let $\tau_{\cdot r} = \frac{1}{2}(\tau_{1r}+\tau_{2r})$. Then,}	
		\EX[\widehat{\mathbb{V}}_{pop}(\widehat{\tau})]  \leq \EX\left[\frac{P}{P-1}\widehat{\mathbb{V}}_{pair}(\widehat{\tau})\right],
		\\
		\Leftrightarrow 0\leq \sum_{r=1}^R \sum_{p=1,2} 2 (\tau_{pr}-\tau_{\cdot r})^2 + 4R\sum_{r=1}^R (\tau_{1r}-\tau_{\cdot r}+\tau_{\cdot r}-\tau)(\tau_{2r}-\tau_{\cdot r}+\tau_{\cdot r}-\tau),
		\\
		\Leftrightarrow 0\leq \sum_{r=1}^R \sum_{p=1,2}\frac{1}{2} (\tau_{pr}-\tau_{\cdot r})^2 + R\sum_{r=1}^R [(\tau_{1r}-\tau_{\cdot r})(\tau_{2r}-\tau_{\cdot r})+(\tau_{\cdot r}-\tau)^2],
		\\
		\Leftrightarrow 0\leq \sum_{r=1}^R \sum_{p=1,2}\frac{1}{2} (\tau_{pr}-\tau_{\cdot r})^2 + R\sum_{r=1}^R \left[-\sum_{p=1,2}\frac{1}{2}(\tau_{pr}-\tau_{\cdot r})^2+(\tau_{\cdot r}-\tau)^2\right],
		\\
		\Leftrightarrow \frac{1}{R}\sum_{r=1}^R \sum_{p=1,2}\frac{1}{2} (\tau_{pr}-\tau_{\cdot r})^2 \leq \frac{1}{R-1}\sum_{r=1}^R (\tau_{\cdot r}-\tau)^2.
		\end{gather*}
This proves inequality a).

Then, if $\frac{1}{R}\sum_{r=1}^R \sum_{p=1,2}\frac{1}{2} (\tau_{pr}+\tau_{\cdot r})^2 \leq \frac{1}{R-1}\sum_{r=1}^R (\tau_{\cdot r}-\tau)^2$, it follows from Point \ref{le:brs_p1} of the lemma and the previous display that
\begin{align*}
\EX\left[\widehat{\mathbb{V}}_{pop}(\widehat{\tau})\right] \leq&\frac{1}{2}\EX\left[\widehat{\mathbb{V}}_{pop}(\widehat{\tau})\right] +\frac{1}{2}\EX\left[\frac{P}{P-1}\widehat{\mathbb{V}}_{pair}(\widehat{\tau})\right]\\
\leq&
\frac{1}{2}\EX\left[\frac{P}{P-1}\widehat{\mathbb{V}}_{pop}(\widehat{\tau})\right] +\frac{1}{2}\EX\left[\frac{P}{P-1}\widehat{\mathbb{V}}_{pair}(\widehat{\tau})\right]\\
=&\EX\left[\frac{P}{P-1}\widehat{\mathbb{V}}_{brs}(\widehat{\tau})\right],
\end{align*}
which proves inequality b).

Similarly,
if $\frac{1}{R}\sum_{r=1}^R \sum_{p=1,2}\frac{1}{2} (\tau_{pr}+\tau_{\cdot r})^2 \leq \frac{1}{R-1}\sum_{r=1}^R (\tau_{\cdot r}-\tau)^2$, it follows from Point \ref{le:brs_p1} of the lemma and the previous display that
\begin{align*}
\EX\left[\widehat{\mathbb{V}}_{brs}(\widehat{\tau})\right] \leq&\frac{1}{2}\EX\left[\widehat{\mathbb{V}}_{pop}(\widehat{\tau})\right] +\frac{1}{2}\EX\left[\frac{P}{P-1}\widehat{\mathbb{V}}_{pair}(\widehat{\tau})\right]\\
\leq&
\frac{1}{2}\EX\left[\frac{P}{P-1}\widehat{\mathbb{V}}_{pair}(\widehat{\tau})\right] +\frac{1}{2}\EX\left[\frac{P}{P-1}\widehat{\mathbb{V}}_{pair}(\widehat{\tau})\right]\\
=&\EX\left[\frac{P}{P-1}\widehat{\mathbb{V}}_{pair}(\widehat{\tau})\right],
\end{align*}
which proves inequality c).

\noindent \textbf{QED}.

\subsection{Proof of Theorem \ref{th:asym_app}}

     \subsubsection*{Point \ref{th:asym_p2prime}}~\\

\begin{align}
    P\widehat{\mathbb{V}}_{pop}(\widehat{\tau}) - P \mathbb{V}(\widehat{\tau})
    & = \frac{1}{P} \sum_{r=1}^R [\widehat{\tau}_{1r}^2-2\widehat{\tau}_{1r}\widehat{\tau}_{2r}+\widehat{\tau}_{2r}^2] - \frac{1}{P}\sum_{p=1}^P  \mathbb{V}(\widehat{\tau}_p)
    \notag
    \\
    & =  \frac{1}{P} \sum_{p=1}^P \widehat{\tau}_p^2 -\frac{2}{P} \sum_{r=1}^R \widehat{\tau}_{1r}\widehat{\tau}_{2r} - \frac{1}{P}\sum_{p=1}^P  [\EX(\widehat{\tau}_p^2) -  \tau_p^2]
    \notag
    \\
    & =  \sum_{p=1}^P \frac{\widehat{\tau}_p^2-\EX[\widehat{\tau}_p^2]}{P} -\frac{1}{R} \sum_{r=1}^R \widehat{\tau}_{1r}\widehat{\tau}_{2r} + \frac{1}{P}\sum_{r=1}^R  (\tau_{1r}^2+\tau_{2r}^2)
    \notag
    \\
    & \xrightarrow{\mathbb{P}}\underset{P\rightarrow+\infty}{\lim}\frac{1}{P} \sum_{r=1}^R ({\tau}_{1r}-{\tau}_{2r})^2.
    \label{eq:plim_vpop}
\end{align}
The second equality follows from the properties of the variance.  As $P \rightarrow +\infty$, by Lemma \ref{le:plim_tilde_tau_q}, $\sum_{p=1}^P \frac{\widehat{\tau}_p^2-\EX[\widehat{\tau}_p^2]}{P} ~{\overset{\mathbb{P}}{\longrightarrow}}~ 0$.
Likewise,  as $R=P/2\rightarrow +\infty$, by Lemma 1 in \cite{liu1988bootstrap}, $\sum_{r=1}^R \widehat{\tau}_{1r}\widehat{\tau}_{2r}/R-\sum_{r=1}^R \tau_{1r}\tau_{2r}/R ~{\overset{\mathbb{P}}{\longrightarrow}}~ 0$, because $\EX[|\widehat{\tau}_{1r}\widehat{\tau}_{2r}|^{1+\epsilon/2}]$ is uniformly bounded in $r$ by Equation \eqref{eq:cond_tau_p} and the Cauchy-Schwarz inequality, $(\widehat{\tau}_{1r}\widehat{\tau}_{2r})_{r=1}^{+\infty}$ is a sequence of independent random variables by Point \ref{asm:1_p3} of Assumption \ref{asm:1}, and $\EX(\widehat{\tau}_{1r}\widehat{\tau}_{2r})=\EX(\widehat{\tau}_{1r})\EX(\widehat{\tau}_{2r})=\tau_{1r}\tau_{2r}.$ Finally, the convergence arrow follows from Point \ref{asm:1b_p2} of Assumption \ref{asm:1b} and some algebra.

The result follows from Equations \eqref{eq:as_normality} and \eqref{eq:plim_vpop} and a reasoning similar to that used to prove Equation \eqref{eq:as_norm_pair}.

\noindent \textbf{QED}.

	\subsubsection*{Point \ref{th:asym_p2primeprime}}~\\
	
	\begin{align}	
		P\widehat{\mathbb{V}}_{bsr}(\widehat{\tau}) - P\mathbb{V}(\widehat{\tau}) & = 	\frac{1}{2}P(\widehat{\mathbb{V}}_{pair}(\widehat{\tau}) - \mathbb{V}(\widehat{\tau}))
		+ \frac{1}{2}P(\widehat{\mathbb{V}}_{pop}(\widehat{\tau}) - \mathbb{V}(\widehat{\tau}))
		\notag\\
		& \xrightarrow{\mathbb{P}} \frac{1}{2} \underset{P\rightarrow+\infty}{\lim}\frac{1}{P}\sum_{p=1}^P (\tau_p-\tau)^2
        +
        \frac{1}{2} \underset{P\rightarrow+\infty}{\lim}\frac{1}{P} \sum_{r=1}^R ({\tau}_{1r}-{\tau}_{2r})^2.
        \notag
	\end{align}
    The first equality follows from Point \ref{le:brs_p1} of Lemma \ref{le:4.1}. The convergence arrow follows from Equations \eqref{eq:plim_vpair2} and \eqref{eq:plim_vpop}. The result follows from the previous display, Equation \eqref{eq:as_normality}, and a reasoning similar to that used to prove Equation \eqref{eq:as_norm_pair}.

\noindent \textbf{QED}.

	
	    \subsubsection*{Point \ref{th:asym_p6.a}}~\\

	\begin{gather*}
		\sigma_{pair}^2  \leq \sigma_{pop}^2,
		\\ \Leftrightarrow
		\lim_{P\rightarrow +\infty} \frac{1}{R} \sum_{r=1}^R (\tau_{1r}-\tau_{2r})^2 \leq \lim_{P\rightarrow +\infty} \frac{1}{R} \sum_{p=1}^P (\tau_p - \tau)^2,
		\\ \Leftrightarrow		
		\lim_{P\rightarrow +\infty} \frac{1}{R} \sum_{r=1}^R (\tau_{1r}^2+\tau_{2r}^2-2\tau_{1r}\tau_{2r}) \leq \lim_{P\rightarrow +\infty} \frac{1}{R} \sum_{r=1}^R [\tau_{1r}^2 +\tau_{2r}^2- 2(\tau_{1r}+\tau_{2r})\tau+2\tau^2],
		\\ \Leftrightarrow
		0\leq \lim_{P\rightarrow +\infty} \frac{1}{R} \sum_{r=1}^R [2\tau_{1r}\tau_{2r}- 2(\tau_{1r}+\tau_{2r})\tau+2\tau^2],	
		\\ \Leftrightarrow
		0\leq \lim_{P\rightarrow +\infty} \frac{1}{R} \sum_{r=1}^R (\tau_{1r}-\tau)(\tau_{2r}-\tau).
	\end{gather*}
Then, 	
    $\sigma_{pair}^2  \leq \sigma_{bsr}^2\leq \sigma_{pop}^2 \Leftrightarrow \sigma_{pair}^2  \leq \sigma_{pop}^2$.

Point \ref{th:asym_p6.b} is straightforward so we do not prove it.

\noindent \textbf{QED}.
	

\subsection{Proof of Lemma \ref{le:CRVE_nfe_ext}}

Let $e_{igp}$ be the residual from the weighted least squares regression.
One has
\begin{align*}
Y_{igp}= \widetilde{\alpha} + \widetilde{\tau}W_{gp}+e_{igp}.
\end{align*}
Let $\widetilde{Y}=\frac{1}{n} \sum_{i,g,p} V_{gp}Y_{igp}$.
The previous display implies that
\begin{gather*}
    \widetilde{Y}=
    \widetilde{\alpha} \sum_{i,g,p} \frac{V_{gp}}{n} + \widetilde{\tau}\frac{1}{n} \sum_{i,g,p} V_{gp}W_{gp}+\frac{1}{n} \sum_{i,g,p} V_{gp}e_{igp}
    \nonumber\\
    = 2\widetilde{\alpha} + \widetilde{\tau},
\end{gather*}
where the second equality follows from $\frac{1}{n} \sum_{i,g,p} V_{gp}e_{igp}=0$, by the first-order condition attached to $\widetilde{\alpha}$ in the weighted OLS minimization problem.
Then, combining the two preceding displays implies that
\begin{gather}
    Y_{igp} - \frac{1}{2}\widetilde{Y} = \widetilde{\tau}\left(W_{gp} - \frac{1}{2}\right) + e_{igp}.
    \label{eq:wls}
\end{gather}

The next step is to compute the clustered variance estimators for the weighted least squares estimator. To do so, we apply Equation (15) in \cite{cameron2015practitioner} to the residuals and covariates of the regression defined by Equation \eqref{eq:wls}. This equation implies that
\begin{align}
    \widehat{\mathbb{V}}_{pair}(\widetilde{\tau}) & = \frac{\sum_p \left[\sum_g V_{gp}(W_{gp}-\frac{1}{2})\sum_i e_{igp}\right]^2}{\left[\sum_p \sum_g \sum_i V_{gp}(W_{gp}-\frac{1}{2})^2\right]^2}.
    \label{eq:pcve_vhat_tilde_tau}
\end{align}

Let $\widehat{Y}_{igp}=\widetilde{\alpha}+W_{gp}\widetilde{\tau}$, $\widehat{Y}(0)=\widetilde{\alpha}$, and $\widehat{Y}(1)=\widetilde{\alpha}+\widetilde{\tau}$. Note that
\begin{align}
    \sum_{i,g} W_{gp}\frac{e_{igp} }{n_{gp}}
    & = \sum_{i,g} W_{gp}(Y_{igp} - \widehat{Y}_{igp}) / n_{gp}
    \notag
    \\
    & = \sum_g W_{gp}\overline{y}_{gp}(1) - \widehat{Y}(1)\sum_g W_{gp}
    \notag
    \\
    & =  \widehat{Y}_p(1) - \sum_{p'}\frac{n_{p'}}{n}\widehat{Y}_{p'}(1)
    \label{eq:tilde_y_1}
\end{align}
The second equality follows from $W_{gp}Y_{igp} = W_{gp}Y_{igp}(1)$, the definition of $\overline{y}_{gp}(1)$ and $W_{gp}\widehat{Y}_{igp} = W_{gp}\widehat{Y}(1)$. The third equality follows from the definition of $\widehat{Y}_p(1)$, Point 2 of Assumption 1, and the definition of $\widehat{Y}(1)$.

Likewise,
\begin{align}
    \sum_{i,g} (1-W_{gp})\frac{e_{igp} }{n_{gp}} & = \widehat{Y}_p(0) - \sum_{p'}\frac{n_{p'}}{n}\widehat{Y}_{p'}(0)
    \label{eq:tilde_y_0}
\end{align}

The numerator of $\widehat{\mathbb{V}}_{pair}(\widetilde{\tau})$ equals
\begin{align}
    \sum_p \left[\sum_g V_{gp}\left(W_{gp}-\frac{1}{2}\right)\sum_i e_{igp}\right]^2
    & = \sum_p \left[\sum_g n_{p}\left(W_{gp}-\frac{1}{2}\right)(W_{gp}+1-W_{gp})\sum_i \frac{e_{igp}}{n_{gp}}\right]^2
    \notag
    \\
    & = \sum_p n_p^2\left[\left(1-\frac{1}{2}\right)\sum_{i,g} W_{gp} \frac{e_{igp} }{n_{gp}} -\frac{1}{2}\sum_{i,g} (1-W_{gp}) \frac{e_{igp} }{n_{gp}}\right]^2
    \notag
    \\
    & = \sum_p  \frac{n_p^2}{4} \left[\widehat{Y}_p(1) - \sum_{p'}\frac{n_{p'}}{n}\widehat{Y}_{p'}(1)  - \widehat{Y}_p(0) + \sum_{p'}\frac{n_{p'}}{n}\widehat{Y}_{p'}(0) \right]^2
    \notag
    \\
    & = \sum_p  \frac{n_p^2}{4} \left[\widehat{\tau}_p - \widetilde{\tau}\right]^2.
    \label{eq:num_pcve_tilde_tau}
\end{align}
The second equality follows from the fact that $W_{gp}-\frac{1}{2}=1-\frac{1}{2}$ for the treated units and $W_{gp}-\frac{1}{2}=-\frac{1}{2}$ for the untreated units. The third equality follows from Equations \eqref{eq:tilde_y_1} and \eqref{eq:tilde_y_0}.

The denominator of $\widehat{\mathbb{V}}_{pair}(\widetilde{\tau})$ equals
\begin{align}
    \left[\sum_p \sum_g \sum_i V_{gp}\left(W_{gp}-\frac{1}{2}\right)^2\right]^2
    & = \left[2n \frac{1}{4} \right]^2
    \notag
    \\
    & =  \frac{n^2}{4}.
    \label{eq:den_pcve_tilde_tau}
\end{align}

Then, combining Equations \eqref{eq:pcve_vhat_tilde_tau}, \eqref{eq:num_pcve_tilde_tau} and \eqref{eq:den_pcve_tilde_tau},
\begin{align}
    \widehat{\mathbb{V}}_{pair}(\widetilde{\tau}) & = \sum_p  \frac{n_p^2}{n^2} \left[\widehat{\tau}_p - \widetilde{\tau}\right]^2 =\frac{1}{P^2}\sum_p  \frac{n_p^2}{\overline{n}^2} \left[\widehat{\tau}_p - \widetilde{\tau}\right]^2.
    \label{eq:var_hat_tilde_tau}
\end{align}
\textbf{QED.}

\subsection{Proof of Theorem \ref{th:asym_ext}}

It follows from Lemma \ref{le:plim_tilde_tau_q} that
\begin{equation}
    \widetilde{\tau}-\tau^* = \frac{1}{P}\sum_p \frac{n_p}{\overline{n}}(\widehat{\tau}_p-\EX[\widehat{\tau}_p]) \xrightarrow[]{\mathbb{P}} 0,
    \label{eq:hat_tilde_plim}
\end{equation}
and
\begin{gather}
    \frac{1}{P}\sum_p \left(\frac{n_p}{\overline{n}}\right)^2 [\widehat{\tau}_p^2 - \EX(\widehat{\tau}_p^2)] \xrightarrow[]{\mathbb{P}} 0.
    \label{eq:plim_tilde_tau_squared}
\end{gather}
By a similar argument to the one used in the proof of Lemma \ref{le:plim_tilde_tau_q}, one can also show that
\begin{gather}
    \frac{1}{P}\sum_p \left(\frac{n_p}{\overline{n}}\right)^2 [\widehat{\tau}_p - \EX(\widehat{\tau}_p)] \xrightarrow[]{\mathbb{P}} 0.
    \label{eq:plim_tilde_tau_np_squared}
\end{gather}

We now use Point \ref{asm:3_ext_p3} of Assumption \ref{asm:3_ext} to derive the asymptotic distribution of $(\widetilde{\tau}-\tau^*)/(\widetilde{S}_P/P)$. As $\sum_{p=1}^P \EX\left[|\frac{n_p}{\overline{n}}|^{2+\epsilon}|\widehat{\tau}_p-\EX\left[\widehat{\tau}_p\right]|^{2+\epsilon}/{\widetilde{S}}_P^{2+\epsilon}\right] \rightarrow 0 $ for some $\epsilon>0$ (by Point \ref{asm:3_ext_p3} of Assumption \ref{asm:3_ext}), then, by the Lyapunov central limit theorem, $(\widetilde{\tau}-\tau^*)/(\widetilde{S}_P/P) = \sum_p \frac{n_p}{\overline{n}} (\widehat{\tau}_p-\EX[\widehat{\tau}_p])/\widetilde{S}_P  ~{\overset{d}{\longrightarrow}}~  \mathcal{N}(0,1)$ as $P\rightarrow +\infty$, as $\widetilde{S}_P^2 = P^2 \mathbb{V}(\widetilde{\tau}) = \sum_{p=1}^P \mathbb{V}\left(\frac{n_p}{\overline{n}}\widehat{\tau}_p\right)$.
\\
Therefore,
\begin{equation}\label{eq:as_normality2}
(\widetilde{\tau}-\tau^*)/\sqrt{\mathbb{V}(\widetilde{\tau})}   ~{\overset{d}{\longrightarrow}}~  \mathcal{N}(0,1).
\end{equation}

Then,
   \begin{align}
        &P\widehat{\mathbb{V}}_{pair}(\widetilde{\tau}) - P\mathbb{V}(\widetilde{\tau}) \notag\\
        & =  \frac{1}{P}\sum_p \left(\frac{n_p}{\overline{n}}\right)^2 \left(\widehat{\tau}_p-\widetilde{\tau}\right)^2 - \frac{1}{P}\sum_{p=1}^P \left(\frac{n_p}{\overline{n}}\right)^2 \mathbb{V}(\widehat{\tau}_p)
        \notag
        \\
        & =  \frac{1}{P}\sum_p \left(\frac{n_p}{\overline{n}}\right)^2 \left(\widehat{\tau}_p-\widetilde{\tau}\right)^2 - \frac{1}{P}\sum_{p=1}^P \left(\frac{n_p}{\overline{n}}\right)^2 [\EX(\widehat{\tau}_p^2)-\EX[\widehat{\tau}_p]^2]
        \notag
        \\
        & = \frac{1}{P}\sum_p \left(\frac{n_p}{\overline{n}}\right)^2 \left(\widehat{\tau}_p^2-2\widetilde{\tau}\widehat{\tau}_p+\widetilde{\tau}^2\right) - \frac{1}{P}\sum_{p=1}^P \left(\frac{n_p}{\overline{n}}\right)^2 [\EX(\widehat{\tau}_p^2)-\EX[\widehat{\tau}_p]^2]
        \notag
        \\
        & = \frac{1}{P}\sum_p \left(\frac{n_p}{\overline{n}}\right)^2 (\widehat{\tau}_p^2-\EX[\widehat{\tau}_p^2])-2\widetilde{\tau}\frac{1}{P}\sum_p \left(\frac{n_p}{\overline{n}}\right)^2\widehat{\tau}_p+\widetilde{\tau}^2\frac{1}{P}\sum_p \left(\frac{n_p}{\overline{n}}\right)^2 + \frac{1}{P}\sum_{p=1}^P \left(\frac{n_p}{\overline{n}}\right)^2 \EX[\widehat{\tau}_p]^2
        \notag
        \\
        & ~{\overset{\mathbb{P}}{\longrightarrow}}~
         -2\tau^{\infty} \underset{P\rightarrow +\infty}{\lim}\frac{1}{P}\sum_p \left(\frac{n_p}{\overline{n}}\right)^2
         \EX[\widehat{\tau}_p]
         + \left(\tau^{\infty}\right)^2 \underset{P\rightarrow +\infty}{\lim}\frac{1}{P} \sum_p \left(\frac{n_p}{\overline{n}}\right)^2
         + \underset{P\rightarrow +\infty}{\lim} \frac{1}{P} \sum_p \left(\frac{n_p}{\overline{n}}\right)^2
         \EX[\widehat{\tau}_p]^2
        \notag
        \\
        & =
        \underset{P\rightarrow +\infty}{\lim} \frac{1}{P}\sum_p \left(\frac{n_p}{\overline{n}}\right)^2
         \left[
         \EX[\widehat{\tau}_p]^2
         -2\tau^{\infty}\EX[\widehat{\tau}_p]
         +\left(\tau^{\infty}\right)^2
         \right]
        \notag
        \\
        & = \underset{P\rightarrow +\infty}{\lim}\frac{1}{P}\sum_{p=1}^P \left(\frac{n_p}{\overline{n}}\right)^2 [\EX[\widehat{\tau}_p]-\tau^{\infty}]^2.
        \label{eq:plim_vpair_tilde}
    \end{align}
The first equality follows from Equation \eqref{eq:var_hat_tilde_tau} and the fact that the $(\widehat{\tau}_p)_{p=1}^P$ are independent across $p$ by Point \ref{asm:1_p3} of Assumption \ref{asm:1}. The second equality follows from the definition of variance. The convergence in probability follows from Equations \eqref{eq:hat_tilde_plim} and \eqref{eq:plim_tilde_tau_squared}, \eqref{eq:plim_tilde_tau_np_squared}, and Point \ref{asm:3_ext_p2} of Assumption \ref{asm:3_ext}.
\\
Then,
    \begin{align*}
            \frac{\widetilde{\tau}-\tau^*}{\sqrt{\widehat{\mathbb{V}}_{pair}(\widetilde{\tau})}}
            & = \frac{\widetilde{\tau}-\EX[\widetilde{\tau}]}{\sqrt{\mathbb{V}(\widetilde{\tau})}} \sqrt{\frac{P\mathbb{V}(\widetilde{\tau})}{P\widehat{\mathbb{V}}_{pair}(\widetilde{\tau})}}
            \\
            & ~{\overset{d}{\longrightarrow}}~  \mathcal{N}(0,\sigma_{wls}^2).
    \end{align*}
The convergence in distribution follows from Equation \eqref{eq:plim_vpair_tilde}, Equation \eqref{eq:as_normality2}, Lemma \ref{le:p_vartilde_tau_lim}, the Slutsky Lemma, and the CMT.
    \\
\textbf{QED}.


\subsection{Proof of Lemma \ref{le:exp_PCVE_fe}}

\begin{align*}
\EX\left[\sum_p \widetilde{\omega}_p^2(\widehat{\tau}_p-\widehat{\tau}_{fe})^2\right] & =
\sum_p \widetilde{\omega}_p^2\EX[(\widehat{\tau}_p-\widehat{\tau}_{fe})^2]
\\
& = \sum_p \widetilde{\omega}_p^2 \left[\mathbb{V}(\widehat{\tau}_p-\widehat{\tau}_{fe}) + [\EX(\widehat{\tau}_p-\widehat{\tau}_{fe})]^2 \right]
\\
& = \sum_p \widetilde{\omega}_p^2 \left[\mathbb{V}(\widehat{\tau}_p)+\mathbb{V}(\widehat{\tau}_{fe})-2\text{Cov}(\widehat{\tau}_p,\widehat{\tau}_{fe}) + [\EX(\widehat{\tau}_p-\widehat{\tau}_{fe})]^2 \right]
\\
& = \sum_p \widetilde{\omega}_p^2 \left[\mathbb{V}(\widehat{\tau}_p)+\mathbb{V}(\widehat{\tau}_{fe})-2\omega_p\mathbb{V}(\widehat{\tau}_p) + [\EX(\widehat{\tau}_p-\widehat{\tau}_{fe})]^2 \right]
\\
& = \sum_p \widetilde{\omega}_p^2\left[1-2\omega_p\right] \mathbb{V}(\widehat{\tau}_p) +\mathbb{V}(\widehat{\tau}_{fe})\sum_p \widetilde{\omega}_p^2+ \sum_p \widetilde{\omega}_p^2[\EX(\widehat{\tau}_p-\widehat{\tau}_{fe})]^2
\\
& = \sum_p \omega_p^2 \mathbb{V}(\widehat{\tau}_p) +\mathbb{V}(\widehat{\tau}_{fe})\sum_p \widetilde{\omega}_p^2+ \sum_p \widetilde{\omega}_p^2[\EX(\widehat{\tau}_p-\widehat{\tau}_{fe})]^2
\\
& = \mathbb{V}(\widehat{\tau}_{fe})\left(1+\textstyle\sum_p \widetilde{\omega}_p^2\right)+ \sum_p \widetilde{\omega}_p^2[\EX(\widehat{\tau}_p-\widehat{\tau}_{fe})]^2
\end{align*}
The first equality follows from the linearity of the expectation and the fact that the weights $\omega_p$ are not stochastic. The fourth equality follows from Point \ref{asm:1_p3} of Assumption \ref{asm:1}. The sixth equality follows from the definition of $\widetilde{\omega}_p$. The seventh equality follows from the definition of the variance, the definition of $\widehat{\tau}_{fe}$ and Point \ref{asm:1_p3} of Assumption \ref{asm:1}.

\noindent
\textbf{QED}.


\subsection{Auxiliary Lemmas to prove Theorems \ref{th:asym}, \ref{th:asym_app}, and \ref{th:asym_ext}}

\begin{lemma}
\leavevmode
\label{le:plim_tilde_tau_q}
Let $q\geq 1$, under Points \ref{asm:1_p2} and \ref{asm:1_p3} of Assumption \ref{asm:1}, and Assumption \ref{asm:bal_exp} or Point \ref{asm:1b_p1} of Assumption \ref{asm:1b},
\begin{gather}
    \frac{1}{P}\sum_p \left(\frac{n_p}{\overline{n}}\right)^q [\widehat{\tau}_p^q - \EX(\widehat{\tau}_p^q)] \xrightarrow[]{\mathbb{P}} 0
	\notag
\end{gather}
\end{lemma}

\noindent \textit{Proof}. Assumption \ref{asm:bal_exp} implies Point \ref{asm:1b_p1} of Assumption \ref{asm:1b}, so it is sufficient to show that the result holds under Points \ref{asm:1_p2} and \ref{asm:1_p3} of Assumption \ref{asm:1}, and Point \ref{asm:1b_p1} of Assumption \ref{asm:1b}.

Note that by Point \ref{asm:1_p3} of Assumption \ref{asm:1},  $((\frac{n_p}{\overline{n}}\widehat{\tau}_p)^q-\EX[(\frac{n_p}{\overline{n}}\widehat{\tau}_p)^q])_{p=1}^P$, $q\geq 1$, is a sequence of independent random variables with mean zero.
\\
Note that, for all $p$,
    \begin{align}
        \EX\left[\left|\frac{n_p}{\overline{n}}\widehat{\tau}_p\right|^{q+\epsilon}\right]^{1/(q+\epsilon)}
        & = \frac{n_p}{\overline{n}}\EX\left[\left|\widehat{Y}_p(1)-\widehat{Y}_p(0)\right|^{q+\epsilon}\right]^{1/(q+\epsilon)}
        \notag
        \\
        & \leq N\left(\left(\EX\left[\left|\widehat{Y}_p(1)\right|^{q+\epsilon}\right]\right)^{1/(q+\epsilon)}        + \left(\EX\left[\left|\widehat{Y}_p(0)\right|^{q+\epsilon}\right]\right)^{1/(q+\epsilon)}\right)
        \notag
        \\
        & = N\left(\left(\EX\left[\left|\sum_g W_{gp}\overline{y}_{gp}(1)\right|^{q+\epsilon}\right]\right)^{1/(q+\epsilon)}        + \left(\EX\left[\left|\sum_g (1-W_{gp})\overline{y}_{gp}(0)\right|^{q+\epsilon}\right]\right)^{1/(q+\epsilon)}\right)
        \notag
        \\
        & \leq N\left(\sum_g\left(\EX\left[\left| W_{gp}\overline{y}_{gp}(1)\right|^{q+\epsilon}\right]\right)^{1/(q+\epsilon)}        + \sum_g\left(\EX\left[\left| (1-W_{gp})\overline{y}_{gp}(0)\right|^{q+\epsilon}\right]\right)^{1/(q+\epsilon)}\right)
        \notag
        \\
        & = N\left(\sum_g\left(\EX[W_{gp}]\left| \overline{y}_{gp}(1)\right|^{q+\epsilon}\right)^{1/(q+\epsilon)}        + \sum_g\left(\EX[1-W_{gp}]\left| \overline{y}_{gp}(0)\right|^{q+\epsilon}\right)^{1/(q+\epsilon)}\right)
        \notag
        \\
        & = N\left(\sum_g\left(\frac{1}{2}\left| \overline{y}_{gp}(1)\right|^{q+\epsilon}\right)^{1/(q+\epsilon)}        + \sum_g\left(\frac{1}{2}\left| \overline{y}_{gp}(0)\right|^{q+\epsilon}\right)^{1/(q+\epsilon)}\right)
        \notag
        \\
        & < N\frac{4}{2^{1/(q+\epsilon)}} M<+\infty.
        \label{eq:cond_tau_p} 
    \end{align}
    The first equality follows from the definition of $\widehat{\tau}_p$. The first inequality follows from Minkowski's inequality, and from Point \ref{asm:3_ext_p1} of Assumption \ref{asm:3_ext}. The third line follows from the definitions of $\widehat{Y}_p(1)$ and $\widehat{Y}_p(0)$. The fourth line follows from Minkowski's inequality. The fifth line follows from $W_{gp}$ being a binary variable. The sixth line follows from Point \ref{asm:1_p2} of Assumption \ref{asm:1}. The seventh line follows from Point \ref{asm:1b_p1} of Assumption \ref{asm:1b}.

Using the LLN in Lemma 1 in \cite{liu1988bootstrap}, the previous facts and the fact that almost sure convergence implies convergence in probability, one concludes that
\begin{gather}
    \frac{1}{P}\sum_p \left(\frac{n_p}{\overline{n}}\right)^q [\widehat{\tau}_p^q - \EX(\widehat{\tau}_p^q)] \xrightarrow[]{\mathbb{P}} 0
    \label{eq:plim_tilde_tau_q}.
\end{gather}

\noindent \textbf{QED}.

\begin{lemma}
\label{le:p_vartilde_tau_lim}[Strictly positive limit for $P\mathbb{V}(\widetilde{\tau})$]
\leavevmode
Under Point \ref{asm:1b_p2} of Assumption \ref{asm:1b} and Point \ref{asm:3_ext_p1} of Assumption \ref{asm:3_ext},
$\lim_{P\rightarrow +\infty} P\mathbb{V}(\widetilde{\tau})>0$.
\end{lemma}
\noindent \textit{Proof}. Note that
\begin{align}
    \lim_{P\rightarrow +\infty} P\mathbb{V}(\widetilde{\tau}) & = \lim_{P\rightarrow +\infty} \frac{1}{P} \sum_p \left(\frac{n_p}{\overline{n}}\right)^2 \mathbb{V}(\widehat{\tau}_p)
    \notag
    \\
    & \geq \frac{1}{N^2} \lim_{P\rightarrow +\infty} \frac{1}{P} \sum_p  \mathbb{V}(\widehat{\tau}_p) \notag
    \\
    & = \frac{1}{N^2} \lim_{P\rightarrow +\infty} P \mathbb{V}(\widehat{\tau})
	\notag
    \\
    & > 0.
	\notag
\end{align}
The first equality follows from the definition of $\widetilde{\tau}$ and Point \ref{asm:1_p3} of Assumption \ref{asm:1}. The first inequality follows from the fact that $0< \frac{1}{N} \leq \frac{n_p}{\overline{n}} \leq N$ (which follows from Point \ref{asm:3_ext_p1} of Assumption \ref{asm:3_ext}). The second equality follows from the definition of $\mathbb{V}(\widehat{\tau}) $. The second inequality follows from Point \ref{asm:1b_p2} of Assumption \ref{asm:1b}.

\noindent \textbf{QED}.

\end{document}